\documentclass[11pt,a4paper]{article}  

\usepackage{amssymb}
\usepackage{amsmath}
\usepackage{epsfig}
\usepackage{graphicx}
\usepackage{color}
\usepackage{bigints}

\newcommand{\cancel}[1]    {}
\newcommand{\eps}          {\varepsilon}

\newcommand{\norm}[1]      {\|#1\|} 

\newcommand{\real}         {\mathbb{R}}

\newcommand{\mani}         {{\cal M}}

\newcommand{\mat}[1]       {\mathsf{#1}}

\newcommand{\col}[1]       {\mathrm{col}(#1)}
\DeclareMathOperator*{\argmax}{argmax}
\DeclareMathOperator{\sign}{sign}

\newtheorem{definition}{Definition}
\newtheorem{theorem}{Theorem}[section]
\newtheorem{lemma}{Lemma}[section]

\usepackage{color}

\newbox\ProofSym
\setbox\ProofSym=\hbox{%
\unitlength=0.18ex%
\begin{picture}(10,10)
\put(0,0){\framebox(9,9){}}
\put(0,3){\framebox(6,6){}}
\end{picture}}

\newenvironment{emromani}
  {
   
   \begin{enumerate}}
  {\end{enumerate}}

\begin{document}

\begin{titlepage}
     
\title{Implicit Manifold Reconstruction\thanks{A preliminary version
appears in Proceedings of the 25th Annual ACM-SIAM Symposium on Discrete Algorithms, 2014, 161--173.}}
\author{Siu-Wing Cheng\footnote{Supported by Research Grants Council, Hong Kong, China (project~no.~612109).  Department of Computer Science and Engineering, HKUST, Hong Kong.}
\and 
Man-Kwun Chiu\footnote{Supported by JST ERATO Grant Number JPMJER1201, Japan and ERC StG 757609.  Institut f\"ur Informatik, Freie Universit\"at Berlin, Berlin, Germany.}} 


\date{}

\maketitle

\begin{abstract} 

Let ${\cal M} \subset \mathbb{R}^d$ be a compact, smooth and boundaryless manifold with dimension $m$
and unit reach.  We show how to construct a function $\varphi:
\real^d \rightarrow \real^{d-m}$ from a uniform $(\eps,\kappa)$-sample $P$ of $\cal M$
that offers several guarantees.  Let $Z_\varphi$ denote the zero set of
$\varphi$.  Let $\widehat{{\cal M}}$ denote the set of points at distance
$\eps$ or less from $\cal M$.  There exists $\eps_0 \in (0,1)$ that decreases as $d$ increases such that if $\eps \leq \eps_0$, the following guarantees hold.  First, $Z_\varphi \cap \widehat{\cal M}$ is a faithful approximation of $\cal M$ in the sense that $Z_\varphi \cap \widehat{\cal M}$ is
homeomorphic to $\cal M$, the Hausdorff distance between $Z_\varphi \cap \widehat{\cal
M}$ and $\cal M$ is $O(m^{5/2}\eps^{2})$, and the normal spaces at nearby
points in $Z_\varphi \cap \widehat{\cal M}$ and $\cal M$ make an angle
$O(m^2\sqrt{\kappa\eps})$.  Second, $\varphi$ has local support; in particular,
the value of $\varphi$ at a point is affected only by sample points in $P$ that
lie within a distance of $O(m\eps)$.  Third, we give a projection operator that
only uses sample points in $P$ at distance $O(m\eps)$ from the initial point.
The projection operator maps any initial point near $P$ onto $Z_\varphi \cap
\widehat{\cal M}$ in the limit by repeated applications.  

\end{abstract}

\end{titlepage}

\section{Introduction}

Sensory devices and numerical experiments may generate numerous data points in
$\real^d$ for some large $d$ due to the large number of attributes of the data
that are being monitored.  It is often believed that the data points are
governed by some hidden processes with fewer controlling parameters, and
therefore, the data points may lie in some $m$-dimensional manifold $\cal M$
for some $m \ll d$.  This motivates the study of manifold reconstruction.

In computational geometry, there are several known results that offer provably
faithful reconstructions in the sense that the reconstruction is topologically
equivalent to $\cal M$, the Hausdorff distance between the reconstruction and
$\cal M$ decreases as the sampling density increases, and the angular error
between the tangent spaces at nearby points in the reconstruction and $\cal M$
decreases as the sampling density increases.  These include the weighted cocone
complex by Cheng, Dey and Ramos~\cite{CDR05}, the weighted witness complex by
Boissonnat, Guibas and Oudot~\cite{bgo}, and the tangential Delaunay complex by
Boissonnat and Ghosh~\cite{bg}.  These reconstructions are $m$-dimensional
simplicial complexes with the given sample points as vertices.  
The corresponding reconstruction algorithms have to deal with the challenging
issue of ``sliver removal'' in high dimensions.

Solutions of partial differential equations on manifolds are required in quite
a few areas such as biology~\cite{memoli}, image processing~\cite{turk,witkin},
weathering~\cite{dorsey}, and fluid dynamics~\cite{myers,myersb}.  The
underlying manifold is often specified by a point cloud.  It has been reported~\cite{liang} that local reconstructions of a manifold in the form of zero level sets of local functions are preferred for solving partial differential equations on the manifold.  
Several numerical methods for solving partial differential equations on level sets have been developed~\cite{bertalmio,greer,liang,ruuth}.  

In this paper, we propose an implicit reconstruction for manifolds with
arbitrary codimension in $\real^d$.  Let $\cal M$ be a compact, smooth, and
boundaryless manifold with unit reach.  
Let $P$ be a \emph{uniform
$(\eps,\kappa)$-sample} of $\cal M$, that is, every point in ${\cal M}$ is at distance
$\eps$ or less from some point in $P$ and the number of sample points inside any $d$-ball of radius $\eps$ is at most some constant $\kappa$.  We assume that the following
information is specified in the input: (i) the manifold dimension $m$, (ii) a
\emph{neighborhood radius} $\gamma = 4\eps$, and (iii) approximate tangent spaces at points in $P$
such that the true tangent space at each point in $P$ makes an angle at most $m\gamma$ with the given approximate tangent space at that point.  There are many
algorithms for estimating the manifold dimension
(e.g.~\cite{CC09,CWW,HA,LB,TSL00}).  When the sample points satisfy some local
uniformity condition (e.g., a constant upper bound on the number of sample
points inside any ball of radius $\eps$ centered in $\cal M$), the neighborhood
radius $\gamma$ can be set by measuring the maximum distance from a sample
point to its $k$th nearest neighbor for some appropriate $k$.  If the sample
points are drawn from an independent and identical distribution on $\cal M$, a
recently proposed reach estimator can be used to set $\gamma$~\cite{aamari}.
There are many algorithms for estimating tangent spaces
(e.g.~\cite{BSW,CC16,GW,LMR12,SSM98}), which give an $O(\eps)$ angular error. 

We use the conditions of $\gamma = 4\eps$ and angular error at most $m\gamma$ in order to keep the number of unknown parameters small.  One may worry about satisfying these two conditions simultaneously, but it is not a concern as we explain below.  Suppose that the estimation algorithms return an angular error bound of $c\eps$ for some known constant $c \geq 1$ and a value $\ell$ such that $\eps \leq \ell = O(\eps)$.  We can set $\gamma = \max\{4\ell,c\ell\}$.  Then, the angular error is at most $c\eps \leq c\ell \leq m\gamma$.  Moreover, letting $c' = \max\{\frac{\ell}{\eps},\frac{c\ell}{4\eps}\}$, the input sample can be viewed as a uniform $(\eps',\kappa')$-sample, where $\eps' = c'\eps = \gamma/4$ and $\kappa' = (2c'+1)^d \kappa$, because a packing argument shows that if any $d$-ball of radius $\eps$ contains at most $\kappa$ sample points, then any $d$-ball of radius $c'\eps$ contains at most $(2c'+1)^d \kappa$ sample points.

Our main result is a formula for a function $\varphi: \real^d \rightarrow \real^{d-m}$ using the $(\eps,\kappa)$-sample $P$ and the neighborhood radius $\gamma$ such that
the zero set of $\varphi$ near $\cal M$ forms a reconstruction of $\cal M$.  Let $Z_\varphi$ denote the zero set of $\varphi$.  Let $\widehat{\cal M}$ denote the set of points at distance $\eps$ or less from $\cal M$.  We prove that there exists $\eps_0 \in (0,1)$ that decreases as $d$ increases such that if $\eps \leq \eps_0$, the following guarantees hold.  First, $Z_\varphi
\cap \widehat{\cal M}$ is a faithful approximation of $\cal M$ in the sense that
$Z_\varphi \cap \widehat{\cal M}$ is homeomorphic to $\cal M$, the Hausdorff
distance between $Z_\varphi \cap \widehat{\cal M}$ and $\cal M$ is
$O(m^{5/2}\gamma^{2}) = O(m^{5/2}\eps^2)$, and the normal spaces at nearby points in $Z_\varphi \cap
\widehat{\cal M}$ and $\cal M$ make an angle $O(m^2\sqrt{\kappa\gamma}) = O(m^2\sqrt{\kappa\eps})$.  Second,
$\varphi$ has local support; in particular, the value of $\varphi$ at a point
is affected only by sample points in $P$ that lie within a distance of
$m\gamma$.  Third, we give a projection operator that only uses sample points
in $P$ at distance $m\gamma$ from the initial point.  The projection operator
maps any initial point near $P$ onto $Z_\varphi \cap \widehat{\cal M}$ in the
limit by repeated applications.  


Implicit surfaces in three dimensions have been extensively
studied, particularly in computer graphics and solid modeling
(e.g.~\cite{alexa,carr,hoppe,levin}).  Two functions have been defined
in~\cite{DS,Kol} and shown to give faithful reconstruction of the underlying
surface in three dimensions.  In $\real^d$, a function is defined
in~\cite{bf} and shown to give faithful reconstruction of $(d-1)$-dimensional
manifold.  There seems to be no prior work with provable guarantees on implicit reconstructions of manifolds in $\real^d$ with codimension less than $d-1$.
In the computer graphics community, similar functions have been proposed as projection operators by Adamson and Alexa~\cite{AA} for designing a complex of surface patches connected via vertices and curves in three dimensions.  Each surface patch is the set of stationary points under a projection operator.  For each surface patch, some input points with prescribed tangent spaces are given for defining the corresponding projection operator, but these input points need not form an $\eps$-sample of the resulting surface patch.  It is discussed how to generalize the framework to $\real^d$ for a complex of submanifolds.  However, no mathematical guarantee was provided in~\cite{AA} for $\real^3$ or $\real^d$.

Although the zero set of our function $\varphi$ has a subset near $\cal M$ that is a faithful reconstruction, $\varphi$ should not be confused to be an smooth implicit function as in the Implicit Function Theorem.  If the normal bundle of $\cal M$ is topologically non-trivial, one cannot define a smooth implicit function whose zero set is a faithful reconstruction of $\cal M$.

We provide the definition of our function $\varphi$ in the next section.
Afterwards, we give the proofs of the theoretical guarantees.

\section{Function formulation}
\label{sec:def}

We use lowercase and uppercase letters in {\sf mathsf} font to denote column
vectors and matrices, respectively.  A point is always specified as a column
vector.  Given a matrix $\mat{K}$, we use $\col{\mat{K}}$ to denote the column
space of $\mat{K}$.  We call the unit eigenvectors of a square matrix
corresponding to the $k$ largest (resp. smallest) eigenvalues the $k$
\emph{most dominant} (resp. \emph{least dominant}) unit eigenvectors.  

Recall that $\gamma = 4\eps$ is the
input neighborhood radius.  We will make use of a weight
function $\omega : \real^d \rightarrow \real$ defined as
\begin{eqnarray*}
\omega(\mat{x},\mat{p}) = 
\frac{h(\norm{\mat{x}-\mat{p}})} {\sum_{\mat{q} \in P}
h(\norm{\mat{x}-\mat{q}})},
\end{eqnarray*}
where
\[
h(s)  = 
\left\{\begin{array}{ll}
\displaystyle \left(1-\frac{s}{m\gamma}\right)^{2m}
\left(\frac{2s}{\gamma}+ 1\right), & \mbox{if $s \in [0, m\gamma]$}, \\ [1.5em]
0, & \mbox{if $s > m\gamma$}.
\end{array}\right.
\]
Note that $h$ is differentiable in $(0,\infty)$ and $h'(s) = 0$ for $s \geq
m\gamma$. This weight function is inspired by the Wendland functions~\cite{wendland}.

Since approximate tangent spaces at the sample points are specified in the
input, we can assume that a $d \times m$ matrix $\mat{T}_\mat{p}$ is given for
each $\mat{p} \in P$ such that $\mat{T}_\mat{p}$ has orthogonal unit columns and
$\col{\mat{T}_\mat{p}}$ is the approximate tangent space at $\mat{p}$.
Define the following matrix and vector space for each point
$\mat{x} \in \real^d$:
\begin{center}
\begin{tabular}{lcp{4.5in}}
$\mat{C}_{\mat{x}}$ & = & $\sum_{\mat{p} \, \in P} 
\omega(\mat{x},\mat{p}) \cdot \mat{T}_{\mat{p}}^{} \cdot \mat{T}_{\mat{p}}^t$, \\
$L_\mat{x}$ & = & space spanned by the $(d-m)$ least
dominant unit eigenvectors of $\mat{C}_{\mat{x}}$. \\
\end{tabular}
\end{center}
The $(d-m)$ least dominant unit eigenvectors of $\mat{T}_\mat{p}^{}
\cdot \mat{T}_\mat{p}^t$ span an approximate normal space of $\cal M$ at $\mat{p}$.
So $L_\mat{x}$ is the ``weighted average'' of the approximate normal
spaces at the sample points near $\mat{x}$.

Define a class $\Phi$ of functions $\varrho: \real^d \rightarrow \real^{d-m}$
as follows: \begin{quote} $\displaystyle \Phi = \left\{\varrho \, :\,
\varrho(\mat{x}) = \sum_{\mat{p}\, \in P} \omega(\mat{x},\mat{p}) \cdot
\mat{B}_{\varrho,\mat{x}}^t \cdot (\mat{x}-\mat{p})\right\}$, where
$\mat{B}_{\varrho,\mat{x}}$ is any $d \times (d-m)$ matrix with linearly
independent columns such that $\col{\mat{B}_{\varrho,\mat{x}}} = L_\mat{x}$.
\end{quote} Evaluating $\varrho(\mat{x})$ requires only the sample points at
distance $m\gamma$ or less from $\mat{x}$, and $\omega$ gives more weight to
sample points nearer $\mat{x}$.  Different choices of
$\mat{B}_{\varrho,\mat{x}}$ at each $\mat{x} \in \real^d$ give rise to
different functions in $\Phi$.  A natural choice is a $d \times (d-m)$ matrix
consisting of $d-m$ orthogonal unit vectors that span $L_\mat{x}$.  
We denote the corresponding function in $\Phi$ by $\varphi$ and so
\[
\varphi(\mat{x}) = \sum_{\mat{p}\, \in P}
\omega(\mat{x},\mat{p}) \cdot \mat{B}_{\varphi,\mat{x}}^t \cdot
(\mat{x}-\mat{p}).
\]

We will show that every function in $\Phi$ has the same zero set.  $Z_\varphi$
as a whole is not a good reconstruction of $\cal M$.  Indeed, by definition,
$\varphi(\mat{x}) = 0$ for any $\mat{x} \in \real^d$ at distance $m\gamma$ or
more from $\cal M$.  
We focus on the subset $\widehat{\cal M}$ of $\real^d$ (i.e., the set of points at distance $\eps$ or less from $\mani$).  We
show that $Z_\varphi \cap \widehat{\cal M}$ is a faithful reconstruction of $\cal M$.  

\section{Preliminaries}
\label{sec:prelim}

\subsection{Definitions}

Given a matrix or vector, the corresponding italic lowercase letter with
subscripts denotes an element.  For example, $k_{ij}$ denotes the $(i,j)$ entry
of a matrix $\mat{K}$ and $v_i$ denotes the $i$-th coordinate of a vector
$\mat{v}$.  We use $\mat{I}_j$ to denote a $j \times j$ identity matrix and
$\mat{0}_{i,j}$ an $i \times j$ zero matrix.  The 2-norms of $\mat{v}$ and
$\mat{K}$ are $\norm{\mat{v}} = \left(\sum_i v_i^2 \right)^{1/2}$ and
$\norm{\mat{K}} = \max\left\{\, \norm{\mat{K} \mat{v}} : \norm{\mat{v}} = 1
\,\right\}$.  

We use $B(\mat{x},r)$ to denote the geometric $d$-ball centered at $\mat{x}$
with radius $r$.  We use $\angle (\mat{v},E)$ to denote the angle between a
vector $\mat{v}$ and its projection in an affine subspace $E$.  The angle
$\angle (E,F)$ between two affine subspaces $E$ and $F$, where $\mathrm{dim}(E)
\leq \mathrm{dim}(F)$, is $\max\{ \angle (\mat{v},F) : \mbox{vector $v$ in
$E$}\}$.  

The normal space of $\mani$ at a point $\mat{z}$, denoted $N_\mat{z}$, is the
linear subspace of $\real^d$ that comprises of all vectors normal to $\mani$ at
$\mat{z}$.  Each vector in $N_\mat{z}$ has $d$ coordinates although $N_\mat{z}$
has dimension $d-m$.  The tangent space of $\mani$ at $\mat{z}$, denoted
$T_\mat{z}$, is the orthogonal complement of $N_\mat{z}$.  

The medial axis of $\mani$ is the closure of the set of points in $\real^d$
that have two or more closest points in $\mani$.  The \emph{local feature size}
at a point $\mat{z} \in \mani$ is the distance from $\mat{z}$ to the medial
axis.  We assume that the reach or minimum local feature size of $\cal M$ is 1.  

Let $\nu$ denote the nearest point map.  That is, for every point $\mat{x}$
that does not belong to the medial axis of $\mani$, $\nu(\mat{x})$ is the point
in $\mani$ nearest to $\mat{x}$.

\subsection{Basic results}
\label{sec:basics}

We need the following basic results on $\eps$-sampling theory, 
matrices, and linear subspaces.

\begin{lemma}{\em (\cite{CDR05,GW})}
\label{lem:basic}
\hspace{.2in}
\begin{emromani}
\item For all $\mat{y},\mat{z} \in \mani$, if $\norm{\mat{y}-\mat{z}} \leq
\xi$ for some $\xi < 1$, $\mat{y}$ is at distance $\xi^2/2$ or less
from $\mat{z} + T_\mat{z}$.

\item For all $\mat{y},\mat{z} \in \mani$, if $\norm{\mat{y}-\mat{z}} \leq
\xi$ for a small enough $\xi$, then $\angle (N_\mat{y}, N_\mat{z}) \leq
4\xi$.

\end{emromani}
\end{lemma}

\begin{lemma}
	\label{lem:ball}
	Let $P$ be a uniform $(\eps,\kappa)$-sample of $\mani$.  For any $\mat{x} \in \mathbb{R}^d$ and any $t \in \bigl[1,\frac{1}{\sqrt{2\eps}}\bigr]$, $| P \cap B(\mat{x}, t \eps) | \leq (4t+1)^m\kappa$.
\end{lemma}

\begin{proof}
We first show an upper bound on the minimum number of balls with radii $\eps$ such that their union contains $\mani \cap B(\mat{x}, t \eps)$, which will imply the desired result. 
We pick a maximal set $S$ of points in $\mani \cap B(\mat{x}, t \eps)$ such that any two of them are at distance $\eps$ or more apart. 
It implies that $\mani \cap B(\mat{x}, t \eps) \subseteq \cup_{\mat{z} \in S} B(\mat{z}, \eps)$. 
Otherwise there exists a point $\mat{z} \in \mani \cap B(\mat{x}, t \eps)$ such that the distance between $\mat{z}$ and $S$ is larger than $\eps$, then we can get a larger set by adding $\mat{z}$ to $S$, a contradiction to the definition of $S$.  
Let $S'$ denote the projection of $S$ onto $\mat{x} + T_{\nu(\mat{x})}$. By Lemma~\ref{lem:basic}(i), the distance between any two points in $S'$ is at least $\eps - (t\eps)^2 \geq \eps/2$ when $t \leq \frac{1}{\sqrt{2\eps}}$. Thus, any two balls centered at points in $S'$ with radius $\eps/4$ are interior-disjoint. 
Since the projection of $\mani \cap B(\mat{x}, t \eps)$ into $\mat{x} + T_{\nu(\mat{x})}$ is contained in $(\mat{x} + T_{\nu(\mat{x})}) \cap B(\mat{x}, t \eps)$, $|S'|$ is no more than the size of a maximal packing of interior-disjoint $m$-dimensional balls with radius $\eps/4$ in $(\mat{x} + T_{\nu(\mat{x})}) \cap B(\mat{x}, t \eps + \eps/4)$, which 
is at most the volume of $(\mat{x} + T_{\nu(\mat{x})}) \cap B(\mat{x},t\eps + \eps/4)$ divided by $(\eps/4)^m V_m$, where $V_m$ is the volume of a unit $m$-ball.  Thus, $|S| = |S'| \leq \frac{(t\eps + \eps/4)^m}{(\eps/4)^m} = (4t+1)^m$.  
Then, $| P \cap B(\mat{x}, t \eps) | \leq (4t+1)^m\kappa$ by the definition of uniform $(\eps,\kappa)$-sampling.	
\end{proof}

Partition a square matrix $\mat{K}$ into blocks:
\[
\begin{pmatrix}
\mat{K}_{11} & \cdots & \mat{K}_{1r} \\
\vdots       & \ddots & \vdots \\
\mat{K}_{r1} & \cdots & \mat{K}_{rr}
\end{pmatrix}
\]
The matrices $\mat{K}_{ii}$ are square, but they may have different dimensions.
For $j \not= i$, $\mat{K}_{ij}$ may be square or rectangular.  For any $i,j,k
\in [1,r]$, $\mat{K}_{ik}$ and $\mat{K}_{jk}$ have the same number of columns
and $\mat{K}_{ij}$ and $\mat{K}_{ik}$ have the same number of rows.  Each row
of blocks $\begin{pmatrix}\mat{K}_{i1} & \cdots &
\mat{K}_{ir}\end{pmatrix}$ defines a \emph{generalized gershgorin set} $G_i$ as
follows.  Let $n_i$ be the dimension of $\mat{K}_{ii}$.  
\[
G_i = \left\{\mu \in \real : 
\frac{1}{\norm{(\mat{K}_{ii} -
\mu\mat{I}_{n_i})^{-1}}} \leq \sum_{j \not= i} \norm{\mat{K}_{ij}}\right\}  
\]
It
follows that the numbers in $G_i$ are at least the smallest eigenvalue of
$\mat{K}_{ii}$ minus $\sum_{i \not= j} \norm{\mat{K}_{ij}}$ and at most the
maximum eigenvalue of $\mat{K}_{ii}$ plus $\sum_{i \not= j}
\norm{\mat{K}_{ij}}$.  The eigenvalues of $\mat{K}_{ii}$ are defined to be in
$G_i$ using a continuity argument~\cite{gc}.

\begin{lemma}{\em (\cite{gc})}
\label{lem:gc}
Consider any partition of a square matrix $\mat{K}$ into blocks.
Every
eigenvalue of $\mat{K}$ lies in some generalized gershgorin set $G_i$ with
respect to this partition.  Moreover, if a generalized gershgorin set $G_i$
is disjoint from the union of the other generalized gershgorin sets, then
$G_i$ contains exactly $n_i$ eigenvalues of $\mat{K}$, where $n_i$ is the
dimension of $\mat{K}_{ii}$.
\end{lemma}

\begin{lemma}{\em (\cite{golub})}
\label{lem:matrix}
Let $(\mat{U} \,\,\, \mat{V})$ be a $d \times d$ orthogonal matrix,
where $\mat{U}$ is $d \times r$ and $\mat{V}$ is $d \times (d-r)$.  Let
$\mat{K}$ be a $d \times r$ matrix with orthogonal unit columns.  Then, $\angle
(\col{\mat{U}},\col{\mat{K}}) = \arcsin(\norm{\mat{V}^t \cdot \mat{K}})$.
\end{lemma}

\begin{lemma}{\em (\cite[Lemma~1.1]{EI94})}
\label{lem:slant}
Let $\mat{M}_1$ be an $s \times s$ real symmetric matrix with eigenvalues
$\lambda_1,\ldots,\lambda_s$ in an arbitrary order.  Let $\mat{v}_i$ denote a
unit eigenvector of $\mat{M}_1$ corresponding to $\lambda_i$.  If $\mat{M}_1 +
\mat{M}_2$ is a real symmetric matrix, $\sigma$ is an eigenvalue of $\mat{M}_1
+ \mat{M}_2$, and $\mat{e}$ is a unit eigenvector of $\mat{M}_1 + \mat{M}_2$
corresponding to $\sigma$, then for every $r \in [1,s-1]$, the angle between
$\mat{e}$ and the space spanned by $\{\mat{v}_1,\ldots,\mat{v}_r\}$ is at most
$\arcsin\left(\norm{\mat{M}_2}/\min_{i \in [r+1,s]}
|\lambda_i - \sigma|\right)$.
\end{lemma}

\begin{lemma}
\label{lem:choice}
Let $V$ and $W$ be two linear subspaces of the same dimension $k$ in
$\real^d$ such that $\theta = \angle (V,W) < \pi/2$.
\begin{emromani}

\item For each orthonormal basis $\{\mat{v}_1,\ldots,\mat{v}_k\}$ of $V$,
there exists an orthonormal basis $\{\mat{w}_1,\ldots,\mat{w}_k\}$ of $W$ such
that $\angle (\mat{v}_i,\mat{w}_i) \leq \theta$ for $i \in [1,k]$ and $\angle
(\mat{v}_i, \mat{w}_j-\mat{v}_j) \in
\left[\frac{\pi-\theta}{2},\frac{\pi+\theta}{2}\right]$ for $i,j \in [1,k]$.

\item If $k > d/2$, then there exist orthonormal bases
$\{\mat{v}_1,\ldots,\mat{v}_k\}$ and $\{\mat{w}_1,\ldots,\mat{w}_k\}$ of $V$
and $W$, respectively, such that $\mat{v}_i = \mat{w}_i$ for $i \in [1,2k-d]$,
$\angle (\mat{v}_i,\mat{w}_i) \leq \theta$ for $i \in [1,k]$, and $\angle
(\mat{v}_i, \mat{w}_j-\mat{v}_j) \in
\left[\frac{\pi-\theta}{2},\frac{\pi+\theta}{2}\right]$ for $i,j \in [1,k]$.  Hence,
for any distinct $i$ and $j$, if $i \in [1,2k-d]$ or $j \in [1,2k-d]$, then
$\mat{v}_i \perp \mat{w}_j$.

\end{emromani}
\end{lemma}
\begin{proof}
We make use of principal angles and principal vectors~\cite{bjorck,galantai,miao}.  Pick unit vectors $\mat{a}_1 \in V$ and $\mat{b}_1 \in W$ that minimizes $\angle(\mat{a}_1,\mat{b}_1)$.  For $i \in [2,k]$, pick unit vectors $\mat{a}_i \in V$ and $\mat{b}_i \in W$ that minimizes $\angle(\mat{a}_i,\mat{b}_i)$ subject to $\mat{a}_i \perp \mat{a}_j$ and $\mat{b}_i \perp \mat{b}_j$ for all $j \in [1,i-1]$.  The angles $\angle(\mat{a}_1,\mat{b}_1),\ldots,\angle(\mat{a}_k,\mat{b}_k)$ are called the principal angles.  The vectors $\{\mat{a}_1,\ldots,\mat{a}_k\}$ and $\{\mat{b}_1,\ldots,\mat{b}_k\}$ are called principal vectors.  Note that $\{\mat{a}_1,\ldots,\mat{a}_k\}$ and $\{\mat{b}_1,\ldots,\mat{b}_k\}$ are orthonormal bases of $V$ and $W$, respectively.  The alternative definition of principal angles in~\cite{galantai} implies that for $i \in [1,k]$, $\theta_i \leq \theta = \angle (V,W)$.  It is also known that $\mat{a}_i \perp \mat{b}_j$ for $i \not= j$~\cite{bjorck,galantai}.

Consider (i).  Given an orthonormal basis $\{\mat{v}_1,\ldots,\mat{v}_k\}$ of $V$, for each $i \in [1,k]$, $\mat{v}_i = \sum_{r=1}^k c_{ir}\mat{a}_r$ for some real coefficients $c_{ir}$'s.  Correspondingly, define $\mat{w}_i = \sum_{r=1}^k c_{ir}\mat{b}_r$.  Note that $\norm{\mat{w}_i} = (\sum_{r=1}^k c_{ir}^2)^{1/2} = \norm{\mat{v}_i} = 1$.  Also, for $i \not= j$, $\mat{w}_i^t\mat{w}_j^{} = \sum_{r=1}^k c_{ir}c_{jr} = \mat{v}_i^t\mat{v}_j^{} = 0$.  So $\{\mat{w}_1,\ldots,\mat{w}_k\}$ is an orthonormal basis of $W$.

For $i \in [1,k]$, $\mat{v}_i^t\mat{w}_i^{} = \sum_{r=1}^k c_{ir}^2 \mat{a}_r^t\mat{b}_r^{} \geq \cos\theta$ because $\angle(\mat{a}_r,\mat{b}_r) \leq \theta$ and $\sum_{r=1}^k c_{ir}^2 = \norm{\mat{v}_i} = 1$.  It follows that $\angle(\mat{v}_i,\mat{w}_i) \leq \theta$.  Since $\mat{v}_i$ and $\mat{w}_i$ are unit vectors and $\angle (\mat{v}_i,\mat{w}_i) \leq \theta$, $\mat{v}_i + \mat{w}_i$ is an angle bisector between $\mat{v}_i$ and $\mat{w}_i$. Hence, $\angle (\mat{v}_i,\mat{v}_i + \mat{w}_i) \leq \theta/2$.  It suffices to show that for any $i,j \in [1,k]$, $\mat{v}_i + \mat{w}_i \perp \mat{w}_j - \mat{v}_j$, which then implies that $\left|\frac{\pi}{2} - \angle (\mat{v}_i,\mat{w}_j-\mat{v}_j) \right| \leq \angle (\mat{v}_i,\mat{v}_i + \mat{w}_i) \leq \theta/2$, completing the proof of (i).  To see that $\mat{v}_i + \mat{w}_i \perp \mat{w}_j - \mat{v}_j$, we check $(\mat{v}_i + \mat{w}_i)^t \cdot (\mat{w}_j - \mat{v}_j) = \sum_{r=1}^k (c_{ir}\mat{a}_r + c_{ir}\mat{b}_r)^t \cdot \sum_{r=1}^k (c_{jr}\mat{b}_r - c_{jr}\mat{a}_r)$.  Recall that $\mat{a}_r$ and $\mat{b}_r$ are unit vectors and for $r \not=s$, $\mat{a}_r \perp \mat{a}_s$, $\mat{b}_r \perp \mat{b}_s$, and $\mat{a}_r \perp \mat{b}_s$.  Therefore, $\sum_{r=1}^k (c_{ir}\mat{a}_r + c_{ir}\mat{b}_r)^t \cdot \sum_{r=1}^k (c_{jr}\mat{b}_r - c_{jr}\mat{a}_r) = 0$.

Consider (ii).  Since $k > d/2$, the dimension of $V \cap W$ is at least $2k-d$.  Pick an arbitrary subset $\{\mat{u}_1,\ldots,\mat{u}_{2k-d}\}$ of the orthonormal basis of $V \cap W$.  Set $\mat{v}_i = \mat{w}_i = \mat{u}_i$ for $i \in [1,2k-d]$.  Complete $\{\mat{v}_1,\ldots,\mat{v}_{2k-d}\}$ arbitrarily to an orthonormal basis $\{\mat{v}_1,\ldots,\mat{v}_k\}$ of $V$.  Then, we construct $\mat{w}_j$ as the same way as in (i) for $j \in [2k-d+1,k]$.
\end{proof}

\begin{lemma}
\label{lem:angle}
Let $E_1$ and $E_2$ be two $k$-dimensional linear subspaces.  Let
$\{\mat{u}_1,\ldots,\mat{u}_k\}$ be a basis of $E_1$ consisting 
of unit vectors such that for any distinct $i,j \in [1,k]$, 
$\angle (\mat{u}_i,\mat{u}_j) \in [\pi/2-\phi,\pi/2+\phi]$ for
some $\phi \in \left[0,\arcsin\left(\frac{1}{k}\right)\right)$.  For any 
$\theta \in \left[0,\arcsin\left(\sqrt{\frac{1}{k}-\sin\phi}\right)\right)$,
if $\angle (\mat{u}_i,E_2) \leq \theta$ for all $i \in [1,k]$, then 
$\angle
(E_1,E_2) \leq \arctan\Bigl(\frac{\sqrt{k}\sin\theta}{\sqrt{1-k\sin^2\theta-k\sin\phi}}\Bigr)$.
\end{lemma}
\begin{proof}
Orient space such that $E_2$ is spanned by the first $k$ coordinate axes
of $\real^d$.
Then, for all $i \in [1,k]$, we can write 
\[
\mat{u}_i = \begin{pmatrix}
\mat{v}_i \\
\mat{w}_i
\end{pmatrix},
\]
where $\mat{v}_i$ consists of the first $k$ coordinates and $\mat{w}_i$
consists the remaining $d-k$ coordinates.  Note that
\[
\begin{pmatrix}
\mat{0}_{k,1} \\
\mat{w}_i
\end{pmatrix}
\perp E_2  
\quad \mbox{and} \quad
\begin{pmatrix}
\mat{v}_i \\
\mat{0}_{d-k,1}
\end{pmatrix}
\in E_2.
\]
Since $\angle (\mat{u}_i,E_2) \leq \theta$ by assumption, we have $\norm{\mat{w}_i}
\leq \sin\theta$.  As a result, $\norm{\mat{v}_i} \in [\cos\theta, 1]$.
For any $i \not= j$, we have
\begin{eqnarray*}
& & \begin{pmatrix}
\mat{v}_i^t \,\,\,\, \mat{w}_i^t
\end{pmatrix} 
\cdot 
\begin{pmatrix}
\mat{v}_j \\ \mat{w}_j
\end{pmatrix} 
\in \left[\cos\left(\frac{\pi}{2}+\phi\right),\cos\left(\frac{\pi}{2}-\phi\right)\right] \\
& \Rightarrow & 
\mat{v}_i^t\cdot \mat{v}_j + \mat{w}_i^t\cdot\mat{w}_j \in [-\sin\phi,\sin\phi] \\
& \Rightarrow & |\mat{v}_i^t \cdot \mat{v}_j| \leq
\norm{\mat{w}_i} \cdot \norm{\mat{w}_j} + \sin\phi 
\leq \sin^2\theta + \sin\phi.
\end{eqnarray*}
Let $\mat{n}$ be a vector in $E_1$ that makes the angle $\angle(E_1,E_2)$ with
$E_2$.  By flipping the orientation of any $\mat{u}_i$'s if necessary, we can ensure that
$\mat{n}$ is a convex combination of $\{\mat{u}_1,\ldots,\mat{u}_k\}$, i.e.,
$\mat{n} = \displaystyle \sum_{i=1}^k \lambda_i \begin{pmatrix} \mat{v}_i \\
\mat{w}_i \end{pmatrix}$ for some $\lambda_i$'s in $[0,1]$ such that
$\sum_{i=1}^k \lambda_i = 1$.  Note that flipping the orientation of any $\mat{u}_i$
preserves the angle $\angle (\mat{u}_i,E_2)$ and the fact that 
for any distinct $i,j \in [1,k]$, $\angle (\mat{u}_i,\mat{u}_j) \in [\pi/2-\phi,\pi/2+\phi]$.
Hence,
\begin{eqnarray*}
\angle (E_1,E_2) & = & 
\arctan\left(\frac{\|\sum_{i=1}^k \lambda_i \mat{w}_i\|}
{\|\sum_{i=1}^k \lambda_i \mat{v}_i\|}\right)   
\;\; \leq \;\; 
\arctan\left(\frac{\sum_{i=1}^k \lambda_i\norm{\mat{w}_i}}
{\sqrt{\sum_{i=1}^k\sum_{j=1}^k  \lambda_i\lambda_j \cdot \mat{v}_i^t \cdot \mat{v}_j}} \right)  \\
& \leq & \arctan\left(\frac{\sin\theta} 
{\sqrt{\cos^2\theta \sum_{i=1}^k \lambda_i^2 - (\sin^2\theta+\sin\phi) \sum_{i\not=j}\lambda_i\lambda_j}} \right) \\ 
& = & \arctan\left(\frac{\sin\theta} 
{\sqrt{\sum_{i=1}^k \lambda_i^2 - (\sin^2\theta+\sin\phi) \left(\sum_{i=1}^k\lambda_i\right)^2 }} \right) \\ 
& \leq & \arctan\left(\frac{\sqrt{k}\sin\theta} 
{\sqrt{1 - k\sin^2\theta-k\sin\phi}} \right).
\end{eqnarray*}
The last step uses the fact that $\sum_{i=1}^k\lambda_i^2$ is minimized 
when $\lambda_i = 1/k$ for all $i$.
\end{proof}

\section{Accuracy of $L_\mat{x}$}

The main result of this section is Lemma~\ref{lemma::normal_angle} below: for
every point $\mat{z} \in \mani$ and every point $\mat{x}$ near $\mat{z}$,
$N_\mat{z}$ is approximated by $L_\mat{x}$.  We need the following technical
result.  Recall that $\nu$ is the nearest point map.

\begin{lemma}
\label{lem:C}
Let $\mat{x}$ be a point at distance $2\eps$ or less from $\mani$.  
Assume a coordinate frame such that the columns of
$\begin{pmatrix}
\mat{I}_m \\ \mat{0}_{d-m,m} \end{pmatrix}$ form an
orthonormal basis of $T_{\nu(\mat{x})}$.
Partition $\mat{C}_\mat{x}$ into
$\begin{pmatrix}
\mat{C}_{11} & \mat{C}_{12} \\
\mat{C}_{21} & \mat{C}_{22} \\ 
\end{pmatrix}$, where
$\mat{C}_{11}$ is $m \times m$, $\mat{C}_{12}$ is $m \times (d-m)$,
$\mat{C}_{21}$ is $(d-m) \times m$, and $\mat{C}_{22}$ is $(d-m)
\times (d-m)$.  Then, 
$\norm{\mat{C}_{12}}$ and $\norm{\mat{C}_{21}}$ are $O(m\gamma)$,
$\norm{\mat{C}_{22}}$ is $O(m^2\gamma^2)$, and 
the smallest eigenvalue of $\mat{C}_{11}$ is at least $1 - O(m^2\gamma^2)$. 
\end{lemma}
\begin{proof}
Consider any sample point $\mat{p} \in P$.  Partition $\mat{T}_\mat{p}$ into
$\begin{pmatrix} \mat{Y}_\mat{p} \\ \mat{Z}_\mat{p} \end{pmatrix}$, where
$\mat{Y}_\mat{p}$ is $m \times m$ and $\mat{Z}_\mat{p}$ is $(d-m) \times m$.
For all $\mat{p} \in P \cap B(\mat{x},m\gamma)$,
\[
\norm{\mat{p}-\nu(\mat{x})} 
\leq \norm{\mat{p}-\mat{x}} + \norm{\mat{x}-\nu(\mat{x})} 
\leq m\gamma + 2\eps 
< (m+1)\gamma. 
\]
Then, $\angle (T_\mat{p}, T_{\nu(\mat{x})}) \leq 4(m+1)\gamma$
by Lemma~\ref{lem:basic}(ii).

Since $\begin{pmatrix} \mat{I}_m \\ \mat{0}_{d-m,m} \end{pmatrix}$ and
$\begin{pmatrix} \mat{0}_{m,d-m} \\ \mat{I}_{d-m} \end{pmatrix}$ form a $d
\times d$ orthogonal matrix, we obtain
\begin{eqnarray*}
\arcsin(\norm{\mat{Z}_\mat{p}}) 
& = & \arcsin(\norm{(\mat{0}_{d-m,m} \,\, \mat{I}_{d-m}) \cdot \mat{T}_\mat{p}}) \\
& = & \angle (T_{\nu(\mat{x})}, \col{\mat{T}_\mat{p}})  \quad\quad\quad\quad\quad\quad\quad 
(\because \text{Lemma~\ref{lem:matrix}}) \\ 
& \leq &
\angle (T_\mat{p}, T_{\nu(\mat{x})})  
+ \angle (T_\mat{p},\col{\mat{T}_\mat{p}}) \\
& \leq & 4(m+1)\gamma + m \gamma.
\end{eqnarray*}
(We use the assumption that the input approximate tangent spaces have angular errors at most $m\gamma$.  Although an angular error of $O(m\gamma)$ also works, an exact bound of $m\gamma$ makes explicit the input requirement for constructing the formula of $\varphi$.)  Hence, we have 
\begin{equation}
\forall \, \mat{p} \in P \cap B(\mat{x},m\gamma), \quad
\norm{\mat{Z}_\mat{p}} = O(m\gamma).
\label{eq:Z}
\end{equation}

Because $\omega(\mat{x},\mat{p})$ vanishes for all $\mat{p} \not\in
B(\mat{x},m\gamma)$, $\mat{C}_{12} = \sum_{\mat{p} \in P \cap
B(\mat{x},m\gamma)} \omega(\mat{x},\mat{p}) \cdot \mat{Y}_\mat{p}^{}
\cdot \mat{Z}_\mat{p}^t$. 
Since the columns in $\mat{T}_\mat{p}$ have unit 2-norm, we get
$\norm{\mat{Y}_\mat{p}} \leq 1$.  Thus, 
\[
\norm{\mat{C}_{12}} = 
\left\| \sum_{\mat{p} \in P \cap B(\mat{x},m\gamma)} \omega(\mat{x},\mat{p}) \cdot
\mat{Y}_{\mat{p}}^{} \cdot \mat{Z}_{\mat{p}}^t \right\| 
\leq \sum_{\mat{p}
\in P \cap B(\mat{x},m\gamma)} \omega(\mat{x},\mat{p}) \cdot
\norm{\mat{Z}_\mat{p}} 
= O(m\gamma).
\]
Similarly, 
\[
\norm{\mat{C}_{21}} = 
\left\| \sum_{\mat{p} \in P \cap B(\mat{x},m\gamma)} \omega(\mat{x},\mat{p}) \cdot
\mat{Z}_{\mat{p}}^{} \cdot \mat{Y}_{\mat{p}}^t \right\| 
\leq \sum_{\mat{p}
\in P \cap B(\mat{x},m\gamma)} \omega(\mat{x},\mat{p}) \cdot
\norm{\mat{Z}_\mat{p}} 
= O(m\gamma),
\]
\[
\norm{\mat{C}_{22}} = 
\left\|\sum_{\mat{p} \in P \cap B(\mat{x},m\gamma)} \omega(\mat{x},\mat{p}) \cdot 
\mat{Z}_\mat{p}^{} \cdot \mat{Z}_\mat{p}^t\right\| 
\leq 
\sum_{\mat{p} \in P \cap B(\mat{x},m\gamma)} \omega(\mat{x},\mat{p}) \cdot 
\norm{\mat{Z}_\mat{p}^{}}^2 
= O(m^2\gamma^2).
\]

Since $\mat{T}_\mat{p}^t  \cdot \mat{T}_\mat{p}^{} = \mat{Y}_\mat{p}^t
\cdot \mat{Y}_\mat{p}^{} + \mat{Z}_\mat{p}^t \cdot \mat{Z}_\mat{p}^{}$, the minimum
eigenvalue of $\mat{Y}_\mat{p}^t \cdot \mat{Y}_\mat{p}^{}$ is at least the minimum
eigenvalue of $\mat{T}_\mat{p}^t \cdot \mat{T}_\mat{p}^{}$ minus
$\norm{\mat{Z}_\mat{p}^t \cdot \mat{Z}_\mat{p}^{}}$.  Therefore,
\begin{equation}
\text{minimum eigenvalue of $\mat{Y}_\mat{p}^t \cdot \mat{Y}_\mat{p}^{} \geq 1 - O(m^2\gamma^2)$}.
\label{eq:Y}
\end{equation}
$\mat{Y}_\mat{p}^{} \cdot \mat{Y}_\mat{p}^t$ has the same eigenvalues
as $\mat{Y}_\mat{p}^t \cdot \mat{Y}_\mat{p}^{}$.  The smallest
eigenvalue of a real symmetric matrix $\mat{M}$ is $\min_{\mat{v} \not= \mat{0}}
(\mat{v}^t \cdot \mat{M} \cdot \mat{v})/\norm{\mat{v}}^2$.  Then, using
the relation $\mat{C}_{11} = \sum_{\mat{p} \in P \cap B(\mat{x},m\gamma)}
\omega(\mat{x},\mat{p}) \cdot \mat{Y}_\mat{p}^{} \cdot \mat{Y}_\mat{p}^t$, we
conclude that the smallest eigenvalue of $\mat{C}_{11}$ is at least the sum of the smallest
eigenvalues of $\omega(\mat{x},\mat{p}) \cdot \mat{Y}_\mat{p}^{} \cdot \mat{Y}_\mat{p}^t$.
This sum is at least $1 - O(m^2\gamma^2)$ by \eqref{eq:Y}.
\end{proof}

We are ready to show that the angle between $L_\mat{x}$ and any nearby normal
space of $\mani$ is $O(m\sqrt{m}\gamma)$.

\begin{lemma}
\label{lemma::normal_angle}
For every point $\mat{z} \in \mani$ and every point $\mat{x} \in
B(\mat{z},2\eps)$, $\angle(L_\mat{x},N_\mat{z}) = O(m\sqrt{m}\gamma)$.
\end{lemma}
\begin{proof}
Adopt a coordinate frame such that the columns of $\begin{pmatrix}
\mat{I}_m \\ \mat{0}_{d-m,m} \end{pmatrix}$ form an orthonormal basis
of $T_{\nu(\mat{x})}$.  Let $\mat{A}_\mat{x}$ be the $d
\times m$ matrix whose columns are the $m$ most dominant unit
eigenvectors of $\mat{C}_\mat{x}$.  Thus, $\col{\mat{A}_\mat{x}}$ is the
orthogonal complement of $L_\mat{x}$.  
Let $\mat{e} = \begin{pmatrix} \mat{v} \\ \mat{w} \end{pmatrix}$ be any
column vector of $\mat{A}_\mat{x}$, where $\mat{v}$ consists of the first $m$
coordinates and $\mat{w}$ consists of the last $d-m$ coordinates.
Then, $\angle (\mat{e},T_{\nu(\mat{x})}) =
\arctan(\norm{\mat{w}}/\norm{\mat{v}})$. 

We show that $\angle (\mat{e},T_{\nu(\mat{x})}) = O(m\gamma)$.  
Partition $\mat{C}_\mat{x}$ into
$\begin{pmatrix}
\mat{C}_{11} & \mat{C}_{12} \\
\mat{C}_{21} & \mat{C}_{22} \\ 
\end{pmatrix}$, where
$\mat{C}_{11}$ is $m \times m$, $\mat{C}_{12}$ is $m \times (d-m)$,
$\mat{C}_{21}$ is $(d-m) \times m$, and $\mat{C}_{22}$ is $(d-m)
\times (d-m)$.
Let $\sigma$ be the eigenvalue of $\mat{C}_\mat{x}$ 
corresponding to $\mat{e}$.  Then,
\[
\mat{C}_\mat{x} \, \mat{e}   = 
\begin{pmatrix}
\mat{C}_{11} & \mat{C}_{12} \\
\mat{C}_{21} & \mat{C}_{22} \\ 
\end{pmatrix}
\begin{pmatrix}
\mat{v} \\ 
\mat{w}
\end{pmatrix}
\,\, = \,\, 
\sigma
\begin{pmatrix}
\mat{v} \\
\mat{w}
\end{pmatrix},
\]
which implies that 
\[
\norm{\mat{w}} = \norm{(\sigma\mat{I}_{d-m} -
\mat{C}_{22})^{-1}\mat{C}_{21}\mat{v}} 
\leq \norm{(\sigma\mat{I}_{d-m} -
\mat{C}_{22})^{-1}} \cdot \norm{\mat{C}_{21}}.  
\]
Following the definition of
generalized gershgorin sets (Section~\ref{sec:prelim}), define
\begin{eqnarray*}
G_1 & = & \left\{\mu \in \real :
\frac{1}{\norm{(\mat{C}_{11} - \mu \mat{I}_m)^{-1}}} \leq \norm{\mat{C}_{12}}\right\}, \\
G_2 & = & \left\{\mu \in \real:
\frac{1}{\norm{(\mat{C}_{22} - \mu \mat{I}_{d-m})^{-1}}} \leq \norm{\mat{C}_{21}}\right\}.  
\end{eqnarray*}
The numbers in $G_1$ are
at least the minimum eigenvalue value of $\mat{C}_{11}$ minus
$\norm{\mat{C}_{12}}$, which is at least 
$1 - O(m\gamma + m^2\gamma^2)$ by Lemma~\ref{lem:C}.
The numbers in $G_2$ are at most $\norm{\mat{C}_{22}} + \norm{\mat{C}_{21}} =
O(m\gamma + m^2\gamma^2)$ by Lemma~\ref{lem:C}.  Since every number in $G_1$ is greater than
any number in $G_2$, by Lemma~\ref{lem:gc}, $G_1$ contains the $m$ largest
eigenvalues of $\mat{C}_\mat{x}$.  Thus, $\sigma$ belongs to $G_1$ and $\sigma
\geq 1 - O(m\gamma + m^2\gamma^2)$ which is asymptotically greater than $\norm{\mat{C}_{22}}  = O(m^2\gamma^2)$ (Lemma~\ref{lem:C}).  Therefore, 
\[
\norm{(\sigma\mat{I}_{d-m} - \mat{C}_{22})^{-1}} \leq \frac{1}{1- O(m\gamma + m^2\gamma^2)}. 
\]
By Lemma~\ref{lem:C}, $\norm{\mat{C}_{21}} = O(m\gamma)$, and therefore,
\[
\norm{\mat{w}} \leq
\norm{(\sigma\mat{I}_{d-m}-\mat{C}_{22})^{-1}} \cdot \norm{\mat{C}_{21}} 
\leq \frac{ O(m\gamma)}{1 - O(m\gamma + m^2\gamma^2)} = O(m\gamma).
\]
As a result, $1 \geq \norm{\mat{v}} \geq 1 - \norm{\mat{w}} \geq 1-O(m\gamma)$.   Thus,
$\angle(\mat{e},T_{\nu(\mat{x})}) = \arctan(\norm{\mat{w}}/\norm{\mat{v}}) = O(m\gamma)$.

Since $\mat{e}$ is any column vector of $\mat{A}_\mat{x}$, the angle bound in
the previous paragraph applies to all column vectors of $\mat{A}_\mat{x}$.
We can apply Lemma~\ref{lem:angle} with $E_1 = \col{\mat{A}_\mat{x}}$, $E_2 = T_{\nu(\mat{x})}$, $\{\mat{u}_1,\ldots,\mat{u}_m\}$ equal to the columns of $\mat{A}_\mat{x}$, $\phi = 0$, $k = m$, and $\theta$ equal to the $O(m\gamma)$ bound on $\angle (\mat{e},T_{\nu(\mat{x})})$.  Then, 
\[
\angle (\col{\mat{A}_\mat{x}},T_{\nu(\mat{x})}) \leq
\arctan\left(\frac{O(m\sqrt{m}\gamma)}{\sqrt{1-O(m^3\gamma^2)}}\right) = O(m\sqrt{m}\gamma).
\] 
Since $\norm{\nu(\mat{x}) - \mat{z}} 
\leq \norm{\mat{x} - \nu(\mat{x})} + \norm{\mat{x} - \mat{z}} \leq 4\eps$,
Lemma~\ref{lem:basic}(ii) implies that $\angle (T_{\nu(\mat{x})},T_\mat{z})
\leq 16\eps$.  Hence,
\begin{eqnarray*}
\angle (L_\mat{x},N_\mat{z}) & = & \angle (\col{\mat{A}_\mat{x}},T_\mat{z}) \\
& \leq & \angle (\col{\mat{A}_\mat{x}},T_{\nu(\mat{x})}) + 
\angle (T_{\nu(\mat{x})},T_\mat{z}) \\
& = & O(m\sqrt{m}\gamma).
\end{eqnarray*}
\end{proof}

\section{Projection into $L_\mat{x}$}

For every point $\mat{z} \in \mani$ and every unit vector $\mat{n} \in
N_\mat{z}$, we want to bound the instantaneous change in the normalized
projection of $\mat{n}$ in $L_\mat{x}$ as $\mat{x}$ moves.  If we view the
projection as a map $f$, this is equivalent to analyzing the Jacobian of $f$
which is given in Lemmas~\ref{lemma::gradient} and~\ref{cor:gradient} below.
To this end, some technical results are needed.  First, we need to study the
variation of $\mat{C}_\mat{x}$ as $\mat{x}$ moves (Lemma~\ref{lem:delta_C}).
Second, we need to bound the turn of $L_\mat{x}$ if $\mat{x}$ moves slightly
(Lemma~\ref{lemma::normal_change}).

Let $\delta_k > 0$ denote an arbitrarily small change in the coordinate $x_k$
of $\mat{x}$.  Define 
\[
\Delta h(\norm{\mat{x}-\mat{p}}) = \frac{\partial
h(\norm{\mat{x}-\mat{p}})}{\partial x_k} \cdot \delta_k.  
\]
For simplicity,
we omit the dependence of $\Delta h(\norm{\mat{x}-\mat{p}})$ on $k$ in the
notation.

\begin{lemma}
\label{lem:delta_C}
Let $\mat{x}$ be a point at distance $2\eps$ or less
from $\mani$.  Assume a coordinate frame such that
the columns of $\begin{pmatrix} \mat{I}_m \\ \mat{0}_{d-m,m} 
\end{pmatrix}$ form an orthonormal basis of $T_{\nu(\mat{x})}$.
Define the $d \times d$ matrix $\Delta\mat{C}_\mat{x} = 
\begin{pmatrix}
\dfrac{\partial c_{ij}}{\partial x_k} \cdot \delta_k
\end{pmatrix}$,
where $c_{ij}$ is the $(i,j)$ entry of $\mat{C}_\mat{x}$.  The following
properties hold when $\delta_k$ is small enough.
\begin{emromani}

\item $\displaystyle \norm{\Delta\mat{C}_\mat{x}} \leq \frac{O(m\gamma) \cdot
\sum_{\mat{p} \in P} |\Delta h(\norm{\mat{x}-\mat{p}})|}{\sum_{\mat{p} \in
P} h(\norm{\mat{x}-\mat{p}})}$.

\item The $m$ largest eigenvalues of $\mat{C}_\mat{x} + \Delta\mat{C}_\mat{x}$
are at least $\displaystyle 1 - O(m\gamma) - \frac{O(m\gamma) \cdot \sum_{\mat{p}
\in P} |\Delta h(\norm{\mat{x}-\mat{p}})|}{\sum_{\mat{p} \in P}
h(\norm{\mat{x}-\mat{p}})}$.

\end{emromani}
\end{lemma}

\begin{proof}
Using standard calculus, we obtain
\[
\Delta \mat{C}_\mat{x} = 
\frac{1}{(\sum_{\mat{p} \in P} h(\norm{\mat{x}-\mat{p}}))^2} 
\left(\sum_{\mat{p},\mat{q} \in P}
h(\norm{\mat{x}-\mat{p}}) \cdot 
\Delta h(\norm{\mat{x}-\mat{q}})
\cdot \left( \mat{T}_\mat{q}^{} \cdot \mat{T}_\mat{q}^t 
- \mat{T}_\mat{p}^{} \cdot \mat{T}_\mat{p}^t \right) \right).
\]
Partition $\mat{C}_\mat{x}$ and $\Delta\mat{C}_\mat{x}$ as follows:
\[
\mat{C}_\mat{x} = \begin{pmatrix}
\mat{C}_{11} & \mat{C}_{12} \\
\mat{C}_{21} & \mat{C}_{22} \\ 
\end{pmatrix},
\quad\quad\quad
\Delta\mat{C}_\mat{x} = 
\begin{pmatrix}
\Delta\mat{C}_{11} & \Delta\mat{C}_{12} \\
\Delta\mat{C}_{21} & \Delta\mat{C}_{22} \\ 
\end{pmatrix}
\]
where $\mat{C}_{11}$ and $\Delta\mat{C}_{11}$ are $m \times m$, $\mat{C}_{12}$
and $\Delta\mat{C}_{12}$ are $m \times (d-m)$, $\mat{C}_{21}$ and
$\Delta\mat{C}_{21}$ are $(d-m) \times m$, and $\mat{C}_{22}$ and
$\Delta\mat{C}_{22}$ are $(d-m) \times (d-m)$.  

For every sample point $\mat{p} \in P$, partition $\mat{T}_\mat{p}$ into
$\mat{T}_\mat{p} = \begin{pmatrix} \mat{Y}_\mat{p} \\
\mat{Z}_\mat{p} \end{pmatrix}$,
where $\mat{Y}_\mat{p}$ is an $m \times m$ matrix and
$\mat{Z}_\mat{p}$ is a $(d-m) \times m$ matrix.  

By \eqref{eq:Z} and \eqref{eq:Y}, for every sample point $\mat{p} \in P \cap
B(\mat{x},m\gamma)$, $\norm{\mat{Z}_\mat{p}} = O(m\gamma)$ and the eigenvalues of
$\mat{Y}_\mat{p}^{} \cdot \mat{Y}_\mat{p}^t$ are at least $1 - O(m^2\gamma^2)$.
Moreover, 
\[
\norm{\mat{Y}_\mat{p}^{}\cdot\mat{Y}_\mat{p}^t} =
\norm{\mat{Y}_\mat{p}^t\cdot\mat{Y}_\mat{p}^{}} =
\norm{\mat{T}_\mat{p}^t\cdot\mat{T}_\mat{p}^{} -
\mat{Z}_\mat{p}^t\cdot\mat{Z}_\mat{p}^{}} \leq \norm{\mat{T}_\mat{p}}^2 +
\norm{\mat{Z}_\mat{p}}^2 = 1 + O(m^2\gamma^2),
\]
which also implies that
\[
\norm{\mat{Y}_\mat{p}} = 1 + O(m^2\gamma^2).
\]
Because for
any real symmetric matrix $\mat{M}$, $\norm{\mat{M}} = \max_{\mat{v} \not=
\mat{0}} \left(\mat{v}^t \cdot \mat{M} \cdot \mat{v}\right)/\norm{\mat{v}}^2$,
we conclude that 
$\norm{\mat{Y}_\mat{q}^{} \cdot \mat{Y}_\mat{q}^t -
\mat{Y}_\mat{p}^{} \cdot \mat{Y}_\mat{p}^t}$ is at most the maximum
eigenvalue of 
$\mat{Y}_\mat{q}^{} \cdot \mat{Y}_\mat{q}^t$ minus the minimum
eigenvalue of 
$\mat{Y}_\mat{p}^{} \cdot \mat{Y}_\mat{p}^t$.  Therefore,
\[
\begin{array}{lll}
\norm{\mat{Y}_\mat{q}^{} \cdot \mat{Y}_\mat{q}^t -
\mat{Y}_\mat{p}^{} \cdot \mat{Y}_\mat{p}^t} & \leq & 1 
+ O(m^2\gamma^2) - (1 - O(m^2\gamma^2)) 
\,\, = \,\, O(m^2\gamma^2).
\end{array}
\]
Moreover,
\[
\begin{array}{lll}
\norm{\mat{Y}_\mat{q}^{} \cdot \mat{Z}_\mat{q}^t -
\mat{Y}_\mat{p}^{} \cdot \mat{Z}_\mat{p}^t} & \leq & 
\norm{\mat{Y}_\mat{q}} \cdot \norm{\mat{Z}_\mat{q}} +
\norm{\mat{Y}_\mat{p}} \cdot \norm{\mat{Z}_\mat{p}} \,\, = \,\, O(m\gamma), \\
\norm{\mat{Z}_\mat{q}^{} \cdot \mat{Z}_\mat{q}^t -
\mat{Z}_\mat{p}^{} \cdot \mat{Z}_\mat{p}^t} & \leq & \norm{\mat{Z}_\mat{q}}^2 +
\norm{\mat{Z}_\mat{p}}^2 \,\, = \,\, O(m^2\gamma^2).  
\end{array}
\]
On the other hand, for every sample point $\mat{p} \in P \setminus B(\mat{x},m\gamma)$, 
\[
h(\norm{\mat{x}-\mat{p}}) = 0, \quad\quad\quad \Delta h(\norm{\mat{x}-\mat{p}}) = 0.
\]
Consequently,
\begin{eqnarray}
\norm{\Delta\mat{C}_{11}}
& = &
\frac{\left\|\sum_{\mat{p},\mat{q} \in P} h(\norm{\mat{x}-\mat{p}}) \cdot
\Delta h(\norm{\mat{x}-\mat{q}}) \cdot
(\mat{Y}_\mat{q}^{} \cdot \mat{Y}_\mat{q}^t -
\mat{Y}_\mat{p}^{} \cdot \mat{Y}_\mat{p}^t)\right\|} {(\sum_{\mat{p} \in P}
h(\norm{\mat{x}-\mat{p}}))^2}  \nonumber \\
& \leq &
\frac{\sum_{\mat{p},\mat{q} \in P} h(\norm{\mat{x}-\mat{p}}) \cdot
|\Delta h(\norm{\mat{x}-\mat{q}})| \cdot
\norm{\mat{Y}_\mat{q}^{} \cdot \mat{Y}_\mat{q}^t -
\mat{Y}_\mat{p}^{} \cdot \mat{Y}_\mat{p}^t}} {(\sum_{\mat{p} \in P}
h(\norm{\mat{x}-\mat{p}}))^2}  \nonumber \\
& = & \frac{ O(m^2\gamma^2) \cdot \sum_{\mat{p}
\in P} |\Delta h(\norm{\mat{x}-\mat{p}})|} {\sum_{\mat{p} \in P}
h(\norm{\mat{x}-\mat{p}})}. \label{eq:C11}
\end{eqnarray}
By symmetry,
\begin{eqnarray}
\norm{\Delta\mat{C}_{12}} = \norm{\Delta\mat{C}_{21}} 
& = & \frac{\left\|\sum_{\mat{p},\mat{q} \in P}
h(\norm{\mat{x}-\mat{p}}) \cdot \Delta h(\norm{\mat{x}-\mat{q}})
\cdot (\mat{Y}_\mat{q}^{} \cdot \mat{Z}_\mat{q}^t -
\mat{Y}_\mat{p}^{} \cdot \mat{Z}_\mat{p}^t)\right\|} {(\sum_{\mat{p} \in P}
h(\norm{\mat{x}-\mat{p}}))^2} \nonumber \\
& \leq & \frac{ \sum_{\mat{p},\mat{q} \in P}
h(\norm{\mat{x}-\mat{p}}) \cdot |\Delta h(\norm{\mat{x}-\mat{q}})| 
\cdot \norm{\mat{Y}_\mat{q}^{} \cdot \mat{Z}_\mat{q}^t -
\mat{Y}_\mat{p}^{} \cdot \mat{Z}_\mat{p}^t}} {(\sum_{\mat{p} \in P}
h(\norm{\mat{x}-\mat{p}}))^2} \nonumber \\
& = & \frac{ O(m\gamma) \cdot \sum_{\mat{p} \in P}
|\Delta h(\norm{\mat{x}-\mat{p}})|} {\sum_{\mat{p} \in P}
h(\norm{\mat{x}-\mat{p}})}, \label{eq:C12}
\end{eqnarray}
Similarly,
\begin{eqnarray}
\norm{\Delta\mat{C}_{22}}
& = &
\frac{\left\| \sum_{\mat{p},\mat{q} \in P} h(\norm{\mat{x}-\mat{p}}) \cdot
\Delta h(\norm{\mat{x}-\mat{q}}) \cdot
(\mat{Z}_\mat{q}^{} \cdot \mat{Z}_\mat{q}^t -
\mat{Z}_\mat{p}^{} \cdot \mat{Z}_\mat{p}^t)\right\|} {(\sum_{\mat{p} \in P}
h(\norm{\mat{x}-\mat{p}}))^2} \nonumber \\
& \leq &
\frac{ \sum_{\mat{p},\mat{q} \in P} h(\norm{\mat{x}-\mat{p}}) \cdot
|\Delta h(\norm{\mat{x}-\mat{q}})| \cdot
\norm{\mat{Z}_\mat{q}^{} \cdot \mat{Z}_\mat{q}^t -
\mat{Z}_\mat{p}^{} \cdot \mat{Z}_\mat{p}^t}} {(\sum_{\mat{p} \in P}
h(\norm{\mat{x}-\mat{p}}))^2} \nonumber \\
& = & \frac{O(m^2\gamma^2) \cdot \sum_{\mat{p} \in
P} |\Delta h(\norm{\mat{x}-\mat{p}})|} {\sum_{\mat{p} \in P}
h(\norm{\mat{x}-\mat{p}})}. \label{eq:C22} 
\end{eqnarray}

From the discussion of generalized gershgorin sets (Section~\ref{sec:basics}), we have
\begin{eqnarray}
\norm{\Delta\mat{C}_\mat{x}} & \leq &
\max\{\norm{\Delta\mat{C}_{11}} +
\norm{\Delta\mat{C}_{12}}\,,\, \norm{\Delta\mat{C}_{21}} +
\norm{\Delta\mat{C}_{22}}\}.  \label{eq:Cx}
\end{eqnarray}
The correctness of (i) is then proved by plugging into \eqref{eq:Cx} the
inequalities \eqref{eq:C11}, \eqref{eq:C12}, and \eqref{eq:C22}.

Define the following generalized gershgorin sets:
\begin{eqnarray*}
G_1 & = &
\left\{\mu : \frac{1}{\norm{(\mat{C}_{11} + \Delta\mat{C}_{11} -
\mu \mat{I}_m)^{-1}}} \leq \norm{\mat{C}_{12} + \Delta\mat{C}_{12}}\right\}, \\
G_2 & = & \left\{\mu : \frac{1}{\norm{(\mat{C}_{22} +
\Delta\mat{C}_{22}- \mu \mat{I}_{d-m})^{-1}}} \leq
\norm{\mat{C}_{21}+\Delta\mat{C}_{21}}\right\}.  
\end{eqnarray*}
We give a lower bound for the values in $G_1$ and an upper bound for 
the values in $G_2$.

Consider $G_1$.  The minimum eigenvalue of $\mat{C}_{11}+\Delta\mat{C}_{11}$ is at least the
minimum eigenvalue of $\mat{C}_{11}$ minus
$\norm{\Delta\mat{C}_{11}}$.  Therefore, by Lemma~\ref{lem:C} and \eqref{eq:C11},
\[
\text{minimum eigenvalue of $\mat{C}_{11}+\Delta\mat{C}_{11}
\geq 1 - O(m^2\gamma^2) - 
\frac{O(m^2\gamma^2) \cdot \sum_{\mat{p} \in P} 
|\Delta h(\norm{\mat{x}-\mat{p}})|}{\sum_{\mat{p} \in P}
h(\norm{\mat{x}-\mat{p}})}$}.
\]
On the other hand, by Lemma~\ref{lem:C} and \eqref{eq:C12},
\begin{eqnarray}
\norm{\mat{C}_{12} + \Delta\mat{C}_{12}}
& \leq & \norm{\mat{C}_{12}} + \norm{\Delta\mat{C}_{12}}  \nonumber \\
& \leq &
O(m\gamma) + \frac{O(m\gamma) \cdot \sum_{\mat{p} \in P} 
|\Delta h(\norm{\mat{x}-\mat{p}})|}{\sum_{\mat{p} \in P}
h(\norm{\mat{x}-\mat{p}})}.  \label{eq:C1212}
\end{eqnarray}
The values in $G_1$ are at least the minimum eigenvalue value of $\mat{C}_{11}
+ \Delta\mat{C}_{11}$ minus $\norm{\mat{C}_{12} + \Delta\mat{C}_{12}}$. 
Therefore,
\begin{equation}
\min\{\mu : \mu \in G_1\} \geq
1 - O(m\gamma) -
\frac{O(m\gamma) \cdot \sum_{\mat{p} \in P} 
|\Delta h(\norm{\mat{x}-\mat{p}})|}{\sum_{\mat{p} \in P}
h(\norm{\mat{x}-\mat{p}})}.
\label{eq:G1}
\end{equation}

Consider $G_2$.  By Lemma~\ref{lem:C} and \eqref{eq:C22},
\begin{eqnarray*}
\norm{\mat{C}_{22} + \Delta\mat{C}_{22}}
& \leq & \norm{\mat{C}_{22}} + \norm{\Delta\mat{C}_{22}}  \\
& \leq & O(m^2\gamma^2) + 
\frac{O(m^2\gamma^2) \cdot \sum_{\mat{p} \in P}
|\Delta h(\norm{\mat{x}-\mat{p}})|}{\sum_{\mat{p} \in P}
h(\norm{\mat{x}-\mat{p}})}.
\end{eqnarray*}
By symmetry and \eqref{eq:C1212},
\begin{eqnarray*}
\norm{\mat{C}_{21} + \Delta\mat{C}_{21}} 
= \norm{\mat{C}_{12} + \Delta\mat{C}_{12}} 
& \leq & O(m\gamma) + \frac{O(m\gamma) \cdot \sum_{\mat{p} \in P}
|\Delta h(\norm{\mat{x}-\mat{p}})|}{\sum_{\mat{p} \in P}
h(\norm{\mat{x}-\mat{p}})}.  
\end{eqnarray*}
The values in $G_2$ are at most
$\norm{\mat{C}_{22} + \Delta\mat{C}_{22}} + \norm{\mat{C}_{21} +
\Delta\mat{C}_{21}}$.  Therefore,
\begin{equation}
\max\{\mu : \mu \in G_2\} = O(m\gamma) + \frac{O(m\gamma) \cdot \sum_{\mat{p}
\in P} |\Delta h(\norm{\mat{x}-\mat{p}})|}{\sum_{\mat{p} \in P}
h(\norm{\mat{x}-\mat{p}})}.
\label{eq:G2}
\end{equation}

It follows from \eqref{eq:G1} and \eqref{eq:G2} that $G_1$ and $G_2$ are disjoint
because every number in $G_2$ is much smaller than those in $G_1$.
Lemma~\ref{lem:gc} implies that $G_1$ contains the $m$ largest eigenvalues of
$\mat{C}_\mat{x} + \Delta\mat{C}_\mat{x}$.  The 
correctness of (ii) then follows from \eqref{eq:G1}.
\end{proof}

We need another technical result on bounding $|\Delta
h(\norm{\mat{x}-\mat{p}})|$ from above and $h(\norm{\mat{x}-\mat{q}})$ from
below, where $\mat{q}$ is the nearest sample point to $\nu(\mat{x})$.

\begin{lemma}\label{lemma::in_out_ratio}
Let $\mat{x}$ be any point at distance $2\eps$ or less from $\mani$.
\begin{emromani}

\item For all $\mat{p} \in P$, $\displaystyle |\Delta
h(\norm{\mat{x}-\mat{p}})| \leq \left(1 -
\frac{\norm{\mat{x}-\mat{p}}}{m\gamma}\right)^{2m-1} \cdot
O\left(\frac{m\delta_k}{\gamma}\right)$.

\item $h(\norm{\mat{x}-\mat{q}}) > 0.06$, where $\mat{q}$ is the nearest sample
point to $\nu(\mat{x})$.

\end{emromani}
\end{lemma} 
\begin{proof}
Consider (i).
Since $\Delta h(\norm{\mat{x}-\mat{p}}) = 0$ for any $\mat{p} \in P \setminus
B(\mat{x},m\gamma)$, we only need to consider the case 
of $\norm{\mat{x}-\mat{p}} \leq m\gamma$.  Taking derivative gives
\begin{eqnarray*}
|\Delta h(\norm{\mat{x}-\mat{p}})| 
& \leq &  
2m\left(1 - \frac{\norm{\mat{x}-\mat{p}}}{m\gamma}\right)^{2m-1}
\left(\frac{2\norm{\mat{x}-\mat{p}}}{\gamma} + 1\right) \cdot 
\frac{|x_k-p_k|}{m\gamma\norm{\mat{x}-\mat{p}}} \cdot \delta_k + \\ 
& & \left(1 - \frac{\norm{\mat{x}-\mat{p}}}{m\gamma}\right)^{2m} \cdot
\frac{2|x_k-p_k|}{\gamma\norm{\mat{x}-\mat{p}}} \cdot \delta_k \\
& \leq & \left(1 -
\frac{\norm{\mat{x}-\mat{p}}}{m\gamma}\right)^{2m-1} \cdot 
O\left(\frac{m\delta_k}{\gamma}\right),   
\end{eqnarray*}
establishing the correctness of (i).

Consider (ii).
As $P$ is a uniform $(\eps,\kappa)$-sample, $\norm{\mat{q}-\nu(\mat{x})} \leq \eps$.
Therefore, $\norm{\mat{q}-\mat{x}} \leq \norm{\mat{x}-\nu(\mat{x})} +
\norm{\nu(\mat{x})-\mat{q}} \leq 3\eps$.  Then, 
\begin{eqnarray*}
h(\norm{\mat{x}-\mat{q}})
& = & \left(1-\frac{\norm{\mat{x}-\mat{q}}}{m\gamma}\right)^{2m}
\left(\frac{2\norm{\mat{x}-\mat{q}}}{\gamma}+1\right) \\
& \geq &
\left(1-\frac{3\eps}{m\gamma}\right)^{2m} \\
& = & \left(1 - \frac{3}{4m}\right)^{2m}.
\end{eqnarray*}
The minimum of $\left(1- \frac{3}{4m}\right)^{2m}$ is achieved at $m=1$, and it is equal to 0.0625.
\end{proof}

The following lemma allows us to ignore the contribution of the points near the boundary of $B(\mat{x}, m\gamma)$ in $\frac{\sum_{\mat{p} \in P \cap B(\mat{x}, m\gamma)} |\Delta h(\norm{\mat{x}-\mat{p}})|}{\sum_{\mat{p} \in P \cap B(\mat{x}, m\gamma)} h(\norm{\mat{x}-\mat{p}})}$.

\begin{lemma}
	\label{lem:center}
	Let $\mat{x}$ be any point at distance $2\eps$ or less from $\mani$.  Let $P$ be a uniform $(\eps,\kappa)$-sample of $\mani$.  Let $r = \sqrt{m} \eps/3$.  Then,
	\[
	\sum_{\mat{p} \in P \cap B(\mat{x}, m\gamma)} \left(1-\frac{\norm{\mat{x}-\mat{p}}}{m \gamma}\right)^{2m-1} \leq (23\kappa+1) \cdot \sum_{\mat{p} \in P \cap B(\mat{x},m\gamma - r)} \left(1-\frac{\norm{\mat{x}-\mat{p}}}{m \gamma}\right)^{2m-1}.
	\]
\end{lemma}
\begin{proof}
Observe that
\begin{eqnarray*}
\sum_{\mat{p} \in P \cap B(\mat{x}, m\gamma) } \left(1-\frac{\norm{\mat{x}-\mat{p}}}{m \gamma}\right)^{2m-1}  & = & \sum_{\mat{p} \in P \cap B(\mat{x}, m\gamma-r) } \left(1-\frac{\norm{\mat{x}-\mat{p}}}{m \gamma}\right)^{2m-1} + \\
& & \sum_{\mat{p} \in P \cap B(\mat{x}, m\gamma) \setminus B(\mat{x}, m\gamma-r) } \left(1-\frac{\norm{\mat{x}-\mat{p}}}{m \gamma}\right)^{2m-1}.
\end{eqnarray*}
We prove the lemma by bounding the two terms on the right hand side above.

We show a lower bound for the first term. As $P$ is a uniform $(\eps,\kappa)$-sample, there exists some point $\mat{q} \in P$ such that $\norm{\mat{q}-\nu(\mat{x})} \leq \eps$. Therefore, $\norm{\mat{q}-\mat{x}} \leq \norm{\mat{x}-\nu(\mat{x})} +
\norm{\nu(\mat{x})-\mat{q}} \leq 3\eps \leq m\gamma - r$. Then, 
\begin{eqnarray*}
	\sum_{\mat{p} \in P \cap B(\mat{x}, m\gamma-r) } \left(1-\frac{\norm{\mat{x}-\mat{p}}}{m \gamma}\right)^{2m-1} 
	& \geq  & \left(1-\frac{\norm{\mat{x}-\mat{q}}}{m \gamma}\right)^{2m-1} \\
	& \geq & \left(1 - \frac{3\eps}{m\gamma}\right)^{2m-1} \\
	& \geq & \left(1 - \frac{3}{4m}\right)^{2m}.
\end{eqnarray*}
The quantity $\left(1 - \frac{3}{4m}\right)^{2m}$ achieves its minimum of 1/16 when $m=1$. Hence, 
\[
\sum_{\mat{p} \in P \cap B(\mat{x}, m\gamma-r) } \left(1-\frac{\norm{\mat{x}-\mat{p}}}{m \gamma}\right)^{2m-1}\geq \frac{1}{16}.
\]

We show an upper bound for the second term.  For any point $\mat{p} \in B(\mat{x}, m\gamma) \setminus B(\mat{x}, m\gamma-r)$, $\left(1-\frac{\norm{\mat{x}-\mat{p}}}{m \gamma}\right)^{2m-1}$ achieves its maximum of $\bigl(\frac{r}{m \gamma}\bigr)^{2m-1} = \bigl(\frac{1}{12\sqrt{m}}\bigr)^{2m-1}$ when $\norm{\mat{x}-\mat{p}} = m\gamma-r$. By Lemma~\ref{lem:ball}, $|P \cap B(\mat{x}, m\gamma) \setminus B(\mat{x}, m\gamma-r)| \leq |P \cap B(\mat{x},m\gamma)| \leq (4m\gamma/\eps + 1)^m \kappa$.
Therefore, 
\begin{eqnarray*}
\sum_{\mat{p} \in P \cap B(\mat{x}, m\gamma) \setminus B(\mat{x}, m\gamma-r) } \left(1-\frac{\norm{\mat{x}-\mat{p}}}{m \gamma}\right)^{2m-1} 
& \leq & 
(16m+1)^m\kappa  \left(\frac{1}{12\sqrt{m}}\right)^{2m-1} \\
& \leq & (17)^m \kappa \sqrt{m}/ 12^{2m-1} \\
& \leq & 17 \kappa/12.
\end{eqnarray*}
%
Therefore, the second term is at most the first term multiplied by $23\kappa$. 
%
\end{proof}

We bound the turn of $L_\mat{x}$ when $\mat{x}$ moves slightly in the next
result. 

\begin{lemma} 
\label{lemma::normal_change}
For every point $\mat{x}$ at distance at most $2\eps$ from $\mani$ and for every
vector $\Delta \mat{x} \in N_{\nu(\mat{x})} \cup T_{\nu(\mat{x})}$, if
$\norm{\Delta\mat{x}}$ is small enough and $\mat{x} + \Delta\mat{x}$ is at
distance $2\eps$ or less from $\mani$, then $\angle (L_\mat{x}, L_{\mat{x} +
\Delta\mat{x}}) = O(\kappa m^2 \,\norm{\Delta\mat{x}})$.
\end{lemma}
\begin{proof}
Adopt a coordinate frame such that the columns of $\begin{pmatrix}
\mat{I}_m \\ \mat{0}_{d-m,m} \end{pmatrix}$ form an orthonormal basis
of $T_{\nu(\mat{x})}$,
and $\Delta\mat{x}$ points in the direction of the $x_k$-axis for some $k \in [1,d]$.  Let 
$\delta_k = \norm{\Delta\mat{x}}$.

Every entry of $\mat{C}_{\mat{x}+\Delta\mat{x}}$ is
some algebraic function in $\delta_k$.  By Taylor's Theorem,
the $(i,j)$ entry of $\mat{C}_{\mat{x}+\Delta\mat{x}}$ is equal to
the $(i,j)$ entry of $\mat{C}_\mat{x} + \Delta\mat{C}_\mat{x}$
plus or minus an $O(\delta_k^2)$ term.  Therefore, 
\[
\mat{C}_{\mat{x} + \Delta\mat{x}}
= \mat{C}_\mat{x} + \Delta\mat{C}_\mat{x} + \mat{Z},
\]
where $\mat{Z}$ is a $d \times d$ matrix in which every entry
is $\pm O(\delta_k^2)$.  It follows that
\begin{equation}
\norm{Z} = O(d\delta_k^2).  \label{eq:Z2}
\end{equation}
Since $\mat{Z} = \mat{C}_{\mat{x}+\Delta\mat{x}} - (\mat{C}_\mat{x} + 
\Delta\mat{C}_\mat{x})$, $\mat{Z}$ is real symmetric.

Let $\mat{e}$ be one of the $m$ most dominant unit eigenvectors of
$\mat{C}_{\mat{x}+\Delta\mat{x}}$.  Let $\sigma$ be the eigenvalue
of $\mat{C}_{\mat{x}+\Delta\mat{x}}$
corresponding to $\mat{e}$.  Therefore,
\[
\mat{C}_{\mat{x} + \Delta\mat{x}} \cdot
\mat{e} = (\mat{C}_\mat{x} + \Delta\mat{C}_\mat{x} + \mat{Z}) \cdot \mat{e} =
\sigma\mat{e}.
\]
Let $\mat{A}_\mat{x}$ be the $d \times m$ matrix consisting of the $m$ most
dominant unit eigenvectors of $\mat{C}_\mat{x}$.  So $\col{\mat{A}_\mat{x}}$
is the linear subspace spanned by these eigenvectors.  Let $\Lambda$ be the
set of the $d-m$ smallest eigenvalues of $\mat{C}_\mat{x}$.  
We apply Lemma~\ref{lem:slant} with $\mat{M}_1 = \mat{C}_\mat{x}$,
$\mat{M}_2 = \Delta\mat{C}_\mat{x} + \mat{Z}$, 
and $r = m$:
\begin{eqnarray}
\angle (\col{\mat{A}_\mat{x}},\mat{e}) & \leq &
\arcsin\left(\frac{\norm{\Delta\mat{C}_\mat{x} + \mat{Z}}}
{\min_{\lambda \in \Lambda} |\lambda - \sigma|}\right) \nonumber \\
& \leq &
\arcsin\left(\frac{\norm{\Delta\mat{C}_\mat{x}} + \norm{\mat{Z}}}
{\min_{\lambda \in \Lambda} |\lambda - \sigma|}\right). \label{eq:angle}
\end{eqnarray}
We bound $\angle (\col{\mat{A}_\mat{x}},\mat{e})$ by showing an upper bound
for $\norm{\Delta\mat{C}_\mat{x}}$ and a lower bound for $|\lambda - \sigma|$.

For all $\mat{p} \in P \setminus B(\mat{x},m\gamma)$,
$h(\norm{\mat{x}-\mat{p}}) = \Delta h(\mat{x}-\mat{p}) = 0$.  Then,
Lemmas~\ref{lem:delta_C}(i), \ref{lemma::in_out_ratio}(i) and~\ref{lem:center} imply
that
\begin{eqnarray*}
\norm{\Delta\mat{C}_\mat{x}} 
& \leq & \frac{O(m^2\delta_k) \cdot \sum_{\mat{p} \in
P \cap B(\mat{x},m\gamma)}
\left(1-\frac{\norm{\mat{x}-\mat{p}}}{m\gamma}\right)^{2m-1}}
{\sum_{\mat{p} \in P \cap B(\mat{x},m\gamma)}
\left(1-\frac{\norm{\mat{x}-\mat{p}}}{m\gamma}\right)^{2m}
\left(\frac{2\norm{\mat{x}-\mat{p}}}{\gamma} + 1\right)} \\
& \leq &
\frac{O(\kappa m^2\delta_k) \cdot \sum_{\mat{p} \in
		P \cap B(\mat{x},m\gamma - r)}
	\left(1-\frac{\norm{\mat{x}-\mat{p}}}{m\gamma}\right)^{2m-1}}
{\sum_{\mat{p} \in P \cap B(\mat{x},m\gamma -r)}
	\left(1-\frac{\norm{\mat{x}-\mat{p}}}{m\gamma}\right)^{2m}
	\left(\frac{2\norm{\mat{x}-\mat{p}}}{\gamma} + 1\right)}, 
\end{eqnarray*}
where $r = \sqrt{m} \eps/3$.
In the
denominator, $\left(1-\frac{\norm{\mat{x}-\mat{p}}}{m\gamma}\right)
\left(\frac{2\norm{\mat{x}-\mat{p}}}{\gamma} + 1\right)$ is at its minimum
of $\frac{2\sqrt{m}\eps}{3\gamma} - \frac{2\eps^2}{9\gamma^2} + \frac{\eps}{3\sqrt{m}\gamma} = \Omega(\sqrt{m})$ when $\norm{\mat{x}-\mat{p}} = m\gamma - r$.  It follows that 
\begin{equation}
\norm{\Delta\mat{C}_\mat{x}} = O(\kappa m^{3/2}\delta_k).
\label{eq:Cx2}
\end{equation}

Lemmas~\ref{lem:gc} and~\ref{lem:C} imply that
\begin{equation}
\max\{\lambda : \lambda \in \Lambda\} = O(m\gamma).
\label{eq:lambda}
\end{equation}
We write $\mat{C}_\mat{x} + \Delta\mat{C}_\mat{x}$ as the
sum $\mat{C}_{\mat{x}+\Delta\mat{x}} + (-\mat{Z})$ and apply
Weyl's inequality~\cite[Theorem~3.3.16]{horn} to conclude that
the eigenvalue $\sigma$ is at least the $m$-th largest eigenvalue of
$\mat{C}_\mat{x} + \Delta\mat{C}_\mat{x}$ minus the largest eigenvalue
of $-\mat{Z}$. Then, 
by Lemma~\ref{lem:delta_C}(ii) and \eqref{eq:Z2},
\[
\sigma \geq 1 - O(m\gamma) - \frac{O(m\gamma) \cdot
\sum_{\mat{p} \in P} |\Delta h(\norm{\mat{x}-\mat{p}})|} {\sum_{\mat{p} \in
P} h(\norm{\mat{x}-\mat{p}})} - O(d\delta_k^2). 
\]
Together with \eqref{eq:lambda}, we obtain
\[
\min_{\lambda \in \Lambda} |\lambda - \sigma| \geq 1 - O(m\gamma) - \frac{O(m\gamma)
\cdot \sum_{\mat{p} \in P} |\Delta h(\norm{\mat{x}-\mat{p}})|}
{\sum_{\mat{p} \in P} h(\norm{\mat{x}-\mat{p}})} -
O(d\delta_k^2).  
\]
As $\delta_k$ approaches zero, both
$\Delta h(\norm{\mat{x}-\mat{p}})$ and $O(d\delta_k^2)$ approach zero.  But
$\sum_{\mat{p} \in P} h(\norm{\mat{x}-\mat{p}}) > 0.06$ by
Lemma~\ref{lemma::in_out_ratio}(ii).  Therefore, for a sufficiently
small $\delta_k$, 
\begin{equation}
\exists\, \text{a constant $\eta > 0$ such that $\min\{\lambda \in \Lambda : |\lambda - \sigma|\}
\geq \eta$}.
\label{eq:lambda-sigma}
\end{equation}

Plugging \eqref{eq:Z2}, \eqref{eq:Cx2} and \eqref{eq:lambda-sigma} into \eqref{eq:angle} gives
\[
\angle (\col{\mat{A}_\mat{x}},\mat{e}) \leq 
\arcsin\left(\frac{O(\kappa m^{3/2}\delta_k) + O(d\delta_k^2)}
{\eta}\right) 
= O(\kappa m^{3/2}\delta_k).
\]
Since $\mat{e}$ is any one of the $m$ most dominant unit eigenvectors of 
$\mat{C}_{\mat{x} + \Delta\mat{x}}$, the angle bound $O(\kappa m^{3/2}\delta_k)$
holds for all the $m$ most dominant unit eigenvectors of
$\mat{C}_{\mat{x} + \Delta\mat{x}}$.  Then, by Lemma~\ref{lem:angle},
$\col{\mat{A}_\mat{x}}$ makes an $O(\kappa m^2\delta_k)$ angle with the space
spanned by the $m$ most dominant unit eigenvectors of $\mat{C}_{\mat{x} +
\Delta\mat{x}}$.  It follows that 
$\angle (L_\mat{x},L_{\mat{x} + \Delta\mat{x}}) = O(\kappa m^2\delta_k)$.
\end{proof}

Next, we need a technical result on the angle between a vector in some linear
subspace to its projection in another linear subspace.

\begin{lemma}
\label{lem:proj-angle}
Let $E_1$ and $E_2$ be two $(d-m)$-dimensional linear subspaces that make an angle
$\phi < \pi/2$.  Let $\mat{n}$ be a unit vector in $\real^d$.  Let $\mat{u}_i$
be the projection of $\mat{n}$ in $E_i$ for $i \in [1,2]$.  Let
$\{\mat{v}_1,\ldots,\mat{v}_{d-m}\}$ and $\{\mat{w}_1,\ldots,\mat{w}_{d-m}\}$
be bases of $E_1$ and $E_2$, respectively, that satisfy either
Lemma~\ref{lem:choice}(i) or Lemma~\ref{lem:choice}(ii).
Let $\alpha_1 = \sum_{i=1}^{d-m} (\mat{n}^t\mat{v}_i^{})^2$
and let $\alpha_2 = \sum_{i=d-2m+1}^{d-m} ((\mat{w}_i-\mat{v}_i)^t \mat{n})^2$.  
If $\alpha_1 > \alpha_2 + (2m^2\phi^2)/\cos\phi$, then
\[
\frac{\mat{u}_1^t\mat{u}_2^{}}{\norm{\mat{u}_1}\norm{\mat{u}_2}} \geq
\sqrt{1-\frac{\alpha_2}{\alpha_1}}\cos\phi -
\frac{2m^2\phi^2}{\sqrt{\alpha_1^2-\alpha_1\alpha_2}}.
\]
\end{lemma} 
\begin{proof}
By Lemma~\ref{lem:choice},
\begin{eqnarray}
\forall\, i \in [1,d-2m], & &\mat{v}_i = \mat{w}_i,  \label{eq:choice-1} \\
\forall\, i \in [1,d-m],  & &\angle (\mat{v}_i,\mat{w}_i) \leq \phi, \label{eq:choice-2} \\
\forall\, i,j, \in [1,d-m], & &\angle (\mat{v}_i,\mat{w}_j-\mat{v}_j)
\in \left[(\pi-\phi)/2,(\pi+\phi)/2\right]. \label{eq:choice-3}
\end{eqnarray}
If $m \geq d/2$, then \eqref{eq:choice-1} is vacuous because $[1,d-2m]$ is an
empty range.  There is no harm done as $d-m \leq d/2$ in this case and 
Lemma~\ref{lem:choice}(i) is applicable, leading to \eqref{eq:choice-2} and
\eqref{eq:choice-3} only.  If $m < d/2$, then Lemma~\ref{lem:choice}(ii) is
applicable, leading to \eqref{eq:choice-1}, \eqref{eq:choice-2} and
\eqref{eq:choice-3}. 

Since $\mat{u}_i$ is the projection of $\mat{n}$ into $E_i$, we have
\begin{eqnarray}
\mat{u}_1 & = & (\mat{v}_1 \,\, \cdots \,\, \mat{v}_{d-m})
(\mat{v}_1 \,\, \cdots \,\, \mat{v}_{d-m})^t \mat{n}, \label{eq:u_1} \\
\mat{u}_2 & = & (\mat{w}_1 \,\, \cdots \,\, \mat{w}_{d-m})
(\mat{w}_1 \,\, \cdots \,\, \mat{w}_{d-m})^t \mat{n}. \label{eq:u_2}
\end{eqnarray}

We first bound $\mat{u}_1^t
\mat{u}_2^{}$ from below.  Standard algebra
gives 
\begin{eqnarray}
\mat{u}_1^t \mat{u}_2^{} 
& = & \sum_{i \in [1,d-m]}
\mat{n}^t\mat{v}_i^{}\mat{v}_i^t\mat{w}_i^{}\mat{w}_i^t\mat{n} + 
\sum_{\substack{i\not=j, \\ i,j \in [1,d-m]}} \mat{n}^t\mat{v}_i^{}\mat{v}_i^t\mat{w}_j^{}\mat{w}_j^t\mat{n}.  \label{eq:u_1u_2}
\end{eqnarray}
We analyze the second term in \eqref{eq:u_1u_2}.
By \eqref{eq:choice-1}, if $i \not= j$ and $i$ or $j$ belongs to $[1,d-2m]$, then
$\mat{v}_i \perp \mat{w}_j$.
It implies that $\mat{v}_i^t\mat{w}_j^{} = 0$ in the second term in
\eqref{eq:u_1u_2} whenever $i$ or $j$ belongs to $[1,d-2m]$.  The
remaining case is that both $i$ and $j$ belong to $[d-2m+1,d-m]$.

Define a vector $\mat{h}_i$ for $i \in [1,d-m]$ as follows:
\[
\forall\,i \in [1,d-m],\quad \mat{h}_i = \mat{w}_i - \mat{v}_i.  
\]
It follows from \eqref{eq:choice-2} that
\begin{equation}
\norm{\mat{h}_i} =
2\sin\frac{\angle(\mat{v}_i,\mat{w}_i)}{2} \leq \phi.  \label{eq:h_i}
\end{equation}
We rewrite \eqref{eq:u_1u_2} using $\mat{w}_i = \mat{v}_i + \mat{h}_i$ for $i \in [d-2m+1,d-m]$:
\begin{eqnarray}
\mat{u}_1^t\mat{u}_2^{} & = & \sum_{i \in [1,d-m]}
\mat{n}^t\mat{v}_i^{}\mat{v}_i^t\mat{w}_i^{}\mat{w}_i^t\mat{n} \quad + 
\sum_{\substack{i \not= j, \\ i,j \in
[d-2m+1,d-m]}} \mat{n}^t\mat{v}_i^{}\mat{v}_i^t(\mat{v}_j^{}+\mat{h}_j^{})(\mat{v}_j^{}
+\mat{h}_j^{})^t\mat{n} \nonumber \\
& = & \sum_{i \in [1,d-m]}
\mat{n}^t\mat{v}_i^{}\mat{v}_i^t\mat{w}_i^{}\mat{w}_i^t\mat{n} \quad + 
\sum_{\substack{i \not= j, \\
i,j \in [d-2m+1,d-m]}} (\mat{n}^t\mat{v}_i^{}\mat{v}_i^t\mat{h}_j^{}\mat{v}_j^t\mat{n} +
\mat{n}^t\mat{v}_i^{}\mat{v}_i^t\mat{h}_j^{}\mat{h}_j^t\mat{n}). \label{eq:u_1u_2-1}
\end{eqnarray}
Notice that if $m \geq d/2$, then $d-2m+1 \leq 1$, which implies that $[d-2m+1,d-m]$ 
acts as the range $[1,d-m]$.
In this case, Lemma~\ref{lem:choice}(i) is applicable 
and so \eqref{eq:choice-1} is vacuous, meaning that there is no simplification
from \eqref{eq:u_1u_2} to \eqref{eq:u_1u_2-1}.

By \eqref{eq:choice-2}, we get
\begin{equation}
\forall\, i \in [1,d-m], \quad \mat{v}_i^t\mat{w}_i^{} \geq \cos\phi.  \label{eq:v_iw_i}
\end{equation}
Moreover,
\begin{eqnarray}
\forall\, i,j \in [1,d-m], \quad \mat{v}_i^t\mat{h}_j^{} & = &
\norm{\mat{v}_i} \norm{\mat{h}_j} \cos(\angle(\mat{v}_i,\mat{h}_j)) \nonumber \\
& = &
\norm{\mat{h}_j} \cos(\angle(\mat{v}_i,\mat{w}_j-\mat{v}_j)) \nonumber \\
& \stackrel{\eqref{eq:choice-3}}{\geq} & -\norm{\mat{h}_j}\sin(\phi/2) \nonumber \\
& \stackrel{\eqref{eq:h_i}}{\geq} & -\phi\sin(\phi/2).  \label{eq:v_iw_j}
\end{eqnarray}
By substituting \eqref{eq:v_iw_i} and \eqref{eq:v_iw_j} into the first and second terms in \eqref{eq:u_1u_2-1},
respectively, we obtain
\begin{eqnarray*}
\mat{u}_1^t\mat{u}_2^{} & \geq & \cos\phi \left(\sum_{i \in [1,d-m]}
\mat{n}^t\mat{v}_i^{}\mat{w}_i^t\mat{n}\right) \quad - \quad
\phi\sin\frac{\phi}{2} \left( 
\sum_{\substack{i \not= j,\\ i,j \in [d-2m+1,d-m]}} (\mat{n}^t\mat{v}_i^{}\mat{v}_j^t\mat{n} +
\mat{n}^t\mat{v}_i^{}\mat{h}_j^t\mat{n})\right) \\
& = & 
\cos\phi \left(\sum_{i \in [1,d-m]}
\mat{n}^t\mat{v}_i^{}\mat{w}_i^t\mat{n}\right) \quad - \quad
\phi\sin\frac{\phi}{2} \left( 
\sum_{\substack{i \not= j,\\ i,j \in [d-2m+1,d-m]}} \mat{n}^t\mat{v}_i^{}\mat{w}_j^t\mat{n} \right).
\end{eqnarray*}
Both $\mat{n}^t\mat{v}_i$ and $\mat{w}_j^t \mat{n}$ are at most 1, which implies
that $\mat{n}^t\mat{v}_i^{}\mat{w}_j^t\mat{n} \leq 1$.  Therefore,
\[
\mat{u}_1^t\mat{u}_2^{} \, \geq \, \cos\phi \left(\sum_{i \in [1,d-m]}
\mat{n}^t\mat{v}_i^{}\mat{w}_i^t\mat{n}\right) - m^2\phi^2.
\]

Recall from the lemma statement that
$\alpha_1 = \sum_{i=1}^{d-m} (\mat{n}^t\mat{v}_i^{})^2$ and
$\alpha_2 = \sum_{i=d-2m+1}^{d-m} (\mat{h}_i^t \mat{n})^2$.  
We define one more quantity:
\[
\alpha_3 = \sum_{i=d-2m+1}^{d-m} \mat{n}^t\mat{v}_i^{}\mat{h}_i^t\mat{n}.
\]
Standard algebraic manipulation shows that $\alpha_1 + \alpha_3 = \sum_{i\in
[1,d-m]} \mat{n}^t\mat{v}_i^{}\mat{w}_i^t\mat{n}$, and therefore, 
\[
\mat{u}_1^t\mat{u}_2^{} \, \geq \, (\alpha_1+\alpha_3)\cos\phi - m^2\phi^2.
\]
By definition,
\begin{eqnarray*}
\norm{\mat{u}_1} 
& = & \sqrt{\sum_{i \in [1,d-m]} (\mat{n}^t\mat{v}_i)^2} \,\,\, = \,\,\, \sqrt{\alpha_1}, \\
\norm{\mat{u}_2} 
& = & \sqrt{\sum_{i \in [1,d-m]} (\mat{n}^t\mat{w}_i)^2} \\ 
& \stackrel{\eqref{eq:choice-1}}{=} & 
\sqrt{\sum_{i \in [1,d-2m]} (\mat{n}^t\mat{v}_i)^2 \quad + 
\sum_{i \in [d-2m+1,d-m]} (\mat{n}^t\mat{w}_i)^2 }   \\
& = & 
\sqrt{\sum_{i \in [1,d-2m]} (\mat{n}^t\mat{v}_i)^2 \quad + 
\sum_{i \in [d-2m+1,d-m]} (\mat{n}^t(\mat{v}_i + \mat{h}_i))^2 }   \\
& = & \sqrt{\sum_{i \in [1,d-m]} (\mat{n}^t\mat{v}_i^{})^2 + \sum_{i \in [d-2m+1,d-m]} 
(2\mat{n}^t\mat{v}_i^{}\mat{h}_i^t\mat{n} + (\mat{h}^t_i\mat{n})^2)} \\
& = & \sqrt{\alpha_1 + 2\alpha_3 + \alpha_2}.  
\end{eqnarray*}
Consequently,
\begin{eqnarray}
\frac{\mat{u}_1^t\mat{u}_2^{}}{\norm{\mat{u}_1}\norm{\mat{u}_2}} & \geq &
\frac{(\alpha_1+\alpha_3)\cos\phi - m^2\phi^2}{\norm{u_1}\norm{u_2}} \nonumber \\
& = & \frac{(\alpha_1+\alpha_3)\cos\phi - m^2\phi^2} {\sqrt{\alpha_1} \cdot
\sqrt{\alpha_1+2\alpha_3+\alpha_2}}.  
\label{eq:u_1u_2-2}
\end{eqnarray}

Treating $\alpha_3$ as a free variable while fixing the other values, we can
apply standard calculus to show that the right hand side of \eqref{eq:u_1u_2-2}
is minimized when $\alpha_3 = -\alpha_2-\frac{m^2\phi^2} {\cos\phi}$ under the
condition that $\alpha_1 > \alpha_2 + \frac{2m^2\phi^2}{\cos\phi}$.  (This
condition ensures that the denominator $\sqrt{\alpha_1^2 + 2\alpha_1\alpha_3 +
\alpha_1\alpha_2}$ is real and positive.)  This condition is assumed to be
satisfied in the lemma statement.  Substituting $\alpha_3 =
-\alpha_2-\frac{m^2\phi^2}{\cos\phi}$ into \eqref{eq:u_1u_2-2} gives
\begin{eqnarray*}
\frac{\mat{u}_1^t\mat{u}_2^{}}{\norm{\mat{u}_1}\norm{\mat{u}_2}} & \geq &
\frac{(\alpha_1-\alpha_2)\cos\phi - 2m^2\phi^2}{\sqrt{\alpha_1\left(\alpha_1 - \alpha_2
- 2m^2\phi^2/\cos\phi\right)}} \\
& \geq & 
\frac{(\alpha_1-\alpha_2)\cos\phi - 2m^2\phi^2}{\sqrt{\alpha_1^2 - \alpha_1\alpha_2}} \\
& = & \sqrt{1-\frac{\alpha_2}{\alpha_1}}\cos\phi -
\frac{2m^2\phi^2}{\sqrt{\alpha_1^2-\alpha_1\alpha_2}}.
\end{eqnarray*}
\end{proof}

We are ready to bound the instantaneous change in the normalized projection of
a normal vector of $\mani$ into $L_\mat{x}$ as $\mat{x}$ moves, which is the
main result of this section.

\begin{lemma}\label{lemma::gradient}
Let $\mat{z}$ be any point in $\mani$. Let $\mat{n}$ be any unit vector in
$N_\mat{z}$.  Define the function $f: B(\mat{z},2\eps) \rightarrow L_\mat{x}$
such that $f(\mat{x})$ is the normalized projection of $\mat{n}$ into
$L_\mat{x}$, i.e., $f(\mat{x})$ is the unit vector in $L_\mat{x}$ parallel to
the projection of $\mat{n}$ in $L_\mat{x}$.  
For every point $\mat{x}$ in the interior of $B(\mat{z},2\eps)$
and every $k \in [1,d]$, $\norm{\partial f(\mat{x})/\partial x_k} =
O(\kappa m^{3})$.
\end{lemma}
\begin{proof} 
Let $\mat{x}$ be a point in the interior of $B(\mat{z},2\eps)$.  Consider any index $k \in
[1,d]$.  Let $\Delta\mat{x}$ be a vector parallel to the $x_k$-axis such that
$\mat{x}+\Delta\mat{x} \in B(\mat{z},2\eps)$ and $\delta_k =
\norm{\Delta\mat{x}}$ is arbitrarily small.  Let $\phi$ denote the angle
$\angle (L_\mat{x},L_{\mat{x}+\Delta\mat{x}})$.  By
Lemma~\ref{lemma::normal_change}, $\phi = O(\kappa m^2\,\delta_k)$.  Since $\phi < \pi/2$,
there are orthonormal bases of $L_\mat{x}$ and $L_{\mat{x}+\Delta\mat{x}}$ that
satisfy either Lemma~\ref{lem:choice}(i) or Lemma~\ref{lem:choice}(ii).  Let
$\{\mat{v}_1,\ldots,\mat{v}_{d-m}\}$ and $\{\mat{w}_1,\ldots,\mat{w}_{d-m}\}$
be such orthonormal bases of $L_\mat{x}$ and
$L_{\mat{x}+\Delta\mat{x}}$, respectively. 
We want to apply Lemma~\ref{lem:proj-angle}, so we need to verify that
$\alpha_1 > \alpha_2 + (2m^2\phi^2)/\cos\phi$, where $\alpha_1 = \sum_{i=1}^{d-m}
(\mat{n}^t\mat{v}_i)^2$ and $\alpha_2 = \sum_{i=d-2m+1}^{d-m}
((\mat{w}_i-\mat{v}_i)^t\mat{n})^2$.  

First, $\alpha_2 \leq \sum_{i=d-2m+1}^{d-m} \norm{\mat{w}_i-\mat{v}_i}^2$.
Since $\angle (\mat{v}_i,\mat{w}_i) \leq \phi$ for $i \in [d-2m+1,d-m]$ by
Lemma~\ref{lem:choice}, we obtain $\norm{\mat{w}_i-\mat{v}_i} =
2\sin\frac{\angle(\mat{v}_i,\mat{w}_i)}{2} \leq \phi$.  It follows that 
\[
\alpha_2 \leq m\phi^2 = O(\kappa^2m^5\delta_k^2).  
\]
Second, observe that $\alpha_1 = \left\|(\mat{v}_1 \, \cdots \, \mat{v}_{d-m}) 
(\mat{v}_1 \, \cdots \, \mat{v}_{d-m})^t \mat{n}\right\|^2$, where 
$(\mat{v}_1 \, \cdots \, \mat{v}_{d-m}) (\mat{v}_1 \, \cdots \, \mat{v}_{d-m})^t \mat{n}$
is the projection of $\mat{n}$ into $L_\mat{x}$.   Therefore,
$\alpha_1 \geq \cos^2 (\angle(L_\mat{x},N_\mat{z}))$.  Then, 
Lemma~\ref{lemma::normal_angle} implies that
\[
\alpha_1 \geq \cos^2 (O(m\sqrt{m}\,\gamma)) \geq 1 - O(m^3\gamma^2).
\]
As $\alpha_2 + \frac{2m^2\phi^2}{\cos\phi}$ approaches zero as $\delta_k
\rightarrow 0$, we get $\alpha_1
> \alpha_2 + \frac{2m^2\phi^2}{\cos\phi}$.
Then, by Lemma~\ref{lem:proj-angle},
\[
\frac{\mat{u}_1^t\mat{u}_2^{}}{\norm{\mat{u}_1}\norm{\mat{u}_2}}
\geq 
\sqrt{1-\frac{\alpha_2}{\alpha_1}}\cos\phi -
\frac{2m^2\phi^2}{\sqrt{\alpha_1^2-\alpha_1\alpha_2}},
\]
where $\mat{u}_1$ and $\mat{u}_2$ are the projections of $\mat{n}$
into $L_\mat{x}$ and $L_{\mat{x}+\Delta\mat{x}}$, respectively.
Finally, 
\begin{eqnarray*}
\left\| \frac{\partial f(\mat{x})}{\partial x_k} \right\|^2 
& = & \lim_{\delta_k \rightarrow 0} \, \frac{1}{\delta_k^2}
\left(\frac{\mat{u}_2}{\norm{\mat{u}_2}} -
\frac{\mat{u}_1}{\norm{\mat{u}_1}}\right)^t
\left(\frac{\mat{u}_2}{\norm{\mat{u}_2}} -
\frac{\mat{u}_1}{\norm{\mat{u}_1}}\right) \\
& = & \lim_{\delta_k \rightarrow 0} \,
\frac{1}{\delta_k^2} \left(2 - \frac{2\mat{u}_1^t\mat{u}_2^{}}
{\norm{\mat{u}_1}\norm{\mat{u}_2}}\right) \\
& \leq & \lim_{\delta_k \rightarrow 0} \,
\frac{1}{\delta_k^2} \left(2 - 2\sqrt{1-\frac{\alpha_2}{\alpha_1}}\cos\phi +
\frac{4m^2\phi^2}{\sqrt{\alpha_1^2-\alpha_1\alpha_2}}\right) \\
& \leq &
\lim_{\delta_k \rightarrow 0} \, \frac{1}{\delta_k^2} \left(2 -
2\left(1-\frac{\alpha_2}{\alpha_1}\right)\cos\phi +
\frac{4m^2\phi^2}{\sqrt{\alpha_1^2 - \alpha_1\alpha_2}}\right). 
\end{eqnarray*}
We have shown earlier that $\alpha_2 \leq m\phi^2$
and $\alpha_1 \geq 1-O(m^3\gamma^2)$.
Using these relations and the facts that $\cos\phi \geq 1-\phi^2/2$
and $\phi = O(\kappa m^2\,\delta_k)$, we obtain
\begin{eqnarray*}
\left\| \frac{\partial f(\mat{x})}{\partial x_k} \right\|^2 
& \leq & \lim_{\delta_k \rightarrow 0} \, \frac{1}{\delta_k^2} \left(2 - 2\left(1 -
\frac{m\phi^2}{\alpha_1} \right)\left(1 - \frac{\phi^2}{2}\right) + 
\frac{4m^2\phi^2}{\sqrt{\alpha_1^2 - \alpha_1m\phi^2}} \right) \\
& = & \lim_{\delta_k \rightarrow 0} \, \frac{1}{\delta_k^2} \left(2 - 2\left(1 -
\frac{m\phi^2}{\alpha_1} - \frac{\phi^2}{2} + \frac{m\phi^4}{2\alpha_1}\right) + 
\frac{4m^2\phi^2}{\sqrt{\alpha_1^2 - \alpha_1 m\phi^2}} \right) \\
& \leq & \lim_{\delta_k \rightarrow 0} \, \frac{1}{\delta_k^2} \left(
O\left(\frac{\kappa^2 m^5 \delta_k^2}{\alpha_1}\right) + 
\frac{O(\kappa^2 m^6 \delta_k^2)}{\sqrt{\alpha_1^2 - O(\alpha_1 \kappa^2 m^5 \delta_k^2)}} \right) \\
& = & O(\kappa^2m^6).
\end{eqnarray*}
\end{proof}

We use Lemma~\ref{lemma::gradient} to bound $\norm{\mat{J}_f(\mat{x})}$.
Multiplying the bound in Lemma~\ref{lemma::gradient} by $\sqrt{d}$ already
gives a bound.  We give a tighter analysis that yields a bound independent of
$d$.

\begin{lemma}
\label{cor:gradient}
Let $\mat{z}$ be any point in $\mani$.  Let $\mat{J}_f$ be the Jacobian of the
function $f :
B(\mat{z},2\eps)\rightarrow L_\mat{x}$ defined in Lemma~\ref{lemma::gradient}.
For any point $\mat{x}$ in the interior of $B(\mat{z},2\eps)$,
$\norm{\mat{J}_f(\mat{x})} = O(\kappa m^{3})$.
\end{lemma}
\begin{proof}
Fix a unit vector $\mat{n} \in N_\mat{z}$
as required in the definition of $f$ in Lemma~\ref{lemma::gradient}.
Let $\mat{x}$ be a point in the interior of $B(\mat{z},2\eps)$.
Let $\mat{R}$ be any $d \times d$ orthogonal matrix.
Apply the orthogonal transformation induced by $\mat{R}$ to $\real^d$.  Then
define the function $g : B(\mat{z}',2\eps) \rightarrow L_{\mat{x}'}$, 
where $\mat{z}' = \mat{R} \cdot \mat{z}$ and $\mat{x}' = \mat{R} \cdot \mat{x}$, 
such that
$g(\mat{x}')$ is the normalized projection of $\mat{R} \cdot \mat{n}$ into
$L_{\mat{x}'}$.

First, we show that $f(\mat{x}) = \mat{R}^t \cdot g(\mat{x}')$.  Let $\ell$ be
the length of the projection of $\mat{n}$ into $L_\mat{x}$.  Let $\mat{Q}$ be
any $d \times (d-m)$ matrix whose columns form an orthonormal basis of
$L_\mat{x}$.  It follows from the definition of $f$ that $f(\mat{x}) =
\frac{1}{\ell} \cdot \mat{Q}\cdot \mat{Q}^t\cdot \mat{n}$.  Since an orthogonal
transformation preserves lengths, $\ell$ is also the length of the projection
of $\mat{R}\cdot\mat{n}$ into $L_{\mat{x}'}$.  Then, $g(\mat{x}') =
\frac{1}{\ell} \cdot \mat{R}\cdot\mat{Q}\cdot\mat{Q}^t\cdot\mat{R}^t \cdot
\mat{R}\cdot\mat{n} =
\frac{1}{\ell}\cdot\mat{R}\cdot\mat{Q}\cdot\mat{Q}^t\cdot\mat{n}$, which
implies that $f(\mat{x}) = \mat{R}^t \cdot g(\mat{x}')$.

We show that $\mat{J}_f(\mat{x}) = \mat{R}^t \cdot \mat{J}_g(\mat{x}') \cdot \mat{R}$.  Let
$\Delta\mat{x}$ be an arbitrarily short vector.  By Taylor's Theorem,
\begin{equation}
f(\mat{x}+\Delta\mat{x}) = f(\mat{x}) + \mat{J}_f(\mat{x}) \cdot \Delta\mat{x} +
\mat{e}_f,
\label{eq:cor:gradient-1} 
\end{equation}
where $\mat{e}_f/\norm{\Delta\mat{x}}$ converges to the zero vector
as $\norm{\Delta\mat{x}} \rightarrow 0$.  Similarly, 
\begin{equation}
g(\mat{R}\cdot\mat{x}+\mat{R}\cdot\Delta\mat{x}) =
g(\mat{x}') + \mat{J}_g(\mat{x}') \cdot \mat{R}\cdot\Delta\mat{x} +
\mat{e}_g,  
\label{eq:cor:gradient-2}
\end{equation}
where $\mat{e}_g/\norm{\mat{R}\cdot\Delta\mat{x}}$ converges to the zero vector
as $\norm{\mat{R}\cdot\Delta\mat{x}} \rightarrow 0$.  Since $\mat{R}$ is fixed,
it means that $\mat{e}_g/\norm{\Delta\mat{x}}$ tends to the zero vector as
$\norm{\Delta\mat{x}} \rightarrow 0$.  We multiply both sides of
\eqref{eq:cor:gradient-2} by $\mat{R}^t$ and then subtract the resulting
equation from \eqref{eq:cor:gradient-1}.  Some terms cancel each other because
$f(\mat{x} + \Delta\mat{x}) = \mat{R}^t \cdot g\left(\mat{R}\cdot(\mat{x} +
\Delta\mat{x})\right)$ and $f(\mat{x}) = \mat{R}^t \cdot g(\mat{x}') =
\mat{R}^t \cdot g\left(\mat{R} \cdot \mat{x}\right)$.  We obtain
\[
\left(\mat{J}_f(\mat{x}) - \mat{R}^t \cdot \mat{J}_g(\mat{x}') \cdot \mat{R} \right) \cdot \Delta\mat{x} =
\mat{R}^t\cdot\mat{e}_g - \mat{e}_f.  
\]
Therefore, 
\[
\left\|\left(\mat{J}_f(\mat{x}) -
\mat{R}^t \cdot \mat{J}_g(\mat{x}') \cdot \mat{R}\right) \cdot \Delta\mat{x}\right\| \leq 
\left\|\mat{R}^t\cdot\mat{e}_g\right\|
+ \left\|\mat{e}_f\right\|.  
\]
We are free to choose the direction of $\Delta\mat{x}$.  We choose it 
such that $\left\|(\mat{J}_f -
\mat{R}^t \cdot \mat{J}_g(\mat{x}') \cdot \mat{R}) \cdot \Delta\mat{x}\right\| = \left\|\mat{J}_f(\mat{x}) -
\mat{R}^t \cdot \mat{J}_g(\mat{x}') \cdot \mat{R}\right\| \cdot \left\|\Delta\mat{x}\right\|$, i.e.,
$\Delta\mat{x}$ is an eigenvector corresponding to the largest eigenvalue of
$\mat{J}_f(\mat{x}) - \mat{R}^t \cdot \mat{J}_g(\mat{x}') \cdot \mat{R}$.  Then,
\[
\left\|\mat{J}_f(\mat{x}) - \mat{R}^t \cdot \mat{J}_g(\mat{x}') \cdot \mat{R}\right\| \leq
\frac{\left\|\mat{R}^t\mat{e}_g\right\|}{\left\|\Delta\mat{x}\right\|} +
\frac{\left\|\mat{e}_f\right\|}{\left\|\Delta\mat{x}\right\|}.  
\]
Since the right hand side tends to
zero as $\norm{\Delta\mat{x}} \rightarrow 0$, we conclude that 
\[
\lim_{\norm{\Delta\mat{x}} \rightarrow 0} 
\norm{\mat{J}_f(\mat{x}) - \mat{R}^t \cdot \mat{J}_g(\mat{x}') \cdot \mat{R}} = 0, 
\]
which implies that $\mat{J}_f(\mat{x}) = \mat{R}^t \cdot \mat{J}_g(\mat{x}') \cdot \mat{R}$.

By definition, $\norm{\mat{J}_f(\mat{x})} = \norm{\mat{J}_f(\mat{x}) \cdot \mat{v}}$ for some unit
vector $\mat{v}$.  We choose $\mat{R}$ to be the $d \times d$ orthogonal 
matrix such that $\mat{R} \cdot\mat{v} = (1,0,\ldots,0)^t$.  Then, $\norm{\mat{R}
\cdot \mat{J}_f(\mat{x}) \cdot \mat{v}} = \norm{\mat{R} \cdot \mat{J}_f(\mat{x}) \cdot \mat{R}^t \cdot
\mat{R} \cdot \mat{v}} = \norm{\mat{J}_g(\mat{x}') \cdot (1,0,\ldots,0)^t}$, which is the
2-norm of the first column of $\mat{J}_g(\mat{x}')$.  Lemma~\ref{lemma::gradient} is
independent of the coordinate frame.  So we can apply Lemma~\ref{lemma::gradient} to $g$
and conclude that the 2-norm of the first
column of $\mat{J}_g(\mat{x}')$ is $O(\kappa m^{3})$.  As a result, $\norm{\mat{R} \cdot
\mat{J}_f(\mat{x}) \cdot \mat{v}} = O(\kappa m^{3})$.   Since multiplying any vector with an
orthogonal matrix preserves the 2-norm of the vector, we conclude that
$\norm{\mat{J}_f(\mat{x})} = \norm{\mat{J}_f(\mat{x}) \cdot \mat{v}} = \norm{\mat{R} \cdot
\mat{J}_f(\mat{x}) \cdot \mat{v}} = O(\kappa m^{3})$.
\end{proof}

\section{Faithful reconstruction}

In this section, we prove our main result that $Z_\varphi \cap \widehat{\mani}$ is a faithful
reconstruction of $\mani$.  Recall the class $\Phi$ of functions $\varrho:
\real^d \rightarrow \real^{d-m}$:
\begin{quote}
$\displaystyle \Phi = \left\{ \varrho : \varrho(\mat{x}) = \sum_{\mat{p} \in P}
\omega(\mat{x},\mat{p}) \cdot \mat{B}^t_{\varrho,\mat{x}} \cdot (\mat{x}-\mat{p}) \right\}$,
where $\mat{B}_{\varrho,\mat{x}}$ is any $d \times (d-m)$ matrix with linearly 
independent columns such that $\col{\mat{B}_{\varrho,\mat{x}}} = L_{\mat{x}}$.
\end{quote}
We claim that the choice of $\mat{B}_{\varrho,\mat{x}}$ has no impact on the zero-set $Z_\varrho$ as long as
the columns of $\mat{B}_{\varrho,\mat{x}}$ are linearly independent.  In
this section, we will prove some useful properties of functions in $\Phi$.
These properties will allow us to show that $Z_\varphi \cap \widehat{\mani}$ is 
a faithful approximation of $\mani$.  

We will study properties of $Z_\varphi \cap \widehat{\mani}$ by analyzing $Z_\varrho
\cap \widehat{\mani}$ for another function $\varrho \in \Phi$ conveniently chosen
for the analysis.  Since we will conduct some local analysis, we are only
concerned with functions that are defined near some chosen points in $\mani$.
This motivates us to define for every point $\mat{z} \in \mani$ 
the following class $\Phi_{\mat{z}}$ of functions:
\begin{quote}
$\displaystyle \Phi_{\mat{z}} = \left\{ \varrho : 
\varrho: B(\mat{z},2\eps) \rightarrow \real^{d-m},\,\, 
\varrho(\mat{x}) = \sum_{\mat{p} \in P}
\omega(\mat{x},\mat{p}) \cdot \mat{B}^t_{\varrho,\mat{x}} \cdot (\mat{x}-\mat{p}) \right\}$,
where $\mat{B}_{\varrho,\mat{x}}$ is any $d \times (d-m)$ matrix with linearly 
independent columns such that $\col{\mat{B}_{\varrho,\mat{x}}} = L_{\mat{x}}$.
\end{quote}
$\Phi_{\mat{z}}$ is a local version of $\Phi$.  The next result shows that
functions in $\Phi_\mat{z}$ with overlapping domains have consistent
zero sets. 

\begin{lemma}
\label{lem:agree}
Let $\mat{y}$ and $\mat{z}$ be two arbitrary points in $\mani$
that are not necessarily distinct.
For every point $\mat{x} \in B(\mat{y},2\eps) \cap B(\mat{z},2\eps)$,
if there exists $\varrho \in \Phi_{\mat{y}}$ such
that $\varrho(\mat{x}) = \mat{0}_{d-m,1}$, then
for every $\varrho \in \Phi_{\mat{y}} \cup \Phi_{\mat{z}}$, 
$\varrho(\mat{x}) = \mat{0}_{d-m,1}$.
\end{lemma}
\begin{proof}
Take two functions $\varrho, \bar{\varrho} \in \Phi_{\mat{y}} \cup \Phi_{\mat{z}}$.  
Fix a point $\mat{x} \in B(\mat{y},2\eps) \cap B(\mat{z},2\eps)$.  By
definition, $\varrho(\mat{x}) = \sum_{\mat{p} \in P}
\omega(\mat{x},\mat{p}) \cdot \mat{B}_{\varrho,\mat{x}}^t \cdot
(\mat{x}-\mat{p})$ and $\bar{\varrho}(\mat{x}) = \sum_{\mat{p} \in P}
\omega(\mat{x},\mat{p}) \cdot \mat{B}_{\bar{\varrho},\mat{x}}^t \cdot
(\mat{x}-\mat{p})$.  The columns of $\mat{B}_{\varrho,\mat{x}}$ and
$\mat{B}_{\bar{\varrho},\mat{x}}$ form two 
bases of $L_\mat{x}$, which means that there is a $(d-m) \times (d-m)$ 
invertible matrix $\mat{R}$ such that $\mat{R} \cdot \mat{B}_{\varrho,\mat{x}}^t =
\mat{B}_{\bar{\varrho},\mat{x}}^t$.  If $\varrho(\mat{x}) =
\mat{0}_{d-m,1}$, then $\bar{\varrho}(\mat{x}) = \sum_{\mat{p} \in P}
\omega(\mat{x},\mat{p}) \cdot \mat{R} \cdot \mat{B}_{\varrho,\mat{x}}^t
\cdot (\mat{x}-\mat{p}) = \mat{R} \cdot \varrho(\mat{x}) =
\mat{0}_{d-m,1}$.
\end{proof}

We define a particular function $\varrho_{\mat{z}} \in \Phi_\mat{z}$ to analyze
the properties of $Z_\varphi \cap \widehat{\mani}$ in a small neighborhood of
$\mat{z}$.

\begin{definition}
\label{df:local}
Let $\mat{z}$ be any point in $\mani$.  Let $\{\mat{v}_1,\ldots,
\mat{v}_{d-m}\}$ be any set of unit vectors forming a basis of $N_\mat{z}$.
For $i \in [1,d-m]$, let $f_{\mat{v}_i}$ be the function that maps every
$\mat{x}$ in $B(\mat{z},2\eps)$ to the normalized projection of $\mat{v}_i$ in
$L_\mat{x}$.  Define a {\bf canonical function} $\varrho_\mat{z} : B(\mat{z},2\eps) \rightarrow
\real^{d-m}$ with respect to $\mat{z}$ and $\{\mat{v}_1,\ldots,\mat{v}_{d-m}\}$
such that for all $\mat{x} \in B(\mat{z}, 2\eps)$,
$\varrho_\mat{z}(\mat{x}) = \sum_{\mat{p}\, \in P} \omega(\mat{x},\mat{p})
\cdot [f_{\mat{v}_1}(\mat{x}),\ldots,f_{\mat{v}_{d-m}}(\mat{x})]^t \cdot
(\mat{x}-\mat{p})$.
\end{definition}

We show that whenever $\eps$ is sufficiently small, $\varrho_\mat{z}$ belongs
to $\Phi_{\mat{z}}$ and $\varrho_{\mat{z}}$ is continuous in the interior of
$B(\mat{z}, 2\eps)$. 

\begin{lemma}\label{lemma:local}
Let $\varrho_\mat{z}$ be the canonical function with respect
to a point $\mat{z} \in \mani$ and some set of unit vectors 
$\{\mat{v}_1,\ldots,\mat{v}_{d-m}\}$ forming a basis of $N_\mat{z}$ for
which there exists some $\phi \in \left[0,\arcsin\left(\frac{1}{3d-3m}\right)\right)$ such
that for any distinct $i,j \in [1,d-m]$, $\angle (\mat{v}_i,\mat{v}_j)
\in [\pi/2-\phi,\pi/2+\phi]$. 
There exists $\eps_0 \in (0,1)$ that decreases as $d$ increases such that 
for every point $\mat{z} \in \mani$, if $\eps \leq \eps_0$, then
$\varrho_\mat{z} \in \Phi_{\mat{z}}$ and
$\varrho_\mat{z}$ is continuous in the interior of $B(\mat{z}, 2\eps)$.
\end{lemma}
\begin{proof}
To show that $\varrho_\mat{z} \in \Phi_{\mat{z}}$, it suffices to prove
that $\{f_{\mat{v}_1}(\mat{x}),\ldots,f_{\mat{v}_{d-m}}(\mat{x})\}$ form
a basis of $L_\mat{x}$, which boils down to showing that
$\{f_{\mat{v}_1}(\mat{x}),\ldots,f_{\mat{v}_{d-m}}(\mat{x})\}$ are
linearly independent.  

Since $\angle (L_\mat{x},N_\mat{z}) = O(m\sqrt{m}\,\gamma)$ by
Lemma~\ref{lemma::normal_angle}, we get $\angle
(f_{\mat{v}_i}(\mat{x}),\mat{v}_i) = O(m\sqrt{m}\,\gamma)$.  Assume to the contrary
that $f_{\mat{v}_1}(\mat{x}),  \ldots, f_{\mat{v}_{d-m}}(\mat{x})$ are linearly
dependent. Then, 
\[
\angle \left(f_{\mat{v}_1}(\mat{x}),
\col{(f_{\mat{v}_2}(\mat{x}) \,  
\cdots \, f_{\mat{v}_{d-m}}(\mat{x}))} \right) = 0. 
\]
Since $\angle \left(\mat{v}_i, \col{(f_{\mat{v}_2}(\mat{x})  \,
\cdots \, f_{\mat{v}_{d-m}}(\mat{x}))}\right) = O(m\sqrt{m}\,\gamma)$ for all $i \in [2,d-m]$, 
Lemma~\ref{lem:angle} implies that 
\[
\angle \left(\col{(\mat{v}_2  \,
\cdots \, \mat{v}_{d-m})}, \col{(f_{\mat{v}_2}(\mat{x})  \,
\cdots \, f_{\mat{v}_{d-m}}(\mat{x}))} \right) = O\left(m\sqrt{dm-m^2}\,\gamma\right). 
\]
By triangle inequality, 
$\angle \left(\mat{v}_1, \col{(\mat{v}_2  \,
\cdots \, \mat{v}_{d-m})}\right) \leq
\angle \left(\mat{v}_1, f_{\mat{v}_1}(\mat{x})\right) + 
\angle \left(f_{\mat{v}_1}(\mat{x}), \col{(\mat{v}_2  \,
\cdots \, \mat{v}_{d-m})}\right)$.
The dimension of $\col{(\mat{v}_2  \,
\cdots \, \mat{v}_{d-m})}$ is at least the dimension of $\col{(f_{\mat{v}_2}(\mat{x})  \,
\cdots \, f_{\mat{v}_{d-m}}(\mat{x}))}$. Thus, 
\begin{eqnarray*}
\angle \left(f_{\mat{v}_1}(\mat{x}), \col{(\mat{v}_2  \,
\cdots \, \mat{v}_{d-m})}\right) & \leq &
\angle \left(f_{\mat{v}_1}(\mat{x}), \col{(f_{\mat{v}_2}(\mat{x})  \,
\cdots \, f_{\mat{v}_{d-m}}(\mat{x}))}\right) + \\ 
& & \angle \left(\col{(\mat{v}_2  \,
\cdots \, \mat{v}_{d-m})}, \col{(f_{\mat{v}_2}(\mat{x})  \,
\cdots \, f_{\mat{v}_{d-m}}(\mat{x}))}\right).
\end{eqnarray*}
Combining the above observations, we obtain
\begin{eqnarray*}
\angle (\mat{v}_1, \col{(\mat{v}_2  \,
\cdots \, \mat{v}_{d-m})}) 
& \leq & 
\angle \left(\mat{v}_1, f_{\mat{v}_1}(\mat{x})\right) + 
\angle \left(f_{\mat{v}_1}(\mat{x}), \col{(\mat{v}_2  \,
\cdots \, \mat{v}_{d-m})}\right) \\
& \leq & \angle \left(\mat{v}_1, f_{\mat{v}_1}(\mat{x})\right) + 
\angle \left(f_{\mat{v}_1}(\mat{x}), \col{(f_{\mat{v}_2}(\mat{x})  \,
\cdots \, f_{\mat{v}_{d-m}}(\mat{x}))}\right) + \\
& & \angle \left(\col{(\mat{v}_2  \,
\cdots \, \mat{v}_{d-m})}, \col{(f_{\mat{v}_2}(\mat{x})  \,
\cdots \, f_{\mat{v}_{d-m}}(\mat{x}))}\right) \\
& = & O(m\sqrt{dm-m^2}\,\gamma).
\end{eqnarray*} 

Recall that $\gamma = 4\eps \leq 4\eps_0$.  Assume that $\eps_0 < \frac{1}{Cm\sqrt{dm-m^2}}$ for some appropriate constant $C \geq 1$.  Then $\angle (\mat{v}_1,\col{\mat{v}_2 \cdots \mat{v}_{d-m}}) < \pi/6$.  Note that $\eps_0$ decreases as $d$ increases.  Let $\mat{u}$ be the normalized projection of $\mat{v}_1$ in $\col{\mat{v}_2 \cdots \mat{v}_{d-m}}$.  It means that 
\[
\mat{v}_1^t \cdot \mat{u} > \cos(\pi/6) = \sqrt{3}/2.  
\]
We can write $\mat{u} = \sum_{i=2}^{d-m} \lambda_i \mat{v}_i$ for some $\lambda_i$.   Let $k = \argmax_{i=[2,d-m]} |\lambda_i|$. We take the dot product of $\mat{u}$ and $\sign(\lambda_k)\mat{v}_k$. This dot product is equal to $|\lambda_k|\norm{\mat{v}_k}^2 + \sign(\lambda_k)\sum_{i\neq k} \lambda_i \mat{v}_i^t \cdot \mat{v}_k$ and it is at most 1 as $\mat{u}$ and $\mat{v}_k$ are unit vectors. Since $\angle (\mat{v}_i,\mat{v}_j) \in \bigl[\frac{\pi}{2}-\phi,\frac{\pi}{2}+\phi\bigr]$, the projection of $\mat{v}_j$ in the direction of $\mat{v}_i$ has magnitude at most $\sin\phi$. It follows that
\begin{eqnarray*}
1 & \geq & |\lambda_k| - \sum_{i\neq k} |\lambda_i| \mat{v}_i^t \cdot \mat{v}_k\\
& \geq &
|\lambda_k| - (d-m-2) |\lambda_k| \sin \phi.
\end{eqnarray*}
We get $|\lambda_k| \leq 1/ (1 - (d-m-2)\sin\phi) < 1.5$ because $\sin\phi < \frac{1}{3d-3m}$ by assumption of the lemma.  Thus,



\[
\mat{v}_1^t \cdot \mat{u} = \sum_{i=2}^{d-m} \lambda_i \mat{v}_1^t \cdot \mat{v}_i \leq \sin\phi \cdot \sum_{i=2}^{d-m} |\lambda_i| < 1.5(d-m)\sin\phi < 0.5.
\]
This is a contradiction because we have derived earlier that $\mat{v}_1^t \cdot \mat{u} > \sqrt{3}/2$.  We conclude that  $\{f_{\mat{v}_1}(\mat{x}),\ldots,f_{\mat{v}_{d-m}}(\mat{x})\}$ are
linearly independent, and therefore, $\varrho_\mat{z} \in \Phi_{\mat{z}}$.  


By Lemma~\ref{lemma::gradient}, for $i \in [1,d-m]$, $f_{\mat{v}_i}$ is
differentiable and hence continuous in the interior of $B(\mat{z},2\eps)$.
Because $\varrho_\mat{z}$ is a sum of products of continuous functions, 
$\varrho_\mat{z}$ is also continuous in the interior of 
$B(\mat{z},2\eps)$ \cite[Ch~2:~Corollary~3.7]{mendelson}.
\end{proof}

Next, we show that the gradient of $\varrho_\mat{z}$ varies monotonically.

\begin{lemma}\label{lemma::unique}
Let $\mat{z}$ be any point in $\mani$.  Let $\mat{v}_i$ be any unit
vector in $N_\mat{z}$.  For any $\mat{x} \in B(\mat{z},2\eps)$, let 
$\varrho_{\mat{z},i}(\mat{x}) = \sum_{\mat{p} \in P} \omega(\mat{x},\mat{p})
\cdot f_{\mat{v}_i}(\mat{x})^t \cdot (\mat{x}-\mat{p})$.
Let $\tau$ be any value greater than $1$. 
For every $t \geq 1$ and every point $\mat{x} \in B(\mat{z}, t\eps^{\tau})$, 
\begin{itemize}

\item $\left\|\nabla \varrho_{\mat{z},i}(\mat{x})\right\| \in \left[1 - O(t\kappa\sqrt{m}\eps^{\tau-1}+ \kappa m^{4}\gamma),
1 + O(t\kappa\sqrt{m}\eps^{\tau-1} + \kappa m^{4}\gamma)\right]$ and 

\item $\mat{v}_i^t \cdot \nabla\varrho_{\mat{z},i}(\mat{x}) \geq 1 - O(t\kappa\sqrt{m}\eps^{\tau-1} + \kappa m^{4}\gamma)$.

\end{itemize}
\end{lemma}
\begin{proof}
From the definition of $\varrho_{\mat{z},i}(\mat{x}) = \sum_{\mat{p} \in P} \omega(\mat{x},\mat{p})
\cdot f_{\mat{v}_i}(\mat{x})^t \cdot (\mat{x}-\mat{p})$, we obtain
\begin{eqnarray}
\left\|\nabla \varrho_{\mat{z},i}(\mat{x})\right\| 
& \leq & \sum_{\mat{p} \in P} \left(\omega(\mat{x},\mat{p}) +
\omega(\mat{x},\mat{p}) \cdot \norm{\mat{J}_{f_{\mat{v}_i}}(\mat{x})} \cdot 
\norm{\mat{x}-\mat{p}} \right) + \nonumber \\
& & \left\|\sum_{\mat{p} \in P}
\nabla \omega(\mat{x},\mat{p}) \cdot
f_{\mat{v}_i}(\mat{x})^t \cdot (\mat{x}-\mat{p})\right\|.  \label{eq:unique-1} 
\end{eqnarray}
Consider the first term in \eqref{eq:unique-1}.  
By Lemma~\ref{cor:gradient}, 
$\norm{\mat{J}_{f_{\mat{v}_i}}(\mat{x})} = O(\kappa m^{3})$.
For any $\mat{p} \not\in B(\mat{x},m\gamma)$,
$\omega(\mat{x},\mat{p})$ vanishes.  If $\mat{p} \in B(\mat{x},m\gamma)$, then
\begin{equation}
\norm{\mat{J}_{f_{\mat{v}_i}}(\mat{x})} \cdot \norm{\mat{x}-\mat{p}} =
O(\kappa m^{4}\gamma).  
\label{eq:unique}
\end{equation}
Therefore,
\begin{equation}
\sum_{\mat{p} \in P} \left(\omega(\mat{x},\mat{p}) +
\omega(\mat{x},\mat{p}) \cdot \norm{\mat{J}_{f_{\mat{v}_i}}(\mat{x})} \cdot 
\norm{\mat{x}-\mat{p}} \right)  
\leq 1 + O(\kappa m^{4}\gamma). \label{eq:unique-2}
\end{equation}
Consider the second term in \eqref{eq:unique-1}.  For any point $\mat{p} \not\in B(\mat{x},m\gamma)$, $\nabla \omega(\mat{x},\mat{p})$ is a zero vector.  If $\mat{p} \in B(\mat{x},m\gamma)$, then $\norm{\mat{p}-\nu(\mat{x})} \leq \norm{\mat{p}-\mat{x}} + \norm{\mat{x}-\nu(\mat{x})} \leq m\gamma +
t \eps^{\tau} = O(m\gamma)$.  By Lemma~\ref{lem:basic}(i), $\mat{p} -
\nu(\mat{x})$ makes an angle $\pi/2 - O(m\gamma)$ with $N_{\nu(\mat{x})}$.
It follows from Lemma~\ref{lemma::normal_angle} that $\mat{p} - \nu(\mat{x})$
makes an angle $\pi/2 - O(m\sqrt{m}\,\gamma)$ with $L_\mat{x}$.  Therefore,
the projection of $\mat{p} - \nu(\mat{x})$ onto $L_\mat{x}$ has length
less than $O(m\sqrt{m}\,\gamma) \cdot O(m\gamma) = O(m^{5/2}\gamma^2)$.  Since 
$f_{\mat{v}_i}(\mat{x})$ is a unit vector in $L_\mat{x}$, the
projection $\mat{p} -\nu(\mat{x})$ in $L_\mat{x}$ has
length at least $\left|f_{\mat{v}_i}(\mat{x})^t \cdot (\mat{p}-\nu(\mat{x}))\right|
\geq \left|f_{\mat{v}_i}(\mat{x})^t \cdot (\mat{p}-\mat{x})\right| - \norm{\mat{x}-\nu(\mat{x})}$.
Therefore, 
\begin{equation}
\left| f_{\mat{v}_i}(\mat{x})^t \cdot (\mat{x}-\mat{p})\right| 
\leq \norm{\mat{x}-\nu(\mat{x})} + O(m^{5/2}\gamma^2)
\leq t\eps^{\tau} + O(m^{5/2}\gamma^2). \label{eq:unique-0}
\end{equation}
We conclude that
\begin{equation}
\left\|\sum_{\mat{p} \in P}
\nabla \omega(\mat{x},\mat{p}) \cdot
f_{\mat{v}_i}(\mat{x})^t \cdot (\mat{x}-\mat{p})\right\|
\leq
O(t\eps^{\tau} + m^{5/2}\gamma^2) \cdot \sum_{\mat{p} \in P} \left\|\nabla\omega(\mat{x},\mat{p})\right\|.
\label{eq:unique-3}
\end{equation}
Since 
\[
\nabla\omega(\mat{x},\mat{p}) =
\frac{\left(\sum_{\mat{p} \in P} h(\norm{\mat{x}-\mat{p}})\right) 
\frac{\mathrm{d}h(\norm{\mat{x}-\mat{p}})}{\mathrm{d}
\norm{\mat{x}-\mat{p}}} \cdot \frac{\mat{x}-\mat{p}}{\norm{\mat{x}-\mat{p}}} 
- 
h(\norm{\mat{x}-\mat{p}})
\sum_{\mat{p} \in P}
\frac{\mathrm{d}h(\norm{\mat{x}-\mat{p}})}{\mathrm{d}
\norm{\mat{x}-\mat{p}}} \cdot \frac{\mat{x}-\mat{p}}{\norm{\mat{x}-\mat{p}}}}
{\left(\sum_{\mat{p} \in P} h(\norm{\mat{x}-\mat{p}})\right)^2},
\]
we obtain 
\[
\sum_{\mat{p} \in P} \norm{\nabla\omega(\mat{x},\mat{p})} \leq \frac{2\,
\sum_{\mat{p} \in P} |\mathrm{d}h(\norm{\mat{x}-\mat{p}})/\mathrm{d}
\norm{\mat{x}-\mat{p}}|} {\sum_{\mat{p} \in P}
h(\norm{\mat{x}-\mat{p}})}.  
\]

By Lemma~\ref{lemma::in_out_ratio}(i), differentiating
$h(\norm{\mat{x}-\mat{p}})$ with respect to $\norm{\mat{x}-\mat{p}}$ gives
\[
\left|\frac{\mathrm{d}h(\norm{\mat{x}-\mat{p}})}{\mathrm{d}
\norm{\mat{x}-\mat{p}}}\right| \leq O\left(\frac{m}{\gamma}\right) \cdot 
\left(1- \frac{\norm{\mat{x}-\mat{p}}}{m\gamma}\right)^{2m-1}.  
\]
On the other hand, 
\[
\sum_{\mat{p} \in P} h(\norm{\mat{x}-\mat{p}}) =
\sum_{\mat{p} \in P} \left(1-\frac{\norm{\mat{x}-\mat{p}}}{m\gamma}
\right)^{2m}\left(\frac{2\norm{\mat{x}-\mat{p}}}{\gamma}+1\right).  
\]
For all $\mat{p} \in P \setminus B(\mat{x},m\gamma)$,
$h(\norm{\mat{x}-\mat{p}}) = 0$ and $\left|\frac{\mathrm{d}h(\norm{\mat{x}-\mat{p}})}{\mathrm{d}
	\norm{\mat{x}-\mat{p}}}\right| = 0$.  Then,
\begin{eqnarray*}
	\sum_{\mat{p} \in P} \norm{\nabla\omega(\mat{x},\mat{p})}
	& \leq &  \frac{O\left(\frac{m}{\gamma}\right)  \sum_{\mat{p} \in
		P \cap B(\mat{x},m\gamma)}
	\left(1-\frac{\norm{\mat{x}-\mat{p}}}{m\gamma}\right)^{2m-1}}
{\sum_{\mat{p} \in P \cap B(\mat{x},m\gamma)}
	\left(1-\frac{\norm{\mat{x}-\mat{p}}}{m\gamma}\right)^{2m}
	\left(\frac{2\norm{\mat{x}-\mat{p}}}{\gamma} + 1\right)}.
\end{eqnarray*}
Let $r = \sqrt{m}\eps/3$.  By Lemma~\ref{lem:center}, 
\begin{eqnarray*}
	\sum_{\mat{p} \in P} \norm{\nabla\omega(\mat{x},\mat{p})}
	& \leq & 
\frac{O\left(\frac{\kappa m}{\gamma}\right) \cdot \sum_{\mat{p} \in
		P \cap B(\mat{x},m\gamma - r)}
	\left(1-\frac{\norm{\mat{x}-\mat{p}}}{m\gamma}\right)^{2m-1}}
{\sum_{\mat{p} \in P \cap B(\mat{x},m\gamma -r)}
	\left(1-\frac{\norm{\mat{x}-\mat{p}}}{m\gamma}\right)^{2m}
	\left(\frac{2\norm{\mat{x}-\mat{p}}}{\gamma} + 1\right)}.
\end{eqnarray*}
In the denominator, the term $\left(1-\frac{\norm{\mat{x}-\mat{p}}}{m\gamma}\right)
\left(\frac{2\norm{\mat{x}-\mat{p}}}{\gamma} + 1\right)$ achieves its minimum
 $\frac{2\sqrt{m}\eps}{3\gamma} - \frac{2\eps^2}{9\gamma^2} + \frac{\eps}{3\sqrt{m}\gamma} = \Omega(\sqrt{m})$ when $\norm{\mat{x}-\mat{p}} = m\gamma - r$.  
 It follows that 
\begin{equation}
\sum_{\mat{p} \in P} \left\|\nabla\omega(\mat{x},\mat{p})\right\| = O(\kappa\sqrt{m}/\gamma).
\label{eq:unique-4}
\end{equation}

Substituting \eqref{eq:unique-4} into \eqref{eq:unique-3} gives
\begin{equation}
\left\|\sum_{\mat{p} \in P}
\nabla \omega(\mat{x},\mat{p}) \cdot
f_{\mat{v}_i}(\mat{x})^t \cdot (\mat{x}-\mat{p})\right\|
= O(t\kappa \sqrt{m}\eps^{\tau-1} + \kappa m^{3}\gamma).
\label{eq:unique-5}
\end{equation}
By substituting \eqref{eq:unique-2} and \eqref{eq:unique-5} into \eqref{eq:unique-1},
we have
\[
\left\|\nabla \varrho_{\mat{z},i}(\mat{x})\right\| \leq 1 + O(t \kappa \sqrt{m}\eps^{\tau-1} + \kappa m^{4}\gamma), 
\]
establishing the upper range limit for  
$\left\|\nabla \varrho_{\mat{z},i}(\mat{x})\right\|$.
Symmetrically,
\begin{eqnarray*}
\norm{\nabla \varrho_{\mat{z},i}(\mat{x})} 
& \geq & \sum_{\mat{p} \in P} \left(\omega(\mat{x},\mat{p}) -
\omega(\mat{x},\mat{p}) \cdot \norm{\mat{J}_{f_{\mat{v}_i}}(\mat{x})} \cdot 
\norm{\mat{x}-\mat{p}} \right) - \\
& & \left\|\sum_{\mat{p} \in P}
\nabla \omega(\mat{x},\mat{p}) \cdot
f_{\mat{v}_i}(\mat{x})^t \cdot (\mat{x}-\mat{p})\right\|.  
\end{eqnarray*}
By \eqref{eq:unique} and \eqref{eq:unique-5}, we have 
\[
\norm{\nabla \varrho_{\mat{z},i}(\mat{x})} \geq 1 -
O(\kappa m^{4}\gamma) - O(t \kappa \sqrt{m}\eps^{\tau-1} + \kappa m^{3}\gamma) 
= 1 - O(t \kappa \sqrt{m}\eps^{\tau-1} + \kappa m^{4}\gamma), 
\]
establishing the lower range limit for
$\left\|\nabla \varrho_{\mat{z},i}(\mat{x})\right\|$.

Observe that 
\begin{eqnarray*}
\mat{v}_i^t \cdot \nabla\varrho_{\mat{z},i}(\mat{x}) 
& = &
\sum_{\mat{p} \in P} \omega(\mat{x},\mat{p}) \cdot
\mat{v}_i^t \cdot f_{\mat{v}_i}(\mat{x}) + 
\sum_{\mat{p} \in P} \omega(\mat{x},\mat{p})
\cdot \mat{v}_i^t \cdot \mat{J}_{f_{\mat{v}_i}}(\mat{x})^t \cdot
(\mat{x}-\mat{p}) +  \\
& & \sum_{\mat{p} \in P} \mat{v}_i^t \cdot
\nabla\omega(\mat{x},\mat{p}) \cdot f_{\mat{v}_i}(\mat{x})^t \cdot (\mat{x}-\mat{p}).
\end{eqnarray*}
Therefore, 
\begin{eqnarray*}
\mat{v}_i^t \cdot \nabla\varrho_{\mat{z},i}(\mat{x}) & \geq & \sum_{\mat{p} \in
P} \omega(\mat{x},\mat{p}) \cdot \mat{v}_i^t \cdot f_{\mat{v}_i}(\mat{x}) -  
\left|\sum_{\mat{p} \in P} \omega(\mat{x},\mat{p}) \cdot \mat{v}_i^t \cdot
J_{f_{\mat{v}_i}}(\mat{x})^t \cdot (\mat{x}-\mat{p})\right| - \\ 
& & \left|\sum_{\mat{p} \in P}
\mat{v}_i^t \cdot \nabla\omega(\mat{x},\mat{p}) \cdot
f_{\mat{v}_i}(\mat{x})^t \cdot (\mat{x}-\mat{p}) \right|.  
\end{eqnarray*}
Since $\angle (f_{\mat{v}_i}(\mat{x}),\mat{v}_i)$ is $O(m\sqrt{m}\,\gamma)$ by Lemma~\ref{lemma::normal_angle}, we get 
$\mat{v}_i^t \cdot f_{\mat{v}_i}(\mat{x}) \geq 1 - O(m^3\gamma^2)$, which implies that
$\sum_{\mat{p} \in P} \omega(\mat{x},\mat{p}) \cdot
\mat{v}_i^t \cdot f_{\mat{v}_i}(\mat{x}) \geq 1-O(m^3\gamma^2)$.  The second term is at
most $\sum_{\mat{p} \in P} \omega(\mat{x},\mat{p}) \cdot
\norm{\mat{J}_{f_{\mat{v}_i}}(\mat{x})} \cdot \norm{\mat{x}-\mat{p}} \leq O(\kappa m^{4}\gamma)$
by \eqref{eq:unique}.
The third term is at most $\sum_{\mat{p} \in P}
\norm{\nabla\omega(\mat{x},\mat{p})} \cdot
|f_{\mat{v}_i}(\mat{x})^t \cdot (\mat{x}-\mat{p})|$, which is 
$O(t \kappa \sqrt{m}\eps^{\tau-1} + \kappa m^{3}\gamma)$ by \eqref{eq:unique-0} and \eqref{eq:unique-4}. 
As a result, 
$\mat{v}_i^t \cdot \nabla
\varrho_{\mat{z},i}(\mat{x}) \geq 1 - O(t \kappa \sqrt{m}\eps^{\tau-1} + \kappa m^{4}\gamma)$.
\end{proof}

The next result shows that every point $z$ in $\mani$ is near $Z_{\varrho_\mat{z}}$.

\begin{lemma}\label{lemma::exist}
Let $\varrho_\mat{z}$ be the canonical function with respect to a point
$\mat{z} \in \mani$ and an orthonormal basis $\{\mat{v}_1,\ldots,
\mat{v}_{d-m}\}$ of $N_\mat{z}$.  
There exists $\eps_0 \in (0,1)$ and $c_m \geq 1$ such that if $\eps \leq \eps_0$, then $Z_{\varrho_\mat{z}} \cap B(\mat{z}, c_m\gamma^{2}) \cap (\mat{z} + N_\mat{z}) \not= \emptyset$ and $Z_{\varrho_\mat{z}} \cap ( B(\mat{z}, 2\eps) \setminus B(\mat{z}, c_m\gamma^{2})) \cap (\mat{z} + N_\mat{z}) = \emptyset$.  The value $\eps_0$ decreases as $d$ increases, and $c_m$ is linear in $m^{5/2}$.
\end{lemma}
\begin{proof}
We first show that $Z_{\varrho_\mat{z}} \cap ( B(\mat{z}, 2\eps) \setminus B(\mat{z}, c_m\gamma^{2})) \cap (\mat{z} + N_\mat{z})$ is empty.
For all $i \in [1,d-m]$ and all point $\mat{x} \in B(\mat{z},2\eps)$, let
$\varrho_{\mat{z},i} = \sum_{\mat{p} \in P} \omega(\mat{x},\mat{p}) \cdot
f_{\mat{v}_i}(\mat{x})^t \cdot (\mat{x}-\mat{p})$.

We claim that there exists a value $c_m \geq 1$ that is linear in $m^{5/2}$ such that for every $\mat{x} \in B(\mat{z}, 2\eps) \cap (\mat{z} +
N_\mat{z})$ and every $i \in [1,d-m]$, if $\mat{v}_i^t \cdot (\mat{x}-\mat{z})
\geq c_m\gamma^{2}$, then $\varrho_{\mat{z},i}(\mat{x}) > 0$.  We ignore all $\mat{p} \in P \setminus
B(\mat{x},m\gamma)$ because $\omega(\mat{x},\mat{p}) = 0$ in this case, so such
points have no influence over $\varrho_\mat{z}(\mat{x})$.  $P \cap B(\mat{x},m\gamma)$
is non-empty because, by uniform $(\eps,\kappa)$-sampling, there is a point $\mat{q} \in P$
such that $\norm{\mat{q}-\mat{z}} \leq \eps$ which implies that $\norm{\mat{q}-\mat{x}} 
\leq \norm{\mat{x}-\mat{z}} + \norm{\mat{q}-\mat{z}} \leq 3\eps \leq m\gamma$.
For every $\mat{p} \in P \cap B(\mat{x},m\gamma)$, 
\[
\mat{v}_i^t \cdot (\mat{x}-\mat{p}) \geq \mat{v}_i^t
\cdot (\mat{x}-\mat{z}) - |\mat{v}_i^t \cdot (\mat{z}-\mat{p})|.  
\]
The first term is bounded from below as
$\mat{v}_i^t \cdot (\mat{x}-\mat{z}) \geq
c_m\gamma^2$ by assumption.  Consider the
second term.  Since $\norm{\mat{p}-\mat{z}} \leq
\norm{\mat{p}-\mat{x}} + \norm{\mat{x}-\mat{z}} \leq m\gamma+ 2\eps <
(m+1)\gamma$, Lemma~\ref{lem:basic}(i) implies that 
the second
term $|\mat{v}_i^t \cdot (\mat{z}-\mat{p})|$ is at most $(m+1)^2\gamma^2/2$.
It
follows that 
\[
\mat{v}_i^t \cdot (\mat{x}-\mat{p}) \geq
c_m \gamma^{2} - (m+1)^2\gamma^2/2.
\]
For $i \in [1,d-m]$, define $\mat{h}_i(\mat{x}) = f_{\mat{v}_i}(\mat{x}) - \mat{v}_i$.
Lemma~\ref{lemma::normal_angle} implies that 
\[
\norm{\mat{h}_i(\mat{x})} \leq 2\sin\frac{\angle (L_\mat{x},N_\mat{z})}{2} = O(m\sqrt{m}\,\gamma).
\]
Observe that 
\begin{eqnarray*}
f_{\mat{v}_i}(\mat{x})^t \cdot (\mat{x}-\mat{p}) 
& = & \mat{v}_i^t \cdot (\mat{x}-\mat{p}) + \mat{h}_i(\mat{x})^t \cdot (\mat{x}-\mat{p}) \\
& \geq & c_m \gamma^{2}  - (m+1)^2\gamma^2/2 -
\norm{\mat{h}_i(\mat{x})} \cdot \norm{\mat{x}-\mat{p}} \\
& \geq & c_m \gamma^{2}  - (m+1)^2\gamma^2/2 - O(m^{5/2}\gamma^2) \\
& > & 0,
\end{eqnarray*}
whenever $c_m$ is a large enough value that is linear in $m^{5/2}$.
As a result, $\varrho_{\mat{z},i}(\mat{x}) > 0$.  
This proves our claim.  

We can symmetrically show that if $\mat{v}_i^t \cdot (\mat{x}-\mat{z}) \leq
-c_m\gamma^2$, then $\varrho_{\mat{z},i}(\mat{x}) < 0$.  Thus,
$\varrho_{\mat{z},i}^{-1}(0) \cap B(\mat{z},2\eps) \cap (\mat{z} +
N_\mat{z})$ lies in a $(d-m)$-dimensional slab $S_{\mat{v}_i} \subset \mat{z} +
N_{\mat{z}}$ that is bounded by two $(d-m-1)$-dimensional flats orthogonal to
$\mat{v}_i$ and at distance $c_m\gamma^2$ from $\mat{z}$.  It follows that $(Z_{\varrho_\mat{z}} \cap ( B(\mat{z}, 2\eps) \cap (\mat{z} + N_\mat{z})) \setminus S_{\mat{v}_i} = \emptyset$. By Lemma~\ref{lem:agree}, $Z_{\varrho_\mat{z}}$ is identical for any choice of the orthonormal basis $\{\mat{v}_1,\ldots, \mat{v}_{d-m}\}$ of $N_\mat{z}$.  It means that we can set $\mat{v}_i$ to be any unit vector $\mat{v} \in N_\mat{z}$ and the proof above still works.  Observe that $\bigcap_{\mat{v} \in N_\mat{z}} S_\mat{v} = B(\mat{z},c_m\gamma^2) \cap (\mat{z} + N_\mat{z})$.  Hence, $Z_{\varrho_\mat{z}} \cap ( B(\mat{z}, 2\eps) \setminus B(\mat{z}, c_m\gamma^2)) \cap (\mat{z} + N_\mat{z}) = \emptyset$.

To establish that $Z_{\varrho_\mat{z}} \cap B(\mat{z}, c_m\gamma^{2}) \cap (\mat{z} + N_\mat{z}) \not= \emptyset$, it suffices to show that $\bigcap_{i=1}^{d-m} \varrho_{\mat{z},i}^{-1}(0)$
contains a point in $\bigcap_{i=1}^{d-m} S_{\mat{v}_i}$.  This is because $\bigcap_{i=1}^{d-m} S_{\mat{v}_i}$ is contained in $B(\mat{z},c_m\sqrt{d-m}\gamma^2)$, and for $\eps_0 \leq 1/(16c_m\sqrt{d-m})$, we have $B(\mat{z},c_m\sqrt{d-m}\gamma^2) \subseteq B(\mat{z},\eps)$ as $c_m\sqrt{d-m}\gamma^2 \leq 16c_m\sqrt{d-m}\,\eps^2 \leq 16c_m\sqrt{d-m}\, \eps_0\eps$.  Then, the fact that $Z_{\varrho_{\mat{z}}} \cap (B(\mat{z},2\eps) \setminus B(\mat{z},c_m\gamma^2)) \cap (\mat{z} + N_\mat{z}) = \emptyset$ implies that $\bigcap_{i=1}^{d-m} \varrho^{-1}_{\mat{z},i}(0)$ contains a point in $B(\mat{z},c_m\gamma^2) \cap (\mat{z} + N_\mat{z})$.

In fact, we choose an even smaller $\eps_0$ such that $\sqrt{\eps_0} \leq 1/(16c_m\sqrt{d-m})$, which gives $c_m\sqrt{d-m} \gamma^2 \leq \eps^{3/2}$.  This will allow us apply Lemma~\ref{lemma::unique} later.  The exponent 3/2 is an arbitrary choice.  Any number greater than 1 will do.

Let $C = \bigcap_{i=1}^{d-m} S_{\mat{v}_i}$.  It is a $(d-m)$-dimensional cube that lies in $\mat{z} + N_{\mat{z}}$, has $\mat{z}$ as its center, and has side
length $2c_m\gamma^{2} $. The facets of $C$ are orthogonal to the
directions $\mat{v}_1,\ldots,\mat{v}_{d-m}$.  

Adopt a coordinate frame such that $\mat{v}_1,\ldots,\mat{v}_{d-m}$ are the
first $d-m$ coordinate axes of $\real^d$.  For $i \in [1,d-m]$, define $H_i$ to
be the set of maximal line segments that lie inside $C$ and are parallel to the
direction $\mat{v}_i$.

First, we claim that every line segment $l \in H_i$ intersects
$\varrho_{\mat{z},i}^{-1}(0)$ at exactly one point.  We have shown earlier that
$\varrho_{\mat{z},i}$ has opposite signs at the endpoints of $l$.  So $l \cap
\varrho_{\mat{z},i}^{-1}(\mat{0}) \not= \emptyset$.  Suppose to the contrary
that $l \cap \varrho_{\mat{z},i}^{-1}(0)$ contains two distinct points
$\mat{y}_1$ and $\mat{y}_2$.  So $\mat{y}_1 - \mat{y}_2$ is parallel to
$\mat{v}_i$.  Assume without loss of generality that
$\mat{y}_1 - \mat{y}_2$ has the same orientation as $\mat{v}_i$. 
By Lemma~\ref{lemma::unique}, $(\mat{y}_1-\mat{y}_2)^t \cdot
\nabla\varrho_{\mat{z},i}(\mat{x}) > 0$ for every $\mat{x} \in
B(\mat{z},c_m\sqrt{d-m}\gamma^2) \subseteq B(\mat{z},\eps^{3/2})$. But then $\varrho_{\mat{z},i}(\mat{x})$ increases
strictly monotonically from $\mat{y}_2$ to $\mat{y}_1$, which implies that
$\varrho_{\mat{z},i}(\mat{y}_1) > 0$. This is a contradiction because
$\mat{y}_1 \in \varrho_{\mat{z},i}^{-1}(0)$, thereby establishing our claim.

Define a function $g_i : C \rightarrow [-c_m\gamma^{2},
c_m\gamma^{2}]$ such that 
$g_i(\mat{x}) = b_{i,\mat{x}}$, where 
\begin{itemize}

\item $(x_1,\ldots,x_{i-1},b_{i,\mat{x}},x_{i+1},\ldots,x_d) \in C$ and

\item $\varrho_{\mat{z},i}(x_1,\ldots,x_{i-1},b_{i,\mat{x}},x_{i+1},\ldots,x_d) = 0$.

\end{itemize} 
Our claim in the previous paragraph ensures the existence and uniqueness of
$b_{i,\mat{x}}$.  We show that $g_i$ is continuous.  Since $\varrho_{\mat{z},i}$ is
continuous, $\varrho_{\mat{z},i}^{-1}(0)$ is compact~\cite[Ch~3:~Theorem
5.4,~Ch~5:~Theorem 2.11]{mendelson}, which implies that for any interval $[a,b]
\subset \real$, $\varrho_{\mat{z},i}^{-1}(0) \cap \{\mat{x} \in C : x_i \in
[a,b]\}$ is compact.  Let $\pi_i$ be the function that projects points in $C$
onto the linear subspace spanned by
$\{\mat{v}_1,\ldots,\mat{v}_{i-1},\mat{v}_{i+1},\ldots,\mat{v}_{d-m}\}$.  Since
$\pi_i$ is continuous, its image is compact and so is the following
product~\cite[Ch~5:~Theorem 2.9 \& Theorem 4.2]{mendelson}:
\[
\pi_i\left(\varrho_{\mat{z},i}^{-1}(0) \cap \{\mat{x}
\in C : x_i \in [a,b]\}\right) \times [-c_m\gamma^{2},
c_m\gamma^{2}].
\]
Observe that this product is homeomorphic to $g_i^{-1}([a,b])$.
Therefore, $g_i^{-1}([a,b])$ is compact for any
interval $[a,b] \subset \real$, which implies that $g_i$ is
continuous~\cite[Ch~2:~Theorem 6.10]{mendelson}.

Define a function $g : C \rightarrow C$ such that 
\[
g(\mat{x}) = \left(g_1(\mat{x}), \ldots, g_{d-m}(\mat{x})\right)^t. 
\]
The function $g$ is continuous as each $g_i$ is continuous.
Notice that $\varrho_{\mat{z},i}^{-1}(0) \cap C$ is the subset of $C$ that satisfy
the equation $g_i(x_1,\ldots,x_i,\ldots,x_d) = x_i$.  Since
$\varrho_\mat{z}(\mat{x}) = (\varrho_{\mat{z},1}(\mat{x}), \ldots,
\varrho_{\mat{z},d-m}(\mat{x}))^t$, we conclude that $Z_{\varrho_\mat{z}} \cap
C$ is the subset of $C$ that satisfy the equation $g(\mat{x}) =
\mat{x}$.  By the Brouwer fixed-point theorem~\cite[Ch~4:~Theorem 4.6]{mendelson}, there is indeed such a point in $C$.
\end{proof}

Recall that $\nu$ is the map that sends every point in $\real^d$ to its nearest
point in $\mani$.  We need to show that $Z_\varphi \cap \widehat{\mani}$ is compact
in order to prove that $Z_\varphi \cap \widehat{\mani}$ and $\mani$ are homeomorphic.  

\begin{lemma}\label{lemma::compact}
$Z_{\varphi} \cap \widehat{\mani}$ is compact.
\end{lemma}
\begin{proof}
By Lemmas~\ref{lem:agree} and~\ref{lemma:local}, for any point $\mat{z} \in \mani$,
$Z_{\varphi}$ agrees locally with $Z_{\varrho_\mat{z}}$ where $\varrho_{\mat{z}}$
is the canonical function with respect to $\mat{z}$ and any
orthonormal basis of $N_\mat{z}$.  Our strategy is to
construct a finite number of such $Z_{\varrho_{\mat{z}}}$'s and prove that each is
compact.  The lemma then follows as a finite union of compact sets is compact. 

Take a maximal set $Y$ of points in $\widehat{\mani}$ such that any two of them are at distance $\eps^{\tau}$ or more apart. 
It implies that any two balls centered at points in $Y$ with radius $\eps^{\tau}/2$ are interior-disjoint. 
Since $\widehat{\mani}$ is the
product of $\mani$ and a ball of radius $\eps$, 
$\widehat{\mani}$ is compact~\cite[Ch~5:~Theorem~4.2]{mendelson}.  
It follows that $|Y|$ is finite.  The maximality 
also implies that $\widehat{\mani} \subseteq \bigcup_{\mat{y} \in Y}
B(\mat{y},\eps^{\tau})$.
The intersection $Z_\varphi \cap \bigcup_{\mat{y} \in Y}
B(\mat{y},\eps^{\tau})$ is equal to $\bigcup_{\mat{y} \in Y} Z_\varphi \cap
B(\mat{y},\eps^{\tau})$ which is a subset of $\bigcup_{\mat{y} \in Y}
Z_\varphi \cap B(\nu(\mat{y}),\eps^{\tau} + \eps)$ because $\norm{\mat{y} -
\nu(\mat{y})} \leq \eps$.  By Lemmas~\ref{lem:agree} and~\ref{lemma:local},
$Z_\varphi \cap B(\nu(y), \eps^\tau + \eps) = Z_{\varrho_{\nu(y)}} \cap B(\nu(y), \eps^\tau + \eps)$.
Therefore,
\[
Z_\varphi \cap \widehat{\mani} \subseteq Z_\varphi \cap \bigcup_{\mat{y} \in Y}
B(\mat{y},\eps^\tau) \subseteq
\bigcup_{\mat{y} \in Y} Z_{\varrho_{\nu(y)}} \cap B(\nu(y),\eps^\tau + \eps).
\]
As $\varrho_{\nu(\mat{y})}$ is continuous in the interior of $B(\nu(y),2\eps)$
by Lemma~\ref{lemma:local}, $Z_{\varrho_{\nu(y)}} \cap B(\nu(y),\eps^{\tau}+\eps)$
is compact~\cite[Ch~3:~Theorem 5.4,~Ch~5:~Theorem 2.11]{mendelson}.  
It implies that the finite union $\bigcup_{\mat{y} \in Y} Z_{\varrho_{\nu(y)}}
\cap B(\nu(y),\eps^{\tau} + \eps)$ is also compact.  Finally, observe that 
\[
Z_\varphi \cap \widehat{\mani} = \left(\bigcup_{\mat{y} \in Y}
Z_{\varrho_{\nu(y)}} \cap B(\nu(y),\eps^{\tau}+\eps)\right) \cap \widehat{\mani},
\]
which is compact because it is the intersection of two 
compact subsets in $\real^d$.
\end{proof}

We are ready to prove the faithful approximation of $\mani$ by $Z_\varphi \cap
\widehat{\mani}$. 

%
%
%

\begin{theorem}
	\label{thm:main}
	Let $\mani$ be an $m$-dimensional compact smooth manifold in $\real^d$.  
	Let $P$ be a uniform $(\eps,\kappa)$-sample of $\mani$ for some constant $\kappa \geq 1$. 
	We assume that $\mani$ has unit reach, $m$ is known,
	a neighborhood radius $\gamma = 4\eps$, and approximate tangent spaces 
	with angular errors
	at most $m\gamma$ are specified at the points in $P$.
	Let $\widehat{\mani}$ be the set of
	points within a distance $\eps$ from $\mani$.  We can construct a function $\varphi: \real^d \rightarrow \real^{d-m}$ 
	for which there exists $\eps_0 \in (0,1)$ that decreases as $d$ increases such that 
	the following properties hold whenever $\eps \leq \eps_0$.
	\begin{emromani}
		
		\item The restriction of the nearest point map to $Z_\varphi \cap
		\widehat{\mani}$ is a homeomorphism between $Z_\varphi \cap \widehat{\mani}$ and $\mani$.
		
		\item The Hausdorff distance between $Z_\varphi \cap \widehat{\mani}$ and $\mani$
		is $O(m^{5/2}\gamma^{2}) = O(m^{5/2}\eps^2)$.
		
		\item For all $\mat{x} \in Z_\varphi \cap \widehat{\mani}$, $N_{\nu(\mat{x})}$
		makes an $O(m^2\sqrt{\kappa\gamma}) = O(m^2\sqrt{\kappa\eps})$ angle with the normal space of $Z_\varphi$ at
		$\mat{x}$.
		
	\end{emromani}
\end{theorem}

\begin{proof}
Consider (i).  Let $\mu$ denote the restriction of $\nu$ to $Z_\varphi \cap
\widehat{\mani}$.  First, we show that $\mu$ is injective.  Suppose to the
contrary that there are two points $\mat{y}_1, \mat{y}_2 \in Z_\varphi \cap
\widehat{\mani}$ such that $\mu(\mat{y}_1)$ and $\mu(\mat{y}_2)$ are the same
point $\mat{z} \in \mani$. Then, $\mat{y}_1$ and $\mat{y}_2$ belong to
$\mat{z} + N_\mat{z}$, which implies that $\mat{y}_1-\mat{y}_2 \in N_\mat{z}$.  
Note that $\mat{y}_1$ and $\mat{y}_2$ lie in $B(\mat{z},\eps)$.  By Lemmas~\ref{lem:agree} and~\ref{lemma:local}, $Z_\varphi \cap B(\mat{z},\eps) = Z_{\varrho_{\mat{z}}} \cap B(\mat{z},\eps)$.  Then, Lemma~\ref{lemma::exist} implies that $\mat{y}_1$ and $\mat{y}_2$ belong to $B(\mat{z}, t\gamma^2)$ for some large enough $t$ that is linear in $m^{5/2}$. By Lemma~\ref{lemma::unique}, we can define $\mat{v}_1 = \mat{y}_1-\mat{y}_2$
and get $(\mat{y}_1-\mat{y}_2)^t \cdot \nabla
\varrho_{\mat{z},1}(\mat{x}) >0$ for all $\mat{x} \in B(\mat{z}, t\gamma^2)$ when $\eps_0$ is sufficiently small.
But then $\varrho_{\mat{z},1}(\mat{x})$ increases strictly monotonically from
$\mat{y}_2$ to $\mat{y}_1$, which implies that $\varrho_{\mat{z},1}(\mat{y}_1)
> 0$.  This is a contradiction because $\mat{y}_1$ belongs to $Z_\varphi$ and
hence $Z_{\varrho_\mat{z}}$ by Lemmas~\ref{lem:agree}
and~\ref{lemma:local}.  This proves that $\mu$ is injective.

Next, we show that $\mu$ is surjective. Let $\mat{z}$ be any point in $\mani$.
It follows from Lemmas~\ref{lem:agree},~\ref{lemma:local}, and~\ref{lemma::exist} that there 
exists a point $\mat{y} \in
Z_\varphi \cap \widehat{\mani} \cap (\mat{z} + N_\mat{z})$.  We show
that $\mu$ must map $\mat{y}$ to $\mat{z}$.  Suppose that $\mu$ maps $\mat{y}$
to another point $\mat{z}_2 \in \mani$, i.e.  $\norm{\mat{y}-\mat{z}_2} <
\norm{\mat{y} - \mat{z}}$.  We grow a ball $B$ tangent to $\mani$ at
$\mat{z}$ by moving its center linearly from $\mat{z}$ towards $\mat{y}$.
When $B$ is tiny, it touches $\mani$ only at $\mat{z}$.  When the center of
$B$ reaches $\mat{y}$, $B$ contains both $\mat{z}$ and $\mat{z}_2$.  Thus, the
radius of the growing $B$ 
must become the local feature size of $\mani$ at $\mat{z}$ before or when its
center reaches $\mat{y}$.  Recall that the reach of $\mani$ is assumed
to be 1.  Thus, $\norm{\mat{y}-\mat{z}} \geq 1 >
\eps$.  This contradicts the fact that $\mat{y} \in \widehat{\mani} \cap
(\mat{z} + N_\mat{z})$, thereby proving that $\mu$ is surjective.

Since $Z_\varphi \cap \widehat{\mani}$ avoids the medial axis, the restriction $\mu$
is continuous.  Therefore, $\mu$ is a continuous bijection from $Z_\varphi \cap
\widehat{\mani}$ to $\mani$.  The spaces $\mani$ and $Z_\varphi \cap \widehat{\mani}$ are
compact by assumption and Lemma~\ref{lemma::compact}, respectively, so we
conclude from the existence of $\mu$ that $\mani$ and $Z_\varphi \cap
\widehat{\mani}$ are homeomorphic~\cite[Ch~5:~Theorem 2.14]{mendelson}.  This 
proves the correctness of (i).

Consider (ii). By Lemmas~\ref{lem:agree},~\ref{lemma:local}, and~\ref{lemma::exist}, for any point $\mat{z} \in \mani$, there exists a point $\mat{x} \in Z_\varphi$ within a distance of $c_m\gamma^2$, where $c_m \geq 1$ is some value linear in $m^{5/2}$.  Therefore, $c_m\gamma^2 = O(m^{5/2}\eps_0\eps) < \eps$ for a small enough $\eps_0$.  So $\mat{x} \in Z_\varphi \cap \widehat{\mani}$.  It follows that the directed Hausdorff distance from $\mani$ to $Z_\varphi \cap \widehat{\mani}$ is $O(m^{5/2}\gamma^2)$.  Conversely, for any point $\mat{x} \in Z_\varphi \cap \widehat{\mani}$, $\norm{\nu(\mat{x}) - \mat{x}} \leq \eps$ and $x \in \nu(x) + N_{\nu(\mat{x})}$. By Lemmas~\ref{lem:agree},~\ref{lemma:local}, and~\ref{lemma::exist}, $Z_\varphi \cap (B(\nu(\mat{x}), 2\eps) \setminus B(\nu(\mat{x}),c_m\gamma^2)) \cap (\nu(x) + N_{\nu(\mat{x})})$ is empty.  So $\norm{\nu(\mat{x}) - \mat{x}} \leq c_m\gamma^2 = O(m^{5/2}\gamma^2)$.  It follows that the directed Hausdorff distance from $Z_\varphi \cap \widehat{\mani}$ to $\mani$ is $O(m^{5/2}\gamma^2)$.

Consider (iii). By Lemma~\ref{lemma::unique}, for every point
$\mat{x} \in Z_\varphi \cap \widehat{\mani}$ and every unit vector $\mat{v}_1 \in N_{\nu(\mat{x})}$, $\norm{\nabla\varrho_{\nu(\mat{x}),1}(\mat{x})} \leq 1 + O(\kappa m^4\gamma)$ and $\mat{v}_1^t \cdot \nabla\varrho_{\nu(\mat{x}),1}(\mat{x}) \geq 1 - O(\kappa m^4 \gamma)$.  Thus,
\[
\angle (\mat{v}_1,
\nabla\varrho_{\nu(\mat{x}),1}(\mat{x})) \leq 
\arccos\left(\frac{\mat{v}_1^t \cdot \nabla\varrho_{\nu(\mat{x}),1}(\mat{x})}
{\norm{\nabla\varrho_{\nu(\mat{x}),1}(\mat{x})}}\right) 
\leq 
\arccos\left(\frac{1-O(\kappa m^4\gamma)} {1 + O(\kappa m^4\gamma)}\right) 
= 
O(m^2\sqrt{\kappa\gamma}).  
\]
The vector $\nabla\varrho_{\nu(\mat{x}),1}(\mat{x})$ belongs to the normal space of
$Z_\varphi$ at $\mat{x}$. (Recall that $Z_\varphi$ agrees with
$Z_{\varrho_{\nu(\mat{x})}}$ locally.)  Thus, the angle between
$N_{\nu(\mat{x})}$ and the normal space of $Z_\varphi$ at $\mat{x}$ is
$O(m^2\sqrt{\kappa\gamma})$.
\end{proof}

\section{Projection operator}

Our proof of convergence will make use of the property that
$\mat{B}_{\varphi,\mat{x}}$ is a $d \times (d-m)$ matrix with orthogonal unit
columns such that $\col{\mat{B}_{\varphi,\mat{x}}} = L_\mat{x}$.  Such a matrix
can be obtained by an eigen-decomposition of $\mat{C}_\mat{x}$.  

We rewrite $\varphi(\mat{x}) = \sum_{\mat{p} \in P} \omega(\mat{x},\mat{p})
\cdot \mat{B}_{\varphi,\mat{x}}^t \cdot (\mat{x}-\mat{p}) =
\mat{B}_{\varphi,\mat{x}}^t \cdot (\mat{x} - \mat{a}_\mat{x})$, where
$\mat{a}_\mat{x} = \sum_{\mat{p} \in P} \omega(\mat{x},\mat{p}) \cdot \mat{p}$.
Intuitively, as $\varphi(\mat{a}_\mat{x}) = \mat{0}$, we want to move the
current point $\mat{x}_i$ closer to $\mat{a}_\mat{x}$.  We also want to move
directly onto $Z_{\varphi}$ without much drifting.  Therefore, it is desirable
to move $\mat{x}_i$ within the affine subspace $\mat{x}_i + L_{\mat{x}_i}$
which is roughly normal to $Z_\varphi$.  The projection follows
an iterative scheme:
\[
\mat{x}_{i+1} = \mat{x}_i +
\mat{B}_{\varphi,\mat{x}_i}^{} \cdot \mat{B}_{\varphi,\mat{x}_i}^t \cdot (
\mat{a}_{\mat{x}_i} - \mat{x}_i). 
\]
Note that $\mat{B}_{\varphi,\mat{x}_i}^{} \cdot \mat{B}_{\varphi,\mat{x}_i}^t
\cdot ( \mat{a}_{\mat{x}_i} - \mat{x}_i)$ is the projection of the vector
$\mat{a}_{\mat{x}_i} - \mat{x}_i$ into $L_{\mat{x}_i}$.  The iterative scheme
moves the current point $\mat{x}_i$ by this projected vector to the new point
$\mat{x}_{i+1}$.  In other words, $\mat{x}_{i+1}$ is the projection
of $\mat{a}_{\mat{x}_i}$ onto the affine subspace
$\mat{x}_i + L_{\mat{x}_i}$.

We prove two technical results in order to establish the proof of convergence.
The first one shows that any initial point near $\mani$ is moved to within an
$O(m^{7/2}\gamma^2)$ distance from $\mani$ after a single iteration.  Let
$\tilde{\mat{x}}_i$ denote the nearest point in $Z_\varphi$ to $\mat{x}_i$.
The second result shows that $\norm{\mat{x}_{i+1} - \tilde{\mat{x}}_i} \ll
\norm{\mat{x}_i - \tilde{\mat{x}}_i}$, which implies that $\norm{\mat{x}_{i+1}
- \tilde{\mat{x}}_{i+1}} \ll \norm{\mat{x}_i - \tilde{\mat{x}}_i}$.

\begin{lemma}
\label{lem:first_itr}
Let $P$ be a uniform $(\eps,\kappa)$-sample of $\mani$.  
For every point $\mat{x}$
within a distance $m\gamma$ from $P$ and every $d \times (d-m)$ matrix
$\mat{B}_{\varphi,\mat{x}}$ that satisfies $\col{\mat{B}_{\varphi,\mat{x}}} =
L_\mat{x}$, we have $\norm{\mat{y} - \nu(\mat{x})} = O(m^{7/2}\gamma^2)$, where
$\mat{y} = \mat{x} + \mat{B}_{\varphi,\mat{x}}^{} \cdot
\mat{B}_{\varphi,\mat{x}}^t \cdot (\mat{a}_\mat{x} - \mat{x})$.
\end{lemma}
\begin{proof}
For every sample point $\mat{p} \in
B(\mat{x},m\gamma)$, $\norm{\mat{p} - \nu(\mat{x})} \leq \norm{\mat{p} -
\mat{x}} + \norm{\mat{x} - \nu(\mat{x})} = O(m\gamma)$.  By
Lemma~\ref{lem:basic}(i), the distance between $\mat{p}$ and $\nu(\mat{x}) +
T_{\nu(\mat{x})}$ is $O(m^2\gamma^2)$.  As $\mat{a}_\mat{x}$ is convex
combination of all $\mat{p} \in B(\mat{x},m\gamma)$, the distance between
$\mat{a}_\mat{x}$ and $\nu(\mat{x}) + T_{\nu(\mat{x})}$ is also $O(m^2\gamma^2)$.

Let $\hat{\mat{a}}_\mat{x}$ be the projection of $\mat{a}_\mat{x}$ into
$\nu(\mat{x}) + N_{\nu(\mat{x})}$.  The vector $\hat{\mat{a}}_\mat{x} -
\mat{a}_\mat{x}$ is parallel to $T_{\nu(\mat{x})}$, so $\hat{\mat{a}}_\mat{x}$
is also at distance $O(m^2\gamma^2)$ from $\nu(\mat{x}) + T_{\nu(\mat{x})}$.  As
$\hat{\mat{a}}_\mat{x} \in \nu(\mat{x}) + N_{\nu(\mat{x})}$, the
vector $\hat{\mat{a}}_\mat{x} - \nu(\mat{x})$ is orthogonal to
$T_{\nu(\mat{x})}$, which implies that
$\norm{\hat{\mat{a}}_\mat{x} - \nu(\mat{x})} = O(m^2\gamma^2)$.  Therefore, it
suffices to prove that $\norm{\hat{\mat{a}}_\mat{x} - \mat{y}} =
O(m^{7/2}\gamma^2)$ as $\norm{\mat{y} - \nu(\mat{x})} \leq
\norm{\hat{\mat{a}}_\mat{x} - \mat{y}} + \norm{\hat{\mat{a}}_\mat{x} -
\nu(\mat{x})} = \norm{\hat{\mat{a}}_\mat{x} - \mat{y}} + O(m^2\gamma^2)$. 

\begin{figure}
\centerline{\begin{tabular}{ccc}
\includegraphics[scale=0.65]{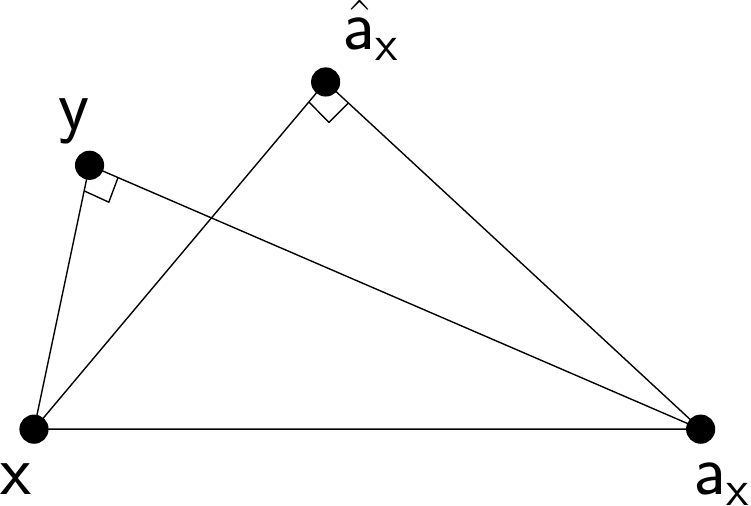} & & 
\includegraphics[scale=0.65]{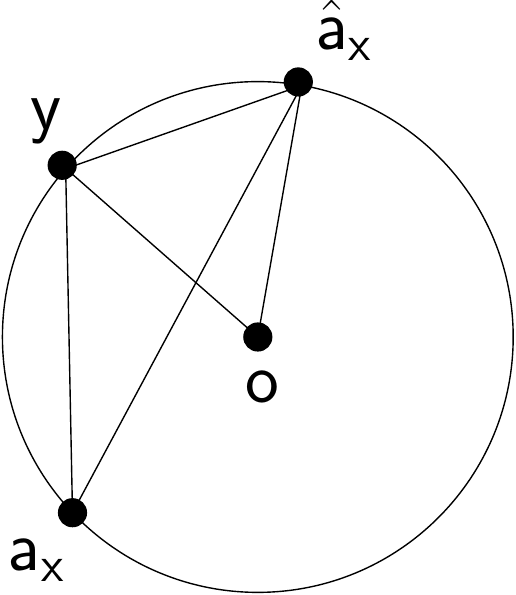} \\
\\
(a) & \hspace*{.5in} & (b) 
\end{tabular}}
\caption{(a)~The points  $\mat{x}$,
$\mat{y}$, $\hat{\mat{a}}_\mat{x}$, and $\mat{a}_\mat{x}$ lie on a
$(d-1)$-dimensional sphere with $\mat{x}\,\mat{a}_\mat{x}$ 
as a diameter. (b)~The circle with center $\mat{o}$ 
circumscribes $\mat{y}\,\hat{\mat{a}}_\mat{x}\,\mat{a}_\mat{x}$.  Also,
$\angle \hat{\mat{a}}_\mat{x}\,\mat{o}\,\mat{y} = 2\angle \hat{\mat{a}}_{\mat{x}}\,\mat{a}_\mat{x}\,\mat{y}$.}
\label{fg:proj-1}
\end{figure}

Refer to Figure~\ref{fg:proj-1}(a).  By construction, $\hat{\mat{a}}_\mat{x}
\in \nu(\mat{x}) + N_{\nu(\mat{x})}$.  Also, $\mat{x} - \nu(\mat{x}) \in
N_{\nu(\mat{x})}$, implying that $\mat{x} \in \nu(\mat{x}) + N_{\nu(\mat{x})}$.
Therefore, $\angle \mat{x}\,\hat{\mat{a}}_\mat{x}\,\mat{a}_\mat{x} = \pi/2$.
From the previous discussion, $\mat{y}$ is the projection of $\mat{a}_\mat{x}$
onto $\mat{x} + L_\mat{x}$.   So $\angle \mat{x}\,\mat{y}\,\mat{a}_\mat{x} =
\pi/2$.   As a result, $\mat{x}$, $\mat{y}$, $\hat{\mat{a}}_\mat{x}$, and
$\mat{a}_\mat{x}$ lie on a $(d-1)$-dimensional sphere $S$ that has
$\mat{x}\,\mat{a}_\mat{x}$ as a diameter.  Since $\mat{a}_\mat{x}$ is a convex
combination of all $\mat{p} \in P \cap B(\mat{x},m\gamma)$, we have
$\norm{\mat{a}_\mat{x} - \mat{x}} \leq m\gamma$.  Thus,
$\mathrm{radius}(S) = O(m\gamma)$.

Since $\angle \mat{x}\,\hat{\mat{a}}_\mat{x}\,\mat{a}_\mat{x} = \pi/2$, we have
$\norm{\hat{\mat{a}}_\mat{x} - \mat{x}}^2 + \norm{\hat{\mat{a}}_\mat{x} -
\mat{a}_\mat{x}}^2 = \norm{\mat{a}_\mat{x} - \mat{x}}^2$.  It follows that
$\norm{\hat{\mat{a}}_\mat{x} - \mat{x}} \geq \norm{\mat{a}_\mat{x} -
\mat{x}}/2$ or $\norm{\hat{\mat{a}}_\mat{x} - \mat{a}_\mat{x}} \geq
\norm{\mat{a}_\mat{x} - \mat{x}}/2$.  We prove that $\angle
\hat{\mat{a}}_\mat{x}\,\mat{x}\,\mat{y} = O(m^{5/2}\gamma)$ if
$\norm{\hat{\mat{a}}_\mat{x}
- \mat{x}} \geq \norm{\mat{a}_\mat{x} - \mat{x}}/2$.  Let
  $\{\mat{v}_1,\ldots,\mat{v}_{d-m}\}$ and $\{\mat{w}_1,\ldots,\mat{w}_{d-m}\}$
be orthonormal bases of $N_{\nu(\mat{x})}$ and $L_{\mat{x}}$, respectively,
that satisfy Lemma~\ref{lem:choice}.  
Note that $\hat{\mat{a}}_\mat{x} - \mat{x} \in N_{\nu(\mat{x})}$
and $\mat{y}-\mat{x} \in L_\mat{x}$.
Refer to Lemma~\ref{lem:proj-angle}.  Let
$(\mat{a}_\mat{x}-\mat{x})/\norm{\mat{a}_\mat{x}-\mat{x}}$ be the unit vector
$\mat{n}$, let $\hat{\mat{a}}_\mat{x} - \mat{x}$ be the vector $\mat{u}_1$, let $\mat{y} - \mat{x}$ be the vector $\mat{u}_2$ as specified in
Lemma~\ref{lem:proj-angle}, and let $\phi = \angle (L_\mat{x},N_{\nu(\mat{x})}) =
O(m\sqrt{m}\,\gamma)$ by Lemma~\ref{lemma::normal_angle}.  We need to show that
the values $\alpha_1$ and $\alpha_2$ defined in Lemma~\ref{lem:proj-angle}
satisfy the assumption that $\alpha_1 > \alpha_2 + (2m^2\phi^2)/\cos\phi$.

By Lemma~\ref{lem:choice}, $\angle (\mat{v}_i,\mat{w}_i) \leq \phi$ for $i \in
[1,d-m]$, which implies that $\norm{\mat{v}_i-\mat{w}_i} \leq 2\sin(\phi/2)
\leq \phi$.  By definition, $\alpha_2 = \sum_{i=d-2m+1}^{d-m}
((\mat{w}_i-\mat{v}_i)^t \mat{n})^2$, and therefore, 
$\alpha_2 \leq \sum_{i=d-2m+1}^{d-m} \norm{\mat{w}_i-\mat{v}_i}^2 \leq m\phi^2 =
O(m^4\gamma^2)$.  By definition, $\alpha_1$ is the squared norm of the projection of
$\mat{n} = (\mat{a}_\mat{x}-\mat{x})/\norm{\mat{a}_\mat{x}-\mat{x}}$ onto
$N_{\nu(\mat{x})}$.  Since $\hat{\mat{a}}_\mat{x} - \mat{x}$ is the
projection of $\mat{a}_\mat{x}-\mat{x}$ onto $N_{\nu(\mat{x})}$, we get 
$\alpha_1 =
\norm{\hat{\mat{a}}_\mat{x}-\mat{x}}^2/\norm{\mat{a}_\mat{x}-\mat{x}}^2 \geq
1/4$ because $\norm{\hat{\mat{a}}_\mat{x} - \mat{x}} \geq
\norm{\mat{a}_\mat{x}-\mat{x}}/2$ by assumption.  This shows that $\alpha_1 >
\alpha_2 + (2m^2\phi^2)/\cos\phi$.
Then, Lemma~\ref{lem:proj-angle} implies that $\angle
\hat{\mat{a}}_\mat{x}\,\mat{x}\,\mat{y} = \angle (\mat{u}_1,\mat{u}_2) \leq
\arccos\left(\sqrt{1-\frac{\alpha_2}{\alpha_1}}\cos\phi -
\frac{2m^2\phi^2}{\sqrt{\alpha_1^2-\alpha_1\alpha_2}}\right)$.  One can verify
that the right hand side is $\arccos(1-O(m^5\gamma^2))$ and so $\angle
\hat{\mat{a}}_\mat{x}\,\mat{x}\,\mat{y} = O(m^{5/2}\gamma)$.  

Similarly, we can
prove that $\angle \hat{\mat{a}}_\mat{x}\,\mat{a}_\mat{x}\,\mat{y}
=O(m^{5/2}\gamma)$ if $\norm{\hat{\mat{a}}_\mat{x} - \mat{a}_\mat{x}} \geq
\norm{\mat{a}_\mat{x} - \mat{x}}/2$.  We conclude that $\angle
\hat{\mat{a}}_\mat{x}\,\mat{x}\,\mat{y} = O(m^{5/2}\gamma)$ or $\angle
\hat{\mat{a}}_\mat{x}\,\mat{a}_\mat{x}\,\mat{y} =O(m^{5/2}\gamma)$.

Without loss of generality, assume that $\angle
\hat{\mat{a}}_\mat{x}\,\mat{a}_\mat{x}\,\mat{y} = O(m^{5/2}\gamma)$.  Consider
the circumcircle of $\hat{\mat{a}}_\mat{x}\,\mat{a}_\mat{x}\,\mat{y}$.  Let
$\mat{o}$ be its center.  Refer to Figure~\ref{fg:proj-1}(b).  The angle
$\angle \hat{\mat{a}}_\mat{x}\,\mat{o}\,\mat{y} = 2\angle
\hat{\mat{a}}_\mat{x}\,\mat{a}_\mat{x}\,\mat{y}$.  Then,
$\norm{\hat{\mat{a}}_\mat{x} - \mat{y}} = 2\norm{\mat{o} -
\mat{y}}\,\sin(\angle \hat{\mat{a}}_\mat{x}\,\mat{o}\,\mat{y}/2) \leq
\mbox{radius}(S) \cdot O(m^{5/2}\gamma) = O(m^{7/2}\gamma^2)$.  
\end{proof}

Next, we prove that $\mat{x}_{i+1}$ is much closer to
$Z_\varphi$ than $\mat{x}_i$.

\begin{lemma}
\label{lem:iterate}
Let $P$ be a uniform $(\eps,\kappa)$-sample of $\mani$. There exists $\eps_0 \in (0,1)$ that decreases as $d$ and $\kappa$ increase such that if $\eps \leq \eps_0$, then for any point $\mat{y}$
at distance $O(m^{7/2}\gamma^2)$ or less from $\mani$, we have $\norm{\mat{y}' -
\tilde{\mat{y}}} \leq \gamma^{1/4} \cdot \norm{\mat{y} - \tilde{\mat{y}}}$, where
$\tilde{\mat{y}}$ is the nearest point in $Z_\varphi \cap \widehat{\mani}$ to
$\mat{y}$ and $\mat{y}' = \mat{y} + \mat{B}_{\varphi,\mat{y}}^{} \cdot
\mat{B}_{\varphi,\mat{y}}^t \cdot ( \mat{a}_{\mat{y}} - \mat{y})$. 
\end{lemma}
\begin{proof}
Let $\mat{z} = \nu(\mat{y})$.  For $i \in [1,d-m]$, let $\mat{v}_i$ be the unit
vector in $N_{\mat{z}}$ such that $\mat{B}_{\varphi,\mat{y}} =
(f_{\mat{v}_1}(\mat{y}), \ldots, f_{\mat{v}_{d-m}}(\mat{y}))$ consists of orthogonal unit column vectors.  By
Lemma~\ref{lemma::normal_angle}, $\angle (L_{\mat{y}},N_{\mat{z}}) =
O(m\sqrt{m}\,\gamma)$, so for any distinct $i,j \in [1,d-m]$, $\angle (\mat{v}_i,
\mat{v}_j) = \pi/2 \pm O(m\sqrt{m}\,\gamma)$.  This allows us to prove as in the
proof of Lemma~\ref{lemma:local} that $\{\mat{v}_1,\ldots,\mat{v}_{d-m}\}$ are
linearly independent and hence they form a basis of $N_{\mat{z}}$.

Let $\varrho_{\mat{z}}$ be the canonical function with respect to $\mat{z}$ and
the basis $\{\mat{v}_1,\ldots,\mat{v}_{d-m}\}$ of $N_\mat{z}$.  
Since $\norm{\mat{y}-\tilde{\mat{y}}}$ is at most
$\norm{\mat{y}-\mat{z}}$ plus the distance from $\mat{z}$ to $Z_\varphi \cap \widehat{\mani}$, 
by Theorem~\ref{thm:main}, we have $\norm{\mat{y}-\tilde{\mat{y}}} \leq
O(m^{7/2}\gamma^2) + O(m^{5/2}\gamma^2) = O(m^{7/2}\gamma^2)$.
So $\norm{\tilde{\mat{y}}-\mat{z}} \leq \norm{\mat{y} -
\tilde{\mat{y}}} + \norm{\mat{y}-\mat{z}} = O(m^{7/2}\gamma^2)$.  Therefore,
\[
\mbox{segment $\mat{y}\,\tilde{\mat{y}}$ is
contained in $B(\mat{z}, tm^{7/2}\gamma^2)$ for some constant $t$},
\] 
implying that
$\varrho_\mat{z}(\mat{x})$ is defined for any point $\mat{x}$ in the segment
$\mat{y}\,\tilde{\mat{y}}$ as long as $\eps_0 < 1/(8tm^{7/2})$ so that $tm^{7/2}\gamma^2 \leq
16tm^{7/2}\eps_0\eps < 2\eps$.  By Lemmas~\ref{lem:agree} and~\ref{lemma:local}, 
$\varrho_\mat{z}^{-1}(0)$ agrees with $Z_\varphi$ within
$B(\mat{z},tm^{7/2}\gamma^2)$.  Then, the following relations follow from
Lemma~\ref{lemma::normal_angle}, Lemma~\ref{lemma::unique},
Theorem~\ref{thm:main}, and the facts that
$\angle (\mat{v}_i,f_{\mat{v}_i}(\mat{y})) = O(m\sqrt{m}\,\gamma)$ for any $i \in [1,d-m]$,
and $\angle (\mat{v}_i,f_{\mat{v}_j}(\mat{y})) = \pi/2 \pm O(m\sqrt{m}\,\gamma)$ for any
distinct $i,j \in [1,d-m]$.
\begin{itemize}

\item For all $i \in [1,d-m]$ and all $\mat{x} \in
B(\mat{z},tm^{7/2}\gamma^2)$, $\norm{\nabla\varrho_{\mat{z},i}(\mat{x})} \in
\left[1-O(\kappa m^4\gamma),1+O(\kappa m^4\gamma)\right]$.

\item For all distinct indices $i,j \in [d-m]$ and for all pair of points
$\mat{x},\mat{x}' \in B(\mat{z}, tm^{7/2}\gamma^2)$,
$\nabla\varrho_{\mat{z},i}(\mat{x})^t \cdot \nabla\varrho_{\mat{z},j}(\mat{x}')
= \pm O(\kappa m^4\gamma)$.

\item For all $i \in [d-m]$, $f_{\mat{v}_i}(\mat{y})^t \cdot
\nabla\varrho_{\mat{z},i}(\mat{y}) \in \left[1 - O(\kappa m^4\gamma),1 +
O(\kappa m^4\gamma)\right]$.

\item For all distinct $i,j \in [d-m]$, $f_{\mat{v}_i}(\mat{y})^t \cdot
\nabla\varrho_{\mat{z},j}(\mat{y}) = \pm O(\kappa m^4\gamma)$.

\end{itemize}

We first prove lower and upper bounds on $\norm{\varrho_\mat{z}(\mat{y})}$.
Since $\tilde{\mat{y}}$ is the nearest point in $Z_\varphi \cap \widehat{\mani}$
to $\mat{y}$, the vector $\mat{y}-\tilde{\mat{y}}$ belongs to the normal space
of $Z_\varphi$ at $\tilde{\mat{y}}$.  Recall that $Z_{\varrho_\mat{z}}$ agrees
with $Z_\varphi$ locally, so the normal space of $Z_\varphi$ at
$\tilde{\mat{y}}$ is spanned by $\{\nabla\varrho_{\mat{z},1}(\tilde{\mat{y}}),
\ldots, \nabla\varrho_{\mat{z},d-m}(\tilde{\mat{y}})\}$.  Let $\mat{u} =
\sum_{i=1}^{d-m} \lambda_i\cdot\nabla\varrho_{\mat{z},i}(\tilde{\mat{y}})$
denote the unit vector $(\mat{y}-\tilde{\mat{y}})/\norm{\mat{y} -
\tilde{\mat{y}}}$.  Standard vector calculus gives 
\begin{eqnarray}
\varrho_{\mat{z}}(\mat{y}) & = &
\left(\bigintsss_{0}^1
\left(\nabla \varrho_{\mat{z},1}(\tilde{\mat{y}}+r\mat{u}), \ldots,
\nabla \varrho_{\mat{z},d-m}(\tilde{\mat{y}}+r\mat{u})\right)^t \cdot
(\mat{y}-\tilde{\mat{y}}) \;\; \mbox{d}r \right) \nonumber \\
& = &
\norm{\mat{y}-\tilde{\mat{y}}} \cdot \bigintsss_{0}^1
\left(\nabla \varrho_{\mat{z},1}(\tilde{\mat{y}}+r\mat{u}), \ldots,
\nabla \varrho_{\mat{z},d-m}(\tilde{\mat{y}}+r\mat{u})\right)^t \cdot
\left(\sum_{i=1}^{d-m}
\lambda_i\cdot\nabla\varrho_{\mat{z},i}(\tilde{\mat{y}})\right)
\; \mbox{d}r \nonumber \\
& = &
\norm{\mat{y}-\tilde{\mat{y}}} \cdot 
\begin{pmatrix}
\lambda_1 + \sum_{i=1}^{d-m} (\pm\lambda_i) \cdot O(\kappa m^4\gamma) \\
\vdots \\
\lambda_{d-m} + \sum_{i=1}^{d-m} (\pm\lambda_i) \cdot O(\kappa m^4\gamma)
\end{pmatrix}.  \label{eq:iterate}
\end{eqnarray}
Hence,
\begin{equation}
\sum_{i=1}^{d-m}\lambda_i^2 - O(\kappa m^4\gamma) \left(\sum_{i=1}^{d-m}|\lambda_i|\right)^2
\leq \frac{\norm{\varrho_{\mat{z}}(\mat{y})}^2}{\norm{\mat{y}-\tilde{\mat{y}}}^2} \leq 
\sum_{i=1}^{d-m}\lambda_i^2 + O(\kappa m^4\gamma) \left(\sum_{i=1}^{d-m}|\lambda_i|\right)^2.
\label{eq:iterate-2}
\end{equation}

We claim that if $\eps_0$ is small enough, then
\begin{equation}
\forall\, i \in [1,d-m], \quad |\lambda_i| \leq 1 + O((d-m)\kappa m^4\gamma).
\label{eq:iterate-3}
\end{equation}
Let $k = \argmax_{i=[1,d-m]} |\lambda_i|$. We take the 
dot product of
$\sum_{i=1}^{d-m} \lambda_{i} \cdot
\nabla\varrho_{\mat{z},i}(\tilde{\mat{y}})$ and 
$\nabla\varrho_{\mat{z},k}(\tilde{\mat{y}})$  or
$-\nabla\varrho_{\mat{z},k}(\tilde{\mat{y}})$ depending on whether
$\lambda_{k}$ is non-negative or negative, respectively.  This dot product
is at most $1 + O(\kappa m^4\gamma)$ as
$\norm{\nabla\varrho_{\mat{z},k}(\tilde{\mat{y}})} = 1 +
O(\kappa m^4\gamma)$.  On the other hand,
for each $i \not= k$, $\lambda_{i} \cdot \nabla\varrho_{\mat{z},i}(\tilde{\mat{y}})^t \cdot 
\nabla\varrho_{\mat{z},k}(\tilde{\mat{y}})$ contributes
$\pm|\lambda_{i}| \cdot O(\kappa m^4\gamma)$.  It follows that
\begin{eqnarray*}
	& & |\lambda_{k}|\left(1 - O(\kappa m^4\gamma)\right) - O(\kappa m^4\gamma) \sum_{i\not=k} |\lambda_{i}| \leq 1
	+ O(\kappa m^4\gamma) \\
	& \Rightarrow & \left(1 - O((d-m)\kappa m^4\gamma))\right)|\lambda_{k}| \leq 1 + O(\kappa m^4\gamma) \\
	& \Rightarrow & |\lambda_{k}| \leq 1 + O((d-m)\kappa m^4\gamma)).
\end{eqnarray*}
Since $|\lambda_k| = \max_i |\lambda_i|$, it establishes our claim.

Since $\sum_{i=1}^{d-m}
\lambda_i\cdot\nabla\varrho_{\mat{z},i}(\tilde{\mat{y}})$ is a unit vector, we
get 
\[
\left\|\sum_{i=1}^{d-m} \lambda_i \cdot \nabla\varrho_{\mat{z},i}(\tilde{\mat{y}}) \right\|^2
= \sum_{i=1}^{d-m}
\lambda_i^2\cdot\norm{\nabla\varrho_{\mat{z},i}(\tilde{\mat{y}})}^2 +
\sum_{i\not=j}\lambda_i\lambda_j\cdot\nabla\varrho_{\mat{z},i}(\mat{y})^t \cdot
\nabla\varrho_{\mat{z},j}(\mat{y}) = 1, 
\]
which implies that 
\[
1 - O(\kappa m^4\gamma) -
O(\kappa m^4\gamma)\sum_{i\not=j}|\lambda_i\lambda_j|
\leq \;\; \sum_{i=1}^{d-m} \lambda_i^2 \;\; \leq 1 + O(\kappa m^4\gamma) +
O(\kappa m^4\gamma)\sum_{i\not=j}|\lambda_i\lambda_j|.
\]

Using the above relations concerning $\lambda_i$'s, we get an upper bound of the right
hand side of \eqref{eq:iterate-2} as follows.
\begin{eqnarray*}
\sum_{i=1}^{d-m} \lambda_i^2 + O(\kappa m^4\gamma) \left(\sum_{i=1}^{d-m}
|\lambda_i|\right)^2 & \leq & 1 + O(\kappa m^4\gamma) +  O(\kappa m^4\gamma)\sum_{i\not=j}|\lambda_i\lambda_j| \\
& \leq & 1 + O(\kappa m^4\gamma) +  O(\kappa m^4\gamma) \cdot (d^2 + O(d^2(d-m)\kappa m^4\gamma)) \\
& \leq & 1 + O(d^2\kappa m^4\gamma).
\end{eqnarray*}
Symmetrically, we get a lower bound of the left hand side of \eqref{eq:iterate-2}: 
\[
\sum_{i=1}^{d-m} \lambda_i^2 - O(\kappa m^4\gamma) \left(\sum_{i=1}^{d-m} |\lambda_i|\right)^2 
\geq 1 - O(d^2\kappa m^4\gamma).
\]
Thus, we simplify \eqref{eq:iterate-2} to
\begin{equation}
(1 - O(d^2\kappa m^4\gamma)) \cdot \norm{\mat{y}-\tilde{\mat{y}}}^2
\leq 
\norm{\varrho_{\mat{z}}(\mat{y})}^2 \leq
(1 + O(d^2\kappa m^4\gamma)) \cdot \norm{\mat{y}-\tilde{\mat{y}}}^2.
\label{eq:iterate-4}
\end{equation}
In other words, $\norm{\varrho_{\mat{z}}(\mat{y})}$ is a good 
approximation of the distance from $\mat{y}$ to the zero-set of $\varrho_{\mat{z}}$.

Next, we give a lower bound on $\cos\angle \mat{y}'\,\mat{y}\,\tilde{\mat{y}}$.
Consider the dot product $(\mat{y}'-\mat{y})^t \cdot (\tilde{\mat{y}}-\mat{y})$.  By
expanding $\mat{B}_{\varphi,\mat{y}}^t \cdot (\mat{a}_\mat{y}-\mat{y})$, we get 
\[
\mat{y}'-\mat{y} = \mat{B}_{\varphi,\mat{y}} \cdot \mat{B}_{\varphi,\mat{y}}^t \cdot (\mat{a}_\mat{y} - \mat{y})
= \mat{B}_{\varphi,\mat{y}} \cdot (-\varrho_{\mat{z}}(\mat{y})).
\] 
Since $\mat{B}_{\varphi,\mat{y}}$ consists of orthogonal unit column vectors, we get 
\begin{equation}
\norm{\mat{y}'-\mat{y}} = \norm{\mat{B}_{\varphi,\mat{y}} \cdot (-\varrho_{\mat{z}}(\mat{y}))}
= \norm{\varrho_{\mat{z}}(\mat{y})}.  
\label{eq:iterate-6}
\end{equation}
Therefore,
\begin{eqnarray}
(\mat{y}'-\mat{y})^t \cdot (\tilde{\mat{y}}-\mat{y})
& = & 
\norm{\varrho_{\mat{z}}(\mat{y})} \cdot
\norm{\mat{y}-\tilde{\mat{y}}} \cdot \cos \angle \mat{y}'\,\mat{y}\,\tilde{\mat{y}} \nonumber \\
& \leq &  
\sqrt{1 + O(d^2\kappa m^4\gamma)} \cdot \norm{\mat{y}-\tilde{\mat{y}}}^2 \cdot
\cos \angle \mat{y}'\,\mat{y}\,\tilde{\mat{y}}. \label{eq:iterate-5}
\end{eqnarray}
Recall that $\sum_{i=1}^{d-m}\lambda_i \cdot \nabla\varrho_{\mat{z},i}(\tilde{\mat{y}})$ is
the unit vector $(\mat{y}-\tilde{\mat{y}})/\norm{\mat{y}-\tilde{\mat{y}}}$. 
By expanding $(\mat{y}'-\mat{y})^t \cdot (\tilde{\mat{y}}-\mat{y})$, we get
\begin{eqnarray*}
(\mat{y}'-\mat{y})^t \cdot (\tilde{\mat{y}}-\mat{y}) & = & 
\left(\mat{B}_{\varphi,\mat{y}} \cdot \varrho_{\mat{z}}(\mat{y})\right)^t \cdot
\norm{\mat{y}-\tilde{\mat{y}}} \cdot \sum_{i=1}^{d-m} \lambda_i \cdot \nabla\varrho_{\mat{z},i}(\tilde{\mat{y}}) \\
& = & 
\left(\sum_{i=1}^{d-m} \varrho_{\mat{z},i}(\mat{y}) \cdot f_{\mat{v}_i}(\mat{y})\right)^t 
\cdot
\norm{\mat{y}-\tilde{\mat{y}}} \cdot \sum_{i=1}^{d-m} \lambda_i \cdot \nabla\varrho_{\mat{z},i}(\tilde{\mat{y}}) \\
& = &
\sum_{i=1}^{d-m} \sum_{j=1}^{d-m} \varrho_{\mat{z},i}(\mat{y}) \cdot \lambda_j
\cdot f_{\mat{v}_i}(\mat{y})^t \cdot 
\nabla\varrho_{\mat{z},j}(\tilde{\mat{y}}) \cdot \norm{\mat{y}-\tilde{\mat{y}}} \\
& = &
\sum_{i=1}^{d-m} \varrho_{\mat{z},i}(\mat{y}) \cdot \norm{\mat{y}-\tilde{\mat{y}}} \cdot \beta_i, 
\end{eqnarray*}
where $\beta_i = \lambda_i + \sum_{i=1}^{d-m} (\pm\lambda_i) \cdot O(\kappa m^4\gamma)$ for $i \in [1,d-m]$.
Note the similarity between the $\beta_i$'s and the vector in \eqref{eq:iterate}.  Therefore, 
$\norm{\mat{y}-\tilde{\mat{y}}} \cdot \beta_i = \varrho_{\mat{z},i}(\mat{y}) +
\norm{\mat{y}-\tilde{\mat{y}}} \cdot \sum_{i=1}^{d-m} (\pm\lambda_i) \cdot O(\kappa m^4\gamma) \geq 
\varrho_{\mat{z},i}(\mat{y}) - O((d-m)\kappa m^4\gamma))\cdot\norm{\mat{y}-\tilde{\mat{y}}}$ as $|\lambda_i| 
\leq 1 + O((d-m)\kappa m^4\gamma)$.  Hence,
\begin{eqnarray*}
(\mat{y}'-\mat{y})^t \cdot (\tilde{\mat{y}}-\mat{y}) & \geq &
\sum_{i=1}^{d-m}
\varrho_{\mat{z},i}(\mat{y})^2 - O((d-m)\kappa m^4\gamma)) \cdot
\norm{\mat{y}-\tilde{\mat{y}}} \cdot 
\sum_{i=1}^{d-m} |\varrho_{\mat{z},i}(\mat{y})| \\
& \geq & \norm{\varrho_{\mat{z}}(\mat{y})}^2 - 
O((d-m)\kappa m^4\gamma) \cdot
\norm{\mat{y}-\tilde{\mat{y}}} \cdot \sqrt{d-m} \cdot \norm{\varrho_{\mat{z}}(\mat{y})} \\
& \geq &
\norm{\varrho_{\mat{z}}(\mat{y})}^2 - O((d-m)^{3/2}\kappa m^4\gamma) \cdot \norm{\mat{y}-\tilde{\mat{y}}} 
\cdot \norm{\varrho_\mat{z}(\mat{y})}.
\end{eqnarray*}
Substituting \eqref{eq:iterate-4} 
into the above, we get 
\[
(\mat{y}'-\mat{y})^t \cdot (\tilde{\mat{y}}-\mat{y}) \geq 
\bigl(1 - O(d^2\kappa m^4\gamma)\bigr) \cdot \norm{\mat{y}-\tilde{\mat{y}}}^2.
\]
Combining \eqref{eq:iterate-5} with the above inequality gives
\[
\cos\angle\mat{y}'\,\mat{y}\,\tilde{\mat{y}}
\geq 1 - O(d^2\kappa m^4\gamma).
\]

Finally, consider triangle $\mat{y}'\mat{y}\,\tilde{\mat{y}}$.  By the cosine law,
we have 
\[
\norm{\mat{y}'-\tilde{\mat{y}}}
= \bigl(\norm{\mat{y}'-\mat{y}}^2 + 
\norm{\mat{y}-\tilde{\mat{y}}}^2 - 
2 \norm{\mat{y}'-\mat{y}}\, \norm{\mat{y}-\tilde{\mat{y}}}\,
\cos\angle\mat{y}'\mat{y}\,\tilde{\mat{y}}\bigr)^{1/2}.
\]
By \eqref{eq:iterate-4} and \eqref{eq:iterate-6}, $\norm{\mat{y}'-\mat{y}}^2 
\leq (1+O(d^2\kappa m^4\gamma)) \cdot \norm{\mat{y} - \tilde{\mat{y}}}^2$.  Therefore,
\begin{eqnarray*}
\norm{\mat{y}' - \tilde{\mat{y}}} & \leq & 
\norm{\mat{y}-\tilde{\mat{y}}} \cdot
\left(2 + O(d^2\kappa m^4\gamma) - 
2\left(1 - O(d^2\kappa m^4\gamma)\right)\left(1 - O(d^2\kappa m^4\gamma)\right)\right)^{1/2} \\
& \leq & O(dm^2\sqrt{\kappa\gamma}) \cdot \norm{\mat{y}-\tilde{\mat{y}}} \\
& \leq & \gamma^{1/4} \cdot \norm{\mat{y}-\tilde{\mat{y}}}
\end{eqnarray*}
whenever $\eps_0$ is small enough so that $\gamma^{1/4} = O(\eps^{1/4}) = O(\eps_0^{1/4})$ cancels the $O(dm^2\sqrt{\kappa})$ factor.
This requires $\eps_0$ to decrease as $d$ and $\kappa$ increase.  
\end{proof}

By combining Lemmas~\ref{lem:first_itr} and~\ref{lem:iterate}, we prove that
the projection operator will bring an initial point to a point in $Z_\varphi
\cap \widehat{\mani}$ in the limit.

\begin{theorem}
Let $\varphi$ be the function for a uniform $(\eps,\kappa)$-sample of an $m$-dimensional
compact smooth manifold $\mani$ in $\real^d$ as specified in
Theorem~\ref{thm:main}.   
Define
the projection operator $\mat{x}_{i+1} = \mat{x}_i +
\mat{B}_{\varphi,\mat{x}_i} \cdot \mat{B}_{\varphi,\mat{x}_i}^t \cdot
(\mat{a}_{\mat{x}_i} - \mat{x}_i)$, where $\mat{a}_{\mat{x}_i} = \sum_{\mat{p}
\in P} \omega(\mat{x}_i,\mat{p}) \cdot \mat{p}$. 
There exists $\eps_0 \in (0,1)$ that decreases as $d$ and $\kappa$ increase such that if
$\eps \leq \eps_0$, then for any initial point $\mat{x}_0$ at distance
$m\gamma$ or less from some sample point, where $\gamma$ is the input
neighborhood radius, the following properties hold.
\begin{itemize}

\item  $\lim_{i \rightarrow \infty} x_i \in Z_\varphi \cap \widehat{\mani}$,
where $\widehat{\mani}$ is the set of points within a distance of $\eps$ from
$\mani$.

\item For all $i > 0$, $\norm{\mat{x}_i - \nu(\mat{x}_0)} = O(m^{7/2}\gamma^2) = O(m^{7/2}\eps^2)$.

\end{itemize}
\end{theorem}
\begin{proof}
For any point $\mat{x}$, let $\tilde{\mat{x}}$ denote the nearest point in $Z_\varphi
\cap \widehat{\mani}$ to $\mat{x}$.
By Lemma~\ref{lem:first_itr}, $\norm{\mat{x}_1 - \nu(\mat{x}_0)} = O(m^{7/2}\gamma^2)$.
Let $\mat{b}$ be the nearest point in $Z_\varphi \cap \widehat{\mani}$ to $\nu(\mat{x}_0)$.  Since $\norm{\mat{b} - \nu(\mat{x}_0)} = O(m^{5/2}\gamma^2)$ by Theorem~\ref{thm:main}, triangle inequality implies
that for a small enough $\eps_0$,
\begin{align*}
\norm{\mat{x}_1-\tilde{\mat{x}}_1} \leq \norm{\mat{x}_1 - \mat{b}} & \leq  \norm{\mat{b} - \nu(\mat{x}_0)} + \norm{\mat{x}_1 - \nu(\mat{x}_0)} \\
& \leq
O(m^{5/2}\gamma^2) + O(m^{7/2}\gamma^2) \\
& = O(m^{7/2}\gamma^2).
\end{align*}
Since $\norm{\mat{x}_1 - \nu(\mat{x}_0)} = O(m^{7/2}\gamma^2)$,
Lemma~\ref{lem:iterate} is applicable to $\mat{x}_1$.  It ensures that
$\norm{\mat{x}_2 - \tilde{\mat{x}}_2} \leq 
\norm{\mat{x}_2-\tilde{\mat{x}}_1} \leq \gamma^{1/4} \cdot 
\norm{\mat{x}_1 - \tilde{\mat{x}}_1} = O(m^{7/2}\gamma^{9/4})$, which is smaller than 
$O(m^{7/2}\gamma^2)$ and so Lemma~\ref{lem:iterate}
is applicable to $\mat{x}_2$.
Repeating this argument gives 
\[
\norm{\mat{x}_i - \tilde{\mat{x}}_i}
\leq \norm{\mat{x}_i - \tilde{\mat{x}}_{i-1}} = O(m^{7/2}\gamma^{(7+i)/4}).
\]
This proves that $\lim_{i \rightarrow \infty} x_i \in Z_\varphi \cap \widehat{\mani}$.
By triangle inequality, 
\begin{align*}
\norm{\mat{x}_i - \mat{x}_{i-1}} 
& \leq \norm{\mat{x}_i - \tilde{\mat{x}}_{i-1}} + \norm{\mat{x}_{i-1} - \tilde{\mat{x}}_{i-1}} \\
& = O(m^{7/2}\gamma^{(7+i)/4}) + O(m^{7/2}\gamma^{(6+i)/4}) \\
& = O(m^{7/2}\gamma^{(7+i)/4}).
\end{align*}
Therefore, for a small enough $\eps_0$,
\begin{align*}
\norm{\mat{x}_i - \nu(\mat{x}_0)}
& \leq 
\sum_{j=2}^i \norm{\mat{x}_j-\mat{x}_{j-1}} + \norm{\mat{x}_1-\nu(\mat{x}_0)} \\
& < \sum_{j=2}^i O(m^{7/2}\gamma^{(7+j)/4}) + O(m^{7/2}\gamma^2) \\
& = O(m^{7/2}\gamma^2).
\end{align*}
\end{proof}

\section{Conclusion}

We define a function $\varphi$ from a uniform $(\eps,\kappa)$-sample of a
compact smooth manifold $\mani$ in $\real^d$ such that the zero-set of
$\varphi$ near $\mani$ is a faithful reconstruction of $\mani$.  Moreover, we
give a projection operator that will yield a point on the zero-set near $\mani$
in the limit by iterative applications.  More work is needed to improve the
angular error of $O(m^2\sqrt{\kappa\eps})$, which is weaker than the $O(\eps)$
angular error offered by provably good simplicial reconstructions.  It would also be
desirable for $\eps$ to depend on $m$ only instead of $d$.  Another 
natural question is how to deal with non-smooth manifolds and non-manifolds.

\paragraph{Acknowledgment}
The authors would like to thank the anonymous reviewers for helpful comments, pointing out mistakes in an earlier version that we subsequently corrected, and suggesting the removal of some slack in the bounds on Hausdorff distance and angular error.


\begin{thebibliography}{99}

\bibitem{AA}
A. Adamson and M. Alexa.
Point-sampled cell complexes.
\emph{ACM Transactions on Graphics}, 25 (2006), 671--680.

\bibitem{alexa}
M. Alexa, J. Behr, D. Cohen-OR, S. Fleishman, D. Levin, and C.T. Silva.
Computing and rendering point set surfaces.
\emph{IEEE Transactions on Visualization and Computer Graphics},
9 (2003), 3--15.

\bibitem{aamari}
E. Aamari, J. Kim, F. Chazal, B. Michel, A. Rinaldo, and L. Wasserman.
Estimating the reach of a manifold.
arXiv:1705.04565 [math.ST], 2017.

\bibitem{BSW}
M. Belkin, J. Sun, and Y. Wang.
Constructing Laplace operator from point clouds in $\real^d$.
\emph{Proceedings of the ACM-SIAM Annual Symposium on Discrete
Algorithms}, 2009, 1031--1040.

\bibitem{bertalmio}
M. Bertalm\'{i}o, L.T. Cheng, S. Osher, and G. Sapiro.
Variational problems and partial differential equations on implicit surfaces.
\emph{Journal of Computational Physics}, 174 (2001), 759--780.

\bibitem{bjorck}
A. Bjorck and G. Golub.
Numerical methods for computing angles between linear subspaces.
\emph{Mathematics of Computation}, 27 (1973), 579--594.

\bibitem{bf}
J.-D. Boissonnat and F. Cazals.
Smooth surface reconstruction via natural neighbor interpolation
of distance functions. \emph{Computational Geometry: Theory
and Applications}, 22 (2002), 185--203.

\bibitem{bg}
J.-D. Boissonnat and A. Ghosh.
Manifold reconstruction using tangential Delaunay complexes.
\emph{Proceedings of the 26th Annual Symposium on Computational
Geometry}, 2010, 324--333.

\bibitem{bgo}
J.-D. Boissonnat, L.J. Guibas and S.Y. Oudot.  Manifold reconstruction
in arbitrary dimensions using witness complexes.
\emph{Discrete and Computational Geometry}, 42 (2009), 37--70.


\bibitem{carr}
J.C. Carr, R.K. Beatson, J.B. Cherrie, T.J. Mitchell, W.R. Fright, B.C. McCallum,
and T.R. Evans.
Reconstruction and representation of 3D objects with radial basis functions.
\emph{Proceedings of SIGGRAPH}, 2001, 67--76.

\bibitem{CC16} S.-W. Cheng and M.-K. Chiu.  Tangent estimation from point
samples. \emph{Discrete and Computational Geometry}, 56 (2016), 505--557.

\bibitem{CC09}
S.-W. Cheng and M.-K. Chiu.  Dimension detection via slivers.
\emph{Proceedings of the 20th Annual ACM-SIAM Symposium on Discrete Algorithms},
2009, 1001--1010.

\bibitem{CDR05}
S.-W. Cheng, T.K. Dey and E.A. Ramos.  Manifold reconstruction
from point samples.  \emph{Proceedings of the 16th Annual ACM--SIAM
Symposium on Discrete Algorithms}, 2005, 1018--1027.

\bibitem{CWW}
S.-W. Cheng, Y. Wang, and Z. Wu.
Provable dimension detection using principal component analysis.
\emph{International Journal of Computational Geometry and Applications}, 
18 (2008), 414--440.

\bibitem{DDW}
T.K. Dey, Z. Dong, and Y. Wang.
Parameter-free topology inference and sparsification for data on manifolds.
\emph{Proceedings of the 28th Annual ACM-SIAM Symposium on Discrete Algorithms},
2017, 2733--2747.

\bibitem{DGGZ}
T. K. Dey, J. Giesen, S. Goswami and W. Zhao. Shape
dimension and approximation from samples. 
\emph{Discrete and Computational Geometry}, 29 (2003), 419--434.

\bibitem{DS} 
T.K. Dey and J. Sun.  An adaptive MLS surface for reconstruction
with guarantees.  \emph{Proceedings of the Eurographics Symposium on 
Geometry Processing}, 2005, 43--52.

\bibitem{dorsey}
J. Dorsey and P. Hanrahan.
Digital materials and virtual weathering.
\emph{Scientific American}, 282 (2000), 282--289.

\bibitem{EI94} 
S.C.~Eisenstat and I.C.F.~Ipsen.  Relative perturbation bounds
for eigenspaces and singular vector subspaces.  \emph{Proceedings of the 5th
SIAM Conference on Applied Linear Algebra}, 1994, 62--66.

\bibitem{gc}
D.G. Feingold, R.S. Varga, 
Block diagonally dominant matrices and generalizations of the Gershgorin
circle theorem. 
\emph{Pacific Journal of Mathematics}, 12 (1962), 1241--1250.

\bibitem{galantai}
A. Gal\'antai and Cs.J. Heged\~us.
Jordan's principal angles in complex vector spaces.
\emph{Numerical Linear Algebra with Applications}, 13 (2006), 589--598.

\bibitem{greer}
J.B. Greer. 
An improvement of a recent Eulerian method for solving PDEs on general geometries. 
\emph{Journal of Scientific Computing}, 29 (2006), 321--352.

\bibitem{GW}
J. Giesen and U. Wagner.  
Shape dimension and intrinsic metric from samples of manifolds with high
codimension.
{\em Discrete and Computational Geometry}, 32 (2004), 245--267.

\bibitem{golub}
G.H. Golub and C.F. van Loan.  \emph{Matrix Computations}, Johns 
Hopkins University Press, 1996.

\bibitem{HA}
M. Hein and J.-Y. Audibert.
Intrinsic dimensionality estimation of submanifolds in Euclidean space.
\emph{Proceedings of the 22nd International Conference on Machine Learning},
2005, 289--296.

\bibitem{hoppe}
H. Hoppe, T. DeRose, T. Duchamp, J. McDonald, and W. Stuetzle.
Surface reconstruction from unorganized points.
\emph{Proceedings of SIGGRAPH}, 1992, 71--78.

\bibitem{horn}
R. Horn and C.R. Johnson.
\emph{Topics in Matrix Analysis},
Cambridge University Press, 1991.


\bibitem{Kol}
R. Kolluri.
Provably good moving least squares.
\emph{ACM Transactions on Algorithms}, 2 (2008), article no.~18.

\bibitem{levin}
D. Levin.
Mesh-independent surface interpolation.
In \emph{Geometric Modeling for Scientific Visualization},
G. Brunett, B. Hamann, K. Mueller, and L. Linsen, eds., 
Springer-Verlag, 2003.

\bibitem{LB}
E. Levina and P.J. Bickel.
Maximum likelihood estimation of intrinsic dimension.
\emph{Advances in Neural Information Processing Systems}, 17 (2005),
777---784.  

\bibitem{liang}
J. Liang and H. Zhao.
Solving partial differential equations on point clouds.
\emph{SIAM Journal on Scientific Computing}, 35 (2013), 1461--1486.

\bibitem{LMR12}
A.V.~Little, M.~Maggioni, and L.~Rosasco.
Multiscale geometric methods for data sets I: multiscale SVD, noise and curvature.
\emph{Computer Science and Artificial Intelligence Laboratory Technical Report},
MIT-CSAIL-TR-2012-029, CBCL-310, September 8, 2012.

\bibitem{memoli}
F. M\'{e}moli, G. Sapiro and P. Thompson.
Implicit brain imaging.
\emph{NeuroImage}, 23 (2004), 179--188.

\bibitem{mendelson}
B. Mendelson.  \emph{Introduction to Topology}, Dover Puhblications Inc., New York, third edition, 1990.

\bibitem{miao}
J. Miao and A. Ben-Israel.
On principal angles between subspaces in $\real^n$.
\emph{Linear Algebra and its Applications}, 171 (1992), 81--98.
	
\bibitem{myers}
T.G. Myers and J.P.F. Charpin.
A mathematical model for atmospheric ice accretion and water flow on
a cold surface.
\emph{International Journal of Heat and Mass Transfer}, 47 (2004), 5483--5500.

\bibitem{myersb}
T.G. Myers, J.P.F. Charpin, and S.J. Chapman.
The flow and solidification of a thin fluid film on an arbitrary
three-dimensional surface.
\emph{Physics of Fluids}, 14 (2002), 2788--2803.




\bibitem{ruuth}
S.J. Ruuth and B. Merriman.
A simple embedding method for solving partial differential equations
on surfaces.
\emph{Journal of Computational Physics}, 227 (2008), 1843--196`.

\bibitem{SSM98}
B. Sch\"{o}lkopf, A. Smola, and K.-R. M\"{u}ller. 
Nonlinear component analysis as a kernel eigenvalue problem. 
\emph{Neural Computation}, 10 (1998), 1299--1319.


\bibitem{TSL00}
J.B. Tenenbaum, V. de Silva and J.C. Langford.
A global geometric framework for nonlinear dimensionality reduction.
\emph{Science}, 290 (2000), 2319--2323.

\bibitem{turk}
G. Turk.  
Generating textures on arbitrary surfaces using reaction-diffusion.
\emph{Proceedings of SIGGRAPH}, 1991, 289--298.

\bibitem{wendland}
H. Wendland.
Piecewise polynomial, positive definite and compactly supported
radial functions of minimal degree.
\emph{Advances in Computational Mathematics}, 4 (1995), 389--396.

\bibitem{witkin}
A. Witkin and M. Kass.
Reaction-diffusion textures, 
\emph{Proceedings of SIGGRAPH}, 1991, 299--308.
25 (1991), 299--308.



\end{thebibliography}
\end{document}